\pdfoutput=1
\documentclass[a4paper,12pt]{book}
\newif\ifrozprawa
\rozprawatrue
\usepackage{hyperref}

\newcommand{\newpageateven}{%
	\ifodd\thepage%
	\else%
	\newpage
	\phantom{ }%
	\fi%
}

\usepackage{amsopn}
\usepackage{amsmath}
\usepackage{amssymb}
\usepackage{amsthm}
\usepackage{hyperref}
\usepackage[noadjust]{cite}

\usepackage{subfig}
\usepackage{bm,bbm}
\usepackage{dsfont}
\usepackage{subfig}
\usepackage{bbm}
\usepackage{mathtools}
\usepackage{graphicx}
\usepackage[inline]{enumitem}
\usepackage{textcomp}
\usepackage[polish,english]{babel}
\usepackage[utf8]{inputenc}
\usepackage[T1]{fontenc}
\usepackage{physics}
\usepackage{comment}
\usepackage{xspace}
\usepackage{listings}
\usepackage{jlcode}
\usepackage{courier}
\usepackage{siunitx}

\usepackage{tikz}
\usetikzlibrary{calc}
\usetikzlibrary{shapes.geometric, arrows,matrix}
\usetikzlibrary{decorations.markings}
\usetikzlibrary{calc,trees,positioning,chains,%
	decorations.pathreplacing,decorations.pathmorphing,shapes,shapes.symbols}
\tikzset{nodeStyle/.style = {circle,draw,minimum size=30pt}}
\tikzset{arrowStyle/.style = {-latex}}

\newcommand{\ie}{\emph{i.e.\/}\xspace}

\newcommand{\R}{\ensuremath{\mathbb{R}}}
\newcommand{\C}{\ensuremath{\mathbb{C}}}
\newcommand{\N}{\mathbb{N}}

\newtheorem{theorem}{Theorem}[chapter]

\newtheorem{proposition}[theorem]{Proposition}
\newtheorem{lemma}[theorem]{Lemma}

\theoremstyle{definition}
\newtheorem{definition}[theorem]{Definition}
\usepackage[utf8]{inputenc}

\usepackage{etoolbox}
\patchcmd{\chapter}{plain}{empty}{}{}

\usepackage{hyperref}
\hypersetup{colorlinks=false,hidelinks}

\usepackage{hypcap}

\newcommand{\Id}{\mathrm I}

\newcommand\randg[1][]{
\ifstrempty{#1}{
\mathcal G}{\mathcal G^{\rm #1}}}
\newcommand\randgn[1][]{
	\ifstrempty{#1}{
		{\mathcal G}_n}{{\mathcal G}_n^{\rm #1}}}
\newcommand\randdgn[1][]{
	\ifstrempty{#1}{
		{\vec{\mathcal G}}_n}{{\vec{\mathcal G}}_n^{\rm #1}}}

\newcommand{\PP}{\mathbb P}
\newcommand{\EE}{\mathbb E}

\newcommand{\ER}{Erd\H{o}s-R\'enyi\xspace}

\newcommand{\BA}{Barab\'asi-Albert\xspace}
\newcommand{\CL}{Chung-Lu\xspace}

\newcommand{\indeg}{\operatorname{indeg}}
\newcommand{\outdeg}{\operatorname{outdeg}}

\newcommand{\ZZ}{\mathbb Z}
\newcommand{\RR}{\mathbb R}
\newcommand{\CC}{\mathbb C}

\newcommand{\ii}{\mathrm i}
\newcommand{\ee}{\mathrm e}
\newcommand{\kron}{\otimes}

\newcommand{\vecc}[1]{|#1\rangle\rangle}

\newcommand{\depref}[1]{{\ifrozprawa\ref{#1}\else\ref*{#1}\fi}}

\begin{document}

% !TeX spellcheck = en_GB
\frontmatter
	
	\begin{titlepage}
\begin{center}
\includegraphics[width=0.4\textwidth]{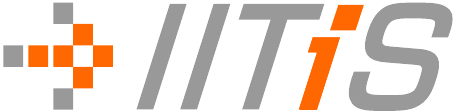}\\
\vspace{0.5em}
\textsc{\large Instytut Informatyki Teoretycznej i Stosowanej\\
Polskiej Akademii Nauk}
\vspace*{1in}
\hrule
\vspace*{0.5em}
\textsc{\Large Wykorzystanie teorii grafów w informatyce kwantowej}
\vspace*{0.5em}
\hrule
\vspace*{1em}
\textsc{\large Rozprawa doktorska}
\par
\vspace{1.5in}
{\large mgr Adam \textsc{Glos}}\\
Promotor: dr hab. Jarosław Adam \textsc{Miszczak}
\vfill
Gliwice, Luty 2021
\end{center}
\end{titlepage}

	\begin{titlepage}
\begin{center}
\includegraphics[width=0.4\textwidth]{iitis_logo}\\
\vspace{0.5em}
\textsc{\large Institute of Theoretical and Applied Informatics, Polish 
Academy of Sciences}
\vspace*{1in}
\hrule
\vspace*{0.5em}
\textsc{\huge  Application of graph theory in quantum computer science}
\vspace*{0.5em}
\hrule
\vspace*{1em}
\textsc{\large Doctoral dissertation}
\par
\vspace{1.5in}
{\large mgr Adam \textsc{Glos}}\\
Supervisor: dr hab. Jarosław Adam \textsc{Miszczak}
\vfill
Gliwice, February 2021
\end{center}
\end{titlepage}

	\setcounter{page}{4}
	
	\cleardoublepage

	\chapter*{\vspace{-2cm}Dedication}

The arrogance of a young scientist who claims that everything he owes is only due to his hard work is truly remarkable. Looking back these few years, at the beginnings of my scientific work, I can see how much of hard work, trust of other people, but also random meetings have led me to the place where I~am. Due to the multitude of people who helped me reach the first giant leap for me (small for mankind, though), I am fully aware that I will not be able to sufficiently thank everyone for their help.

However, there are people who deserve special thanks. My greatest thanks go to my family for supporting me in this path. Special thanks go to my wife, whose support without doubt strengthened me in following my scientific career.

I would also like to thank Prof. Jarosław Miszczak, not only for being my supervisor, but also (or perhaps particularly) for teaching me what science should look like and that it is something more than just optimizing quality measures required for various grant calls. I would like to thank Foundation for Polish Science for convincing me that my supervisor was right.

Special thanks go to all who helped me keeping open-minded. In particular, I would like to thank Piotr Gawron, Mateusz Ostaszewski, Łukasz Pawela, Przemysław Sadowski, Alexander Rivosh, Abuzer Yakaryilmaz and others for discussions on many (not necessarily quantum) topics.

I would like to thank the National Science Center, for granting me scholarship Etiuda, which enabled me to revisit Centre for Quantum Computer Science from the University of Latvia. This dissertation was also prepared under the scholarship Etiuda, no. 2019/32/T/ST6/00158.

Finally, I would like to thank Ryszard Kukulski for many discussions on the convergence of random variables in context of random graphs, and Bruno Coutinho for discussion on efficiency of the classical search. I would also like to Izabela Miszczak for reviewing dissertation and Aleksandra Krawiec for reviewing this dedication.

\cleardoublepage \setcounter{tocdepth}{2} \tableofcontents

	%%%%%%%%%%%%%%%%%%%%%%%%%%%%%%%%%%%%%%%%%%%%%%%%%%%%%%%%%%%%%%%%%%%%%%%%%%%%%%%%
	\chapter*{\vspace{-2cm}List of
		publications}\addcontentsline{toc}{chapter}{List of publications}
	\markboth{\MakeUppercase{List of publications}}{\MakeUppercase{List of
			publications}}
	%%%%%%%%%%%%%%%%%%%%%%%%%%%%%%%%%%%%%%%%%%%%%%%%%%%%%%%%%%%%%%%%%%%%%%%%%%%%%%%%
	% !TeX spellcheck = en_GB

Publications and preprints relevant to the dissertation are highlighted with \textbf{bold}.

\section*{Published work}

\begin{enumerate}[leftmargin=.5cm]
	\itemsep0em
\item A.~Glos, A.~Krawiec, and {\L}.~Pawela, ``Asymptotic entropy of the Gibbs state
of complex networks,'' {\em Scientific Reports}, vol.~11, p.~311, 2021.

\item A.~Glos, N.~Nahimovs, K.~Balakirev, and K.~Khadiev, ``Upperbounds on the
probability of finding marked connected components using quantum walks,''
{\em Quantum Information Processing}, vol.~20, p.~6, 2021.
	
\item Z.~{Tabi}, K.~H. {El-Safty}, Z.~{Kallus}, P.~{Hága}, T.~{Kozsik}, A.~{Glos},
and Z.~{Zimborás}, ``Quantum optimization for the graph coloring problem
with space-efficient embedding,'' in {\em 2020 IEEE International Conference
	on Quantum Computing and Engineering (QCE)}, pp.~56--62, 2020.

\item
A.~Glos, ``Spectral similarity for {B}arabási–{A}lbert and {C}hung–{L}u
models,'' {\em Physica A: Statistical Mechanics and its Applications},
vol.~516, pp.~571--578, 2019.

\item
A.~Glos and J.~A. Miszczak, ``The role of quantum correlations in {C}op and
{R}obber game,'' {\em Quantum Studies: Mathematics and Foundations}, vol.~6,
no.~1, pp.~15--26, 2019.

\item
A.~Glos and J.~A. Miszczak, ``Impact of the malicious input data modification
on the efficiency of quantum spatial search,'' {\em Quantum Information
	Processing}, vol.~18, p.~343, 2019.

\item
A.~Glos and T.~Januszek, ``Impact of global and local interaction on quantum
spatial search on Chimera graph,'' {\em International Journal of Quantum
	Information}, p.~1950040, 2019.

\item
A.~Glos, J.~Miszczak, and M.~Ostaszewski, ``\textbf{{QSW}alk.jl: {J}ulia package for
quantum stochastic walks analysis},'' {\em Computer Physics Communications},
2018.

\item
K.~Domino, A.~Glos, M.~Ostaszewski, P.~Sadowski, and Ł.~Pawela, ``\textbf{Properties of
{Q}uantum {S}tochastic {W}alks from the asymptotic scaling exponent},'' {\em
	Quantum Information and Computation}, vol.~18, no.~3\&4, pp.~0181--0199,
2018.

\item
A.~Glos, A.~Krawiec, R.~Kukulski, and Z.~Pucha{\l}a, ``\textbf{Vertices cannot be
hidden from quantum spatial search for almost all random graphs},'' {\em
	Quantum Information Processing}, vol.~17, no.~4, p.~81, 2018.

\item
A.~Glos and T.~Wong, ``\textbf{Optimal quantum-walk search on {K}ronecker graphs with
dominant or fixed regular initiators},'' {\em Physical Review A}, vol.~98,
no.~6, p.~062334, 2018.

\item
K.~Domino, A.~Glos, and M.~Ostaszewski, ``\textbf{Superdiffusive {Q}uantum {S}tochastic
{W}alk definable on arbitrary directed graph},'' {\em Quantum Information \&
	Computation}, vol.~17, no.~11-12, pp.~973-986, 2017.

\item
A.~Glos, J.~A. Miszczak, and M.~Ostaszewski, ``\textbf{Limiting properties of
stochastic quantum walks on directed graphs},'' {\em Journal of Physics A:
	Mathematical and Theoretical}, vol.~51, no.~3, p.~035304, 2017.
	
\item
A.~Glos and P.~Sadowski, ``Constructive quantum scaling of unitary matrices,''
{\em Quantum Information Processing}, vol.~15, no.~12, pp.~5145--5154, 2016.

\item
D.~Kurzyk and A.~Glos, ``Quantum inferring acausal structures and the {M}onty
{H}all problem,'' {\em Quantum Information Processing}, vol.~15, no.~12,
pp.~4927--4937, 2016.

\end{enumerate}

\section*{Preprints}
\begin{enumerate}[leftmargin=.5cm]
\item R.~Kukulski and A.~Glos, ``\textbf{Comment to `{S}patial search by quantum walk is
	optimal for almost all graphs'},'' {\em arXiv:2009.13309}, 2020.

\item
A.~Glos, A.~Krawiec, and Z.~Zimbor{\'a}s, ``Space-efficient binary optimization
for variational computing,'' {\em arXiv:2009.07309}, 2020.

\item K.~Domino and A.~Glos, ``Hiding higher order cross-correlations of multivariate
data using {A}rchimedean copulas,'' {\em arXiv:1803.07813}, 2018.
\end{enumerate}

	\begin{otherlanguage}{polish}
		%%%%%%%%%%%%%%%%%%%%%%%%%%%%%%%%%%%%%%%%%%%%%%%%%%%%%%%%%%%%%%%%%%%%%%%%%%%%%%%
		\chapter*{\vspace{-4cm}Streszczenie w języku 
			polskim}\addcontentsline{toc}{chapter}{Streszczenie w języku polskim}
		\markboth{\MakeUppercase{Streszczenie w języku 
				polskim}}{\MakeUppercase{Streszczenie w języku polskim}}
		%%%%%%%%%%%%%%%%%%%%%%%%%%%%%%%%%%%%%%%%%%%%%%%%%%%%%%%%%%%%%%%%%%%%%%%%%%%%%%%
		W ramach rozprawy wykazałem, że ciągłe w czasie błądzenie kwantowe pozostaje skuteczne dla ogólnych struktur grafowych. Przeanalizowałem dwa aspekty tego problemu.

Po pierwsze, wiadomym jest, że model \emph{Continuous-Time Quantum Walk} (CTQW), zaproponowany przez Childsa i Goldstone'a, potrafi szybko propagować na grafie będącym nieskończoną ścieżką. Jednak równanie Schr\"odingera wymaga, aby Hamiltonian był symetryczny, przez co mogą być zaimplementowane jedynie nieskierowane grafy. W ramach tej rozprawy przeanalizowałem, czy możliwe jest zaprojektowanie ciągłego w czasie błądzenia kwantowego dla ogólnego grafu skierowanego, tak aby zachowywał on szybką propagację.

Po drugie, przeszukiwanie grafów zdefiniowane przez CTQW jest efektywne dla wielu różnych rodzajów grafów. Jednakże większość z tych grafów miała bardzo prostą strukturę. Najbardziej zaawansowanymi przypadkami były model grafów losowych Erd\H{o}sa-R\'eyniego, który choć najpopularniejszy  nie daje grafów opisujących rzeczywiste interakcje spotykane w przyrodzie, oraz model grafów Barab\'asiego-Alberta, dla których kwadratowe przyspieszenie nie było udowodnione. W ramach tego aspektu przeanalizowałem, czy przyspieszenie kwantowe jest możliwe także dla skomplikowanych struktur grafowych.

Rozprawa składa się z siedmiu rozdziałów. W rozdziale~\ref{sec:intro} umieszczony został wstęp oraz motywacja podjęcia tematu. W rozdziale~\ref{sec:preliminaries} wprowadziłem notację oraz podstawowe koncepcje użyte w rozprawie.

W rozdziałach~\ref{sec:nonmoralizing-qsw} oraz \ref{sec:convergence-qsw} przeanalizowałem pierwszy wprowadzony problem. W rozdziale~\ref{sec:nonmoralizing-qsw} zaproponowałem błądzenie niemoralizujące kwantowo-stochastyczne o globalnych interakcjach, które jest dobrze zdefiniowane dla grafów skierowanych. Wykazałem, że dla tego modelu obserwujemy szybką propagację dla nieskończonej ścieżki. Aby uzyskać ten efekt, istotnie lepszy niż dla klasycznego błądzenia, wprowadziłem mały transfer amplitudy w kierunku niezgodnym z kierunkiem istniejących łuków grafu. W rozdziale~\ref{sec:convergence-qsw} przeanalizowałem własności graniczne wprowadzonego modelu. Zbadałem również dwa inne błądzenia zwane odpowiednio lokalnymi i globalnymi kwantowo-stochastycznymi. Pokazałem, że każdy z wprowadzonych do tej pory modeli miał inne właściwości. W szczególności, w przypadku błądzeń lokalnego i niemoralizującego globalnego wskazałem najbardziej intuicyjne zachowanie dla grafów skierowanych. Badania pokazują, że możliwe jest zaproponowanie szybkiego, ciągłego w czasie błądzenia kwantowego, które jest dobrze zdefiniowane dla ogólnego grafu skierowanego.

W rozdziałach~\ref{sec:hiding} oraz~\ref{sec:complex} przeanalizowałem drugi z postawionych problemów badawczych. W rozdziale~\ref{sec:hiding} poprawiłem i wzmocniłem obecnie wiodące wyniki dotyczące grafów Erd\H{o}sa-R\'enyiego. Wykazałem, że przyspieszenie kwantowe jest poprawne dla wszystkich wierzchołków, nie tylko dla ,,prawie wszystkich''. W porównaniu z obecnie wiodącymi wynikami pokazałem, że Laplasjan jest o wiele prostszym operatorem w analizie niż macierz sąsiedztwa. W rozdziale~\ref{sec:complex} porównałem trzy różne operatory możliwe do wykorzystania w ramach kwantowego przeszukiwania przestrzennego. Pokazałem, że znormalizowany Laplasjan, przy pewnych założeniach, umożliwia osiągnięcie pełnego, kwadratowego przyspieszenia. Przeanalizowałem dwa modele grafów losowych, które zwracają grafy o skomplikowanej strukturze z wysokim prawdopodobieństwem. Analiza potwierdziła, że zaproponowana macierz jest lepsza niż te dotychczas używane. Ostatecznie, zaproponowałem procedurę, która powala rozwiązać problem znajdowania optymalnego czasu pomiaru dla kwantowego przeszukiwania.

W rozdziale~\ref{sec:conclusions} podsumowałem uzyskane wyniki. Rozprawa zawiera również dwa dodatki, gdzie umieszczone zostały dowody wyników użytych w~rozprawie.
	\end{otherlanguage}
	
	%%%%%%%%%%%%%%%%%%%%%%%%%%%%%%%%%%%%%%%%%%%%%%%%%%%%%%%%%%%%%%%%%%%%%%%%%%%%%%%
	\chapter*{\vspace{-4cm}Abstract in
		English}\addcontentsline{toc}{chapter}{Abstract in English}
	\markboth{\MakeUppercase{Abstract in English}}{\MakeUppercase{Abstract in 
			English}}
	%%%%%%%%%%%%%%%%%%%%%%%%%%%%%%%%%%%%%%%%%%%%%%%%%%%%%%%%%%%%%%%%%%%%%%%%%%%%%%%
	In this dissertation we demonstrate that the
continuous-time quantum walk models remain powerful for nontrivial graph structures.
We consider two aspects of this problem.

First, it is known that the standard Continuous-Time Quantum Walk (CTQW),
proposed by Childs and Goldstone, can propagate quickly on the infinite path graph.
However, the Schr\"odinger equation requires the Hamiltonian to be symmetric,
and thus only undirected graphs can be implemented. In this thesis, we address the question,
whether it is possible to construct a continuous-time quantum walk
on general directed graphs, preserving its propagation properties.

Secondly, the quantum spatial search defined through CTQW has been proven to work
well on various undirected graphs. However, most of these graphs have very
simple structures. The most advanced results concerned the \ER model of random graphs, which
is the most popular but not realistic random graph model, and \BA random graphs,
for which full quadratic speed-up was not confirmed. In the scope of this aspect we analyze, whether quantum speed-up is observed for complicated graph structures as well. 

The dissertation consists of seven chapters. In Chapter~\depref{sec:intro} we
provide an introduction and motivation. In Chapter~\depref{sec:preliminaries}
we present a notation and preliminary concepts used in the dissertation.

In Chapters~\depref{sec:nonmoralizing-qsw} and~\depref{sec:convergence-qsw} we
approach the first aspect. In Chapter~\depref{sec:nonmoralizing-qsw} we
propose a nonmoralizing global interaction quantum stochastic walk, which is
well-definable on an arbitrary directed graph. We show that this model propagates rapidly on an infinite
path graph. In order to achieve the propagation speed better than the classical one, we
introduce a small amplitude transfer in the direction opposite to the direction of the existing
arcs. In Chapter~\depref{sec:convergence-qsw} we analyze the convergence
properties of the introduced model. We also analyze two other quantum stochastic walk
models called local and global interaction quantum stochastic walks. We show that
each of these models has very different properties. In particular, local and
nonmoralizing global models present the most intuitive behavior on directed
graphs. Our analysis shows that it is indeed possible to introduce a fast
continuous-time quantum walk which is well-definable on general directed graphs.

In Chapters~\depref{sec:hiding} and~\depref{sec:complex} we study the
second of the posed questions. In Chapter~\depref{sec:hiding} we correct and improve
state-of-the-art results on \ER graphs. We also demonstrate that the quantum speed-up
is correct for all vertices, instead of only `most of them'. Compared to
the previous state-of-the-art results we show that Laplacian matrix is a much simpler
operator to be taken into consideration compared to the adjacency matrix. In
Chapter~\depref{sec:complex} we compare three different operators plausible for
the quantum spatial search. We show that the normalized Laplacian, under certain
conditions, provides the full quadratic speed-up. We analyze two random graph models
which output the graphs with complex structure with high probability. The
analysis confirms that the proposed operator is indeed better than other
commonly used operators. Finally, we propose the procedure which solves the
problem of determining the optimal time for running the quantum
search algorithm.

Finally, in Chapter~\depref{sec:conclusions} we review and conclude our results.
The dissertation also consists of two Appendix sections, which provide the
proofs for the results used in the dissertation.

	\mainmatter	

%%%%%%%%%%%%%%%%%%%%%%%%%%%%%%%%%%%%%%%%%%%%%%%%%%%%%%%%%%%%%%%%%%%%%%%%%%%
%%%%%%%%%%%%%%%%%%%%%%%%%%% beginning %%%%%%%%%%%%%%%%%%%%%%%%%%%%%%%%%%%%%
%%%%%%%%%%%%%%%%%%%%%%%%%%%%%%%%%%%%%%%%%%%%%%%%%%%%%%%%%%%%%%%%%%%%%%%%%%%

%\todo[inline]{preface?}

\chapter{Introduction} \label{sec:intro}
% !TeX spellcheck = en_US

% State of the art

Recently, quantum computers have attracted a huge attention. This is because
	such devices can solve vital computational problems faster than their classical
	counterparts. What is more interesting, the speed-up is observable even in
	the complexity of algorithms. The best example is the Shor's algorithm
	\cite{shor1994algorithms} which solves the integer factorization problem in
	polynomial time in the terms of number length. It is notable to recall that any
	known classical algorithm that solves the same problem requires exponential time
	in a number of bits. Furthermore, the algorithm may threaten the current cryptographic
	protocols, as it can easily break RSA encryption.

The Shor's algorithm and other quantum algorithms \cite{deutsch1992rapid,grover1996fast} started an
important and beautiful field called quantum computer science. The goal
	of this discipline is to construct the algorithms which are faster compared to
the currently known algorithms for conventional computers. Despite numerous important theoretical algorithms
\cite{deutsch1992rapid,simon1997power,aaronson2018forrelation}, there are also
the algorithms which have the potential practical application. One can point to the Grover's
algorithm and its extensions \cite{grover1996fast,ambainis2019quantum}, quantum
annealing algorithms \cite{kadowaki1998quantum,finnila1994quantum}, variational
optimization algorithms
\cite{peruzzo2014variational,farhi2014quantum,glos2020space,tabi2020quantum}, 
and Quantum PageRank \cite{paparo2012google,sanchez2012quantum}.

Quantum algorithms can be divided into various classes according to the problem
they solve or the computational model they are based on \cite{montanaro2016quantum}.
In this dissertation, we focus on a particularly important class called
	\textit{quantum walks}, in which the amplitude transfer is done within some
underlying graph structure
\cite{childs2002example,aharonov2001quantum,childs2004spatial}. It can be
considered as an equivalent of random walk algorithms, where the probability mass
transfer is not allowed when the states are not connected.

\begin{figure}
	\subfloat[\label{fig:introduction-classical}random walk]{\includegraphics{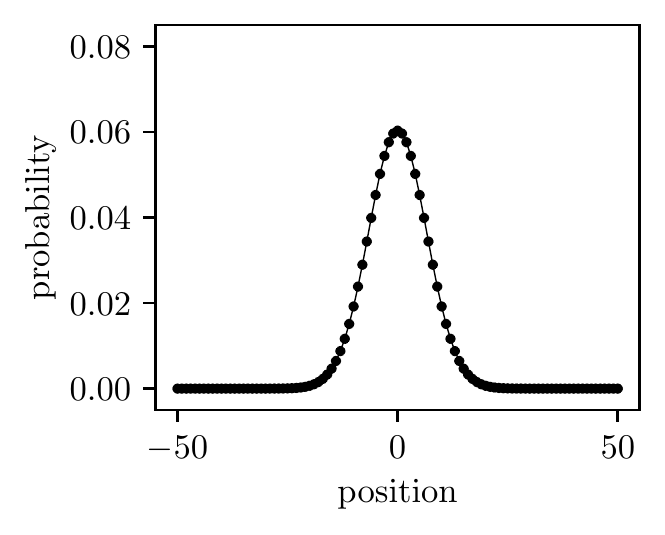}}
	\subfloat[\label{fig:introduction-quantum}quantum walk]{\includegraphics{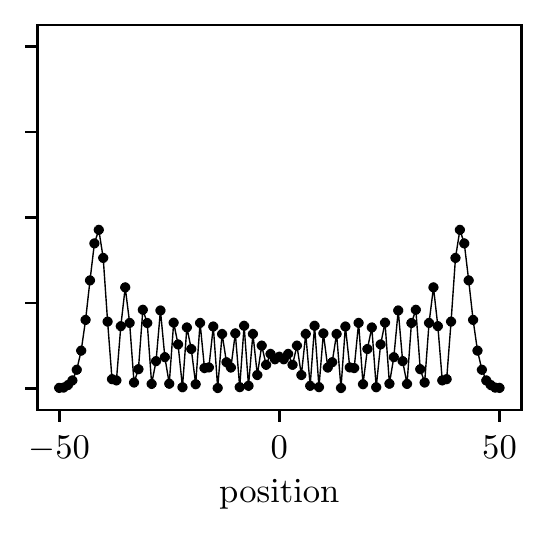}}
	\caption{Distribution of continuous-time random walk (left) vs continuous-time quantum walk (right) on a path graph with 101 nodes and after evolution time 22. The evolution starts in the middle of the graph.}\label{fig:classical_vs_quantum}
\end{figure}

Quantum walks application comes in particular from its ballistic propagation. Let
us consider a random walk on an infinite path with the probability localized at
position 0. Then after time $t$, the probability distribution of finding the walker can be well approximated by Gauss distribution $N(0,\Theta(t))$~\cite{portugal2013quantum}. Since the standard deviation grows proportionally
to the square root of time, we say that the stochastic process obeys a normal
diffusion. This is contrary to a quantum walk, where the variance grows
like $\Theta(t^2)$~\cite{konno2005limit}, i.e. we can observe the ballistic diffusion.
Thus the propagation in a quantum walk may be much faster which may explain the
speed-up appearing in quantum walk algorithms. The resulting
	distributions for both classical and quantum walks are presented in
	Fig.~\ref{fig:classical_vs_quantum}.

% Motivation

Despite the fact that the very first quantum walk is almost 20 years old, there are two
important questions regarding the generality of the results in the terms of graph
structure. Many quantum closed-system walk models, proposed so far, were
definable on relatively general graph
structures~\cite{szegedy2004quantum,childs2004spatial,portugal2016staggered}.
However, it was shown under general and reasonable assumptions that by using
the closed-system quantum evolution one cannot define a quantum walk on a general
directed graph \cite{montanaro2007quantum}. This results from the
quasi-periodicity of the closed-system evolution, i.e.~there exists
an arbitrarily large time evolution $t$ after which the system evolves to the state
close to the initial state. This in turn implies that a closed-system quantum walk
can only be well-defined for graphs, where, for arbitrary two
nodes, there is a path connecting them.

Since close-system quantum walks are not sufficient to model the evolution on general directed graphs, interactions with the environment are necessary. However,
currently known open quantum walk models do not yield the ballistic
propagation~\cite{bringuier2017central,attal2015central,sadowski2016central}. In
particular, for the continuous-time open quantum walk \cite{bringuier2017central},
the classical evolution destroys its coherence, and the proposed model lacks the
ballistic propagation. It has been an open question whether there exists a quantum
walk model which preserves the directed graph structure and whose propagation is
better than the propagation observed in random walks. This may be important, for
the directed graph model, for example in the case of the evolution for classical heuristic optimization
algorithms like simulated annealing or Tabu Search.

There is a similar lack of generality for quantum spatial search algorithms. The quantum search algorithms are defined as the graph-restricted evolution, which aims at finding a
marked node. Note that there are known examples of discrete quantum walks,
yielding even a quadratic speed-up over an arbitrary Markov-chain walk
\cite{ambainis2020quadratic,szegedy2004quantum}. However, general and simple results are still missing for continuous-time
quantum walks.

The first continuous-time quantum spatial search algorithm
\cite{childs2004spatial} has been deeply investigated for the special classes of
graphs like complete graphs \cite{childs2004spatial}, grid graphs
\cite{childs2004spatial}, binary trees \cite{philipp2016continuous}, simplex of
complete graphs \cite{meyer2015connectivity}, and others~\cite{novo2015systematic,glos2019impact_chimera,wong2016laplacian}. Based on
these results, the special properties of quantum walks were presented.
While the obtained results were an important step toward the development of
quantum search algorithms, all of the graphs considered were almost regular (meaning all vertices have very similar degrees) and
we can split the vertices into several classes (so-called vertex-transitivity),
within which the vertices are indistinguishable.

The first approach in generalizing the above results was made for \ER graphs
\cite{chakraborty2016spatial,chakraborty2017optimal,
	cattaneo2018quantum,glos2018vertices,kukulski2020comment}. While these graphs
are not regular, the deviations between the highest and the smallest degrees are
sufficiently small to provide very tight results on the efficiency of quantum
search on these graphs. Then three more general results were provided. The first one
showed a quadratic speed-up compared to a general Markov-chain search
\cite{chakraborty2020finding}, at the cost of larger Hilbert space.
Additionally, in \cite{chakraborty2020optimality} quite
general conditions for (optimal) quantum search for original continuous-time
spatial search were presented. However, the application of these results required the full
eigen-decomposition of the graph-based Hamiltonian, which in general is a hard
computational task. In fact, this task is much more demanding compared to
the quantum or even the classical search itself. Finally, in \cite{osada2020continuous}
the authors determined the efficiency of the quantum spatial search for complex graphs. However, while the speed-up over the classical search was shown, it remains an open
question whether the quadratic speed-up over the Markov search is achievable.

\paragraph{Dissertation overview} 

In the scope of the dissertation we demonstrate that the continuous-time
quantum walk models remain powerful for nontrivial graph structure. The analysis was done by approaching two problems:
\begin{enumerate}
\item Does a time-independent continuous-time quantum walk model which is
definable for general directed graphs and which maintains fast propagation exist?
\item Is the original Continuous-Time Quantum Walk based spatial search
\cite{childs2004spatial} powerful enough to offer the speed-up for heterogeneous
graphs?
\end{enumerate}
Note that for the proposed problems the context of `nontrivial graph structure' changes. For the first problem, we focus on directed graphs, while for the second problem -- on undirected graphs with significant deviation between the degrees of vertices. Currently, most of the results for quantum search considers almost regular graphs. Therefore, we consider heterogeneous graphs as a reasonable next step for investigation.

The first problem is approached using the formalism of quantum stochastic walks
\cite{whitfield2010quantum}. The model is a generalization of both
Continuous-Time Quantum Walk \cite{childs2004spatial}  and continuous-time random
walk. For the second problem, we analyze the CTQW-based spatial search
\cite{childs2004spatial} on random heterogeneous graphs and complex \BA graphs
\cite{albert2002statistical}. The latter is a paradigmatic random graph model
which simulates Internet network evolution.

The dissertation is organized as follows. In Chapter~\depref{sec:preliminaries} we present
a notation and preliminary information used in the dissertation. In
Chapter~\depref{sec:nonmoralizing-qsw} we analyze a quantum walk model
presented in \cite{whitfield2010quantum}, in the context of the propagation. We improve the model into \emph{nonmoralizing quantum stochastic
	walk} which is well-defined on any directed graphs. In order to achieve better than classical diffusion, we allowed a small amplitude transfer in the direction opposite to the direction of the existing arcs. In
Chapter~\depref{sec:convergence-qsw} we present convergence properties of the
introduced model and compare it to other well-known quantum stochastic walk
models. We confirm that the structure of the directed graph is
indeed preserved. In Chapter~\depref{sec:hiding} we improve the
results for \ER presented in \cite{chakraborty2016spatial}, in order to clarify
the approach to the analysis of CTQW-based spatial search to random graphs. In
Chapter~\depref{sec:complex} we present the analysis of the spatial search
algorithm for heterogeneous and complex graphs. Finally, in
Chapter~\depref{sec:conclusions} we justify the correctness of our hypothesis
in the context of the results presented in the
dissertation.

The results presented in Chapters~\depref{sec:nonmoralizing-qsw} and
\depref{sec:convergence-qsw} are based on the results from~\cite{domino2018properties,domino2017superdiffusive,
	glos2017limiting}. The results presented in
chapters~\depref{sec:hiding} and \depref{sec:complex} are based on the results from
\cite{glos2018vertices,glos2018optimal,kukulski2020comment,glos2020complex}.

% !TeX spellcheck = en_US
\chapter{Preliminaries} \label{sec:preliminaries}

\section{General preliminaries} 
In this chapter we introduce basic notation concepts used in the dissertation. The notation includes basics of set theory, complexity notation and linear algebra.

\subsection{Set theory notation} 

We will denote by $\ZZ$, $\RR$, $\CC$ the set of integers, real numbers and
complex numbers. We will use notation $\ZZ_{\geq 0}$ to denote the set of
non-negative integers, similarly for positive, negative and nonpositive, and for
other sets. We will write $|X|$ for the number of elements of the set $X$. We
will apply the notation $[n] \coloneqq \{1,\dots,n\}$.

Let $X$ be an arbitrary set and $\tilde X= \{Y \subseteq X\}$ be such a family
of sets that for arbitrary $Y_1,Y_2\in \tilde X$ we have $Y_1=Y_2$ or $Y_1 \cap
Y_2 = \emptyset$. Furthermore let
\begin{equation}
\bigcup_{Y\in \tilde X} Y = X.
\end{equation}
Then we call $\tilde X$ a partition of $X$.
	
\subsection{Complexity notation} \label{sec:complexity-notation}

Throughout the dissertation we will use the big O notation. Let $f:\RR_{> 0} \to
\RR$ and $g:\RR_{>0} \to \RR_{>0}$ be functions. We will write $f(x) =\order{g(x)}$ if there exists $x_0>0$
and $C> 0$ such that for all $x>x_0$ we have
\begin{equation}
|f(x)| \leq Cg(x).
\end{equation}
In fact $\order{g(x)}$ is usually
defined as a set of all functions $f$ satisfying mentioned relation, so formally
one should write $f(x)\in \order{g(x)}$. However, in the dissertation we will
follow a widely accepted computer science convention and use `$=$' instead of
`$\in$'.
	
With $\order{\cdot}$ notation we can define other asymptotic notations. We
present them and their definition in Tab.~\ref{tab:asymptoptic-definitions}.
Note that any of these symbols hide the constant next to the leading term. For
example, if $f(x) = \order{n^2}$, then at the same time $f(x) = \order{2n^2}$ or
$f(x) = \order{\frac{1}{2}n^2}$. In case we know the constant next to the
leading term we will write $f(x) \sim Cn^2$, which is defined as $f(x) =
Cn^2(1+o(1))= Cn^2 + o(n^2)$.

\begin{table}
	\centering
	\begin{tabular}{l@{\qquad}l}
		Notation & Definition \\\hline
		$f(x) = \order{g(x)}$ & see Sec.~\ref{sec:complexity-notation} \\
		$f(x) = \Omega(g(x))$ & $g(x) = \order{f(x)}$\\
		$f(x) = \Theta(g(x))$ & $g(x) = \order{f(x)}$ and $f(x) = \order{g(x)}$ \\
		$f(x) = o(g(x))$ & $\lim_{x\to\infty} \frac{f(x)}{g(x)} = 0$\\
		$f(x) = \omega(g(x))$ & $\lim_{x\to\infty} \left | \frac{f(x)}{g(x)} \right | = +\infty$\\
		$f(x) \sim g(x)$ & $f(x)=g(x)(1+o(1))$
	\end{tabular}
	\caption{Asymptotic notations and their definitions\label{tab:asymptoptic-definitions}.}
\end{table}

\subsection{Linear algebra} 

Let $X$ be a countable set. Let $\C^X$ be a complex-vector space and let
$\{\ket{x}: x\in X\}$ be its orthonormal basis. Arbitrary vector
$\ket{\psi}\in\C^X$ has a unique representation
\begin{equation}
\ket{\psi} = \sum_{x\in X} \alpha_x \ket{x},
\end{equation}
where $\alpha_x\in \C$. We call $\{\ket{x}: x\in X\}$ a computational basis and
we choose them to be of the form
\begin{equation}
\ket{x} = \begin{bmatrix}
0 \\ \dots \\ 0 \\ 1 \\ 0 \\ \dots \\ 0
\end{bmatrix},
\end{equation}
where $1$ appears on the $x^{\rm th}$ position. A conjugate
transpose of $\ket{\psi}$ is defined as
\begin{equation}
\bra{\psi} \coloneqq \left(\ket{\psi}\right)^\dagger = \sum_{x\in X} \bar\alpha_x \bra{x},
\end{equation}
where $\bar \alpha$ is a conjugate of $\alpha$ and $\bra{x}$ is a row vector
with $1$ on the $x^{\rm th}$ position and $0$ otherwise. Note that
the conjugate transpose is a composition of transpose and element-wise
conjugation of the vector.

If $\ket{\psi},\ket{\phi}\in\C^X$, then $\braket{\psi}{\phi}$ denotes their
scalar product. The outer product of vectors $\ket{\psi}\in
\C^X,\ket{\phi}\in\C^Y$ is denoted as $\ketbra{\psi}{\phi}$. The tensor product
of states $\ket{\psi}=\sum_{i=1}^n\alpha_{x_i}\ket{x_i}\in \C^{X}$ with
$X=\{x_1,\dots,x_n\}$ and $\ket{\phi}\in \C^Y$ is defined as
\begin{equation}
\ket{\psi}\otimes \ket{\phi} = \begin{bmatrix}
\alpha_{x_1}\ket{\phi} \\ \vdots \\ \alpha_{x_n}\ket{\phi}
\end{bmatrix}\in \CC^{X}\otimes \CC^Y.
\end{equation}
We will also use abbreviations $\ket{\psi,\phi}$, $\ket{\psi\phi}$ instead of
$\ket \psi \otimes \ket \phi$. Note $\CC^X\otimes \CC^Y$ is isomorphic to
$\CC^{X \times Y}$.

Let $B\in \C^{X\times Y}$ be a complex-valued matrix. Then $B$ has a unique
representation
\begin{equation}
B = \sum_{x\in X} \sum_{y\in Y} b_{xy} \ketbra{x}{y}.
\end{equation}
The vectorization of $B$ is defined as $\vecc{B} = \sum_{x\in X} \sum_{y\in Y}
b_{xy} \ket{xy}$. Furthermore we define a conjugate transpose of $B$
\begin{equation}
B^\dagger \coloneqq \sum_{x\in X} \sum_{y\in Y} \bar b_{xy} \ketbra{y}{x}.
\end{equation}

Suppose $B$ is a square matrix, \ie{} $B\in \CC^{X\times X}$. If $B^\dagger B =
BB^\dagger$, then we call $B$ a \emph{normal matrix}. For such matrices an
eigendecomposition can be found, i.e. there exists
$\lambda_1(B),\dots,\lambda_{|X|}(B)\in \C$ and orthonormal vectors
$\ket{\lambda_1}(B),\dots,\ket{\lambda_{|X|}(B)}\in\C^X$ such that
\begin{equation}
B = \sum_{i=1}^{|X|} \lambda_i(B) \ketbra{\lambda_i(B)}.
\end{equation}
We call $\lambda_i(B)$ an eigenvalue and $\ket{\lambda_i}(B)$ a corresponding
eigenvector of $B$. Whenever it will be clear from the context, we will write shortly $\lambda$ and $\ket{\lambda}$ instead of $\lambda(B)$ and $\ket{\lambda(B)}$. Furthermore, if all eigenvalues are real we will  assume that $\lambda_i \geq \lambda_j $ for $j< i$.

The space of normal matrices encapsulates many classes of matrices important for
quantum mechanics. In particular if $B^\dagger = B$, then $B$ is Hermitian. If
$B^\dagger B = \Id$, where $\Id$ is identity matrix, then we call $B$ a unitary
matrix. Eigenvalues of Hermitian operators are real, while eigenvalues of
unitary matrices are complex and lie on unit circle.

Matrix $B\in \CC^{X\times X}$ is called nonnegative if for any vector
$\ket{\psi}\in\CC^X$ we have $\bra \psi B \ket{\psi} \geq 0$. If additionally
the trace of $B$ equals 1,
\begin{equation}
\tr(B) \coloneqq \sum _{x\in X} \bra x B \ket x = 1,
\end{equation}
then we call $B$ a \emph{density operator}. It can be shown that eigenvalues of
$B$ form a proper probability vector, i.e. they are nonnegative and they sum up
to 1.

Matrix $B$ is called stochastic if its columns are proper probability
vectors.

Let $A\in\C^{X\times X}$ be a matrix and $\tilde X=\{X_1,\dots,X_k\}$ be a
partition of $X$. We can construct a $\tilde X$-block representation of $A$ as
\begin{equation}
\left [ \begin{array}{c|c|c|c}
A_{1,1} & A_{1,2} & \dots & A_{1,k} \\\hline
A_{2,1} & A_{2,2} & \dots & A_{2,k} \\\hline
\vdots & \vdots & & \vdots \\\hline
A_{k,1} & A_{k,2} & \dots & A_{k,k} 
\end{array}\right ],
\end{equation}
where $A_{i,j}\in \C^{X_i\times X_j}$ satisfies $\bra{k} A_{ij} \ket{l} = \bra k
A \ket l$ for any $k\in X_i$ and $l \in X_j$. Note that currently a block matrix
is considered to be a matrix, which elements are matrices. Such
definition implies even the change of how the multiplication is defined. Our
definition is used for representation purposes only. We call $A$ a $\tilde X$-block
diagonal matrix iff for all $i\neq j$ matrix $A_{ij}$ is a zero matrix.

In the dissertation we will oftenly choose $X=[n]$. In this case the vector $\ket{i}$ will
always have 1 on $i^{\rm th}$ position. 

\section{Graph theory} 

\subsection{General concepts} 

We call a pair $\vec G=(V,\vec E)$ a simple directed graph (digraph), iff $V$ is
a finite set and $\vec E\subseteq \{(v,w):v,w\in V, v\neq w\}$. We call the
elements of $V$ vertices or nodes, and of $\vec E$ arcs. We call $|V|$ the order of the
digraph and $|\vec E|$ the size of the digraph. Similarly we call a pair
$(V,E)$ a simple undirected graph, iff $V$ is a finite set and $\vec E\subseteq
\{\{v,w\}:v,w\in V,v\neq w \}$ with order $|V|$ and size $|E|$. We call the
elements of $V$ vertices or nodes and of $E$ edges.

Note that a directed graphs can be considered as an undirected graph if for any
$(v,w)\in \vec E$ we have $(w,v)\in \vec E$. Thus many definitions for directed
graphs can be formulated for undirected graphs as well. Because of this, unless
explicitly stated, we will provide definitions for directed graphs only. 

Let $\vec G=(V,\vec E)$ be a directed graph. We call $G=(V,E)$ an underlying
graph of $G$ iff
\begin{equation}
\{v,w\}\in E \iff \left ((v,w)\in \vec E \lor (v,w)\in \vec E\right).
\end{equation}
Note that the undirected graph is its own underlying graph. Conversely the
directed graph $\vec G=(V,\vec E)$ is an orientation of the graph $G=(V,E)$ if
each edges $\{v,w\}$ is replaced with a either $(v,w)$ or $(v,w)$. Note that we
have $2^{|E|}$ orientations of graph $G$, but there is single underlying graph
for digraph $\vec G$.

We call set $P(v)=\{w\in V:(w,v)\in \vec E\}$ a set of parents of $v\in V$. We
define a children set of $v\in V$ as the collection $C(v)=\{w\in V:(v,w)\in \vec
E\}$. We define indegree and the outdegree of $v$ as a sizes of these sets i.e.
$\indeg(v) = |P(v)|$ and $\outdeg(v) = |C(v)|$. If $\indeg (v)=0$ we call $v$ a
source. If $\outdeg(v)=0$ then we call $v$ a sink or a leaf. A
collection of all sinks (leaves) of a digraph $\vec G$ is denoted by $L(\vec
G)$. Note that for undirected graphs we have $\deg(v)\coloneqq
\indeg(v)=\outdeg(v)$ which is simply a degree of the vertex $v$

A path form $v_1$ to $v_{k+1}$  is a sequence $(v_1,\dots, v_{k+1})$ such that
$(v_{i},v_{i+1})\in \vec E$ for each $i=1,\dots,k$, and $k$ is called a length
of the path. We say a digraph (graph) is strongly connected (connected) iff for
each $v,w\in V$ there exists a path from $v$ to $w$. We say that a digraph is
weakly connected iff its underlying graph is connected. The distance $d(v,w)$
from $v$ to $w$ is defined to be the minimum length of all paths from $v$ to
$w$. Note that in general for directed graphs $d(v,w)\neq d(w,v)$.

If a path $(v_1,\dots, v_k)$ does not have a vertex repetition except $v_1=v_k$
then we call it a simple path. If $v_1=v_k$ then we call it a cycle. If digraph
does not have a cycle of length $k\geq 3$ then we call it acyclic. We will call
undirected graph $G$ a tree if it has no cycles and is connected.

Let $\vec G=(V,\vec E)$ be a directed graph. A directed graph $\vec H=(V_H, \vec
E_H)$ is called a subgraph iff $V_H\subseteq V$ and $\vec E_H \subseteq \vec 
E$. We denote this fact by $\vec H\subseteq \vec G$. If $\vec E_H$ is maximal in
the number of arcs, \ie{} is of the form
\begin{equation}
E_H = \{(v,w)\in \vec E: v,w\in V_H\},
\end{equation}
the we call $\vec H$ an induced subgraph, which we denote $\vec H \subseteq_{\rm
	ind} \vec G$. Note that given subset of vertices $V_H$ there is a unique induced
subgraph of $\vec G$, however there may be multiple subgraph. Maximal connected
subgraph is called a connected component.

Let $\vec G_1=(V_1,\vec E_1)$ and $\vec G_2=(V_2,\vec E_2)$ be directed graphs
and $f:V_1\to V_2$. We call $f$ a graph homomorphism from $\vec G_1$ to $\vec
G_2$ iff for each $(v,w)\in \vec E_1$ we have $(f(v), f(w))\in\vec   E_2$. If
$f$ is bijection and both $f$ and $f^{-1}$ are homomorphisms then $f$ is an
isomorphism, and we call $G$ and $H$ isomorphic graphs.

Isomorphism from $\vec G$ to $\vec G$ is called automorphism. We call $G=(V,E)$
a vertex-transitive graph, if for any $v,w\in V$ there exists a automorphism
$f:V\to V$ such that $f(v)=w$.

Let $A(\vec G)\in \R^{V\times V}$ be an operator defined as
\begin{equation}
\bra w A(\vec G) \ket v = \begin{cases}
1,& (v,w)\in \vec E,\\
0,& \textrm{otherwise}.
\end{cases}
\end{equation} 
We call $A(\vec G)$ an adjacency matrix of $\vec G$. Note that $A(\vec G)$ is
symmetric iff graph is undirected. Furthermore, in the literature the adjacency
matrix is usually the transpose of the operator above, however our definition is
more convenient based on form of evolution considered in this dissertation. If
clear from the context which graphs is considered, we will write simply $A$
instead of $A(\vec G)$.

Let $D(G)\in \R^{V\times V}$ be a diagonal matrix such that $\bra{v}D(G)\ket{v}= \deg
(v)$. We define (combinatorial) Laplacian as $L(G)\coloneqq D(G) - A(G)$ and
normalized Laplacian as $\mathcal L(G)\coloneqq D(G)^{-1/2}L(G)D(G)^{-1/2}= \Id
-D(G)^{-1/2}A(G)D(G)^{-1/2}$. Note that the normalized Laplacian is well-defined only
for graphs without isolated nodes, i.e. nodes with degree 0. Laplacian and
normalized Laplacian are always nonnegative. The multiplicity of
eigenvalue 0 for both equals the number of connected components of $G$.

We call adjacency matrix, Laplacian and normalized Laplacian \emph{graph
	matrices}.

\subsection{Random graphs} 

Random graph model $\randgn$  is a probabilistic measure space defined over a
set of graphs with $n$ vertices. Precise definition requires the notion of
measurable space. However, it is common to provide the sampling method instead
of writing exact form of probability distribution. Each random graph model will
be denoted by $\randgn[LABEL](\nu_1,\dots,\nu_k)$, where $\rm LABEL$ is the
label setting the sampling method and $\nu_1,\dots,\nu_k$ are free parameters of
sampling method. Note that parameter $\nu_i$ may depend on $n$.

Let us recall here the most popular random graph models. We will start with the
\ER random graph model $\randgn[ER](p)$, where $p\in [0,1]$
\cite{erdos1959random}. The sampling method goes as follows. We start with empty
graph $G=([n],\emptyset)$. Then for each pair of different vertices $v,w\in [n]$
we add edge $\{v,w\}$ independently with probability $p$. Similarly for directed
graphs each arc $(v,w)$ is added independently with probability $p$. The random
graph model is so popular that in many papers authors by `random graphs'
consider precisely this model. The reason for such is, beside the fact that it
is the first random graph model proposed, is because for $p=1/2$ we have uniform
distribution over all graphs with fixed vertex set.

Unfortunately, while the model is well known, it does not represent the
real-world dependencies. In particular, \ER graphs do not have power-law degree
distribution, which means that vertices with degree $d$ are present withf
$\Theta(d^{-\alpha})$ probability for some constant $\alpha$
\cite{albert2002statistical}. Real graphs usually are also small-world, which
means the existence of small-length paths between all vertices, sparse (have
small number of edges) and one can often observe community structures, i.e.
small but dense subgraphs.

There are many graph models which may possess some of these properties. The
closest to the \ER graph model is the \CL model $\randgn[CL](\omega)$~\cite{chung2002connected}. This model depends on single parameter being a
real-valued vector $\omega=(\omega_1,\dots,\omega_n)\in [0,n-1]^n$. Similarly as
for the \ER model we start with empty graph $G=([n],\emptyset)$ with $n$
vertices, and an edge between vertices $i$ and $j$ is added with probability
$\omega_i\omega_j/\sum_k \omega_k$. Let us assume for now that we allow loops.
Then
\begin{equation}
\EE \deg(i) = \sum_{j=1}^n \EE \mathbf 1_{\{i,j\}\in E} = \sum_{j=1}^{n}\left (w_iw_j/\sum_{k=1}^nw_k \right )=w_i.
\end{equation}
where $\mathbf 1_{\varphi}$ is an random variable which outputs one if $\varphi$ is satisfied. While we will remove all
self-loops at the end of sampling method, for large $n$ this simplification has
negligible impact. Note that for proper choice of $\omega$ one obtains almost
surely power-law graphs~\cite{chung2011spectra}.

Very well known \BA random graph model $\randgn[BA](m_0)$ with $m_0\in \ZZ_{\geq
	1}$, which was designed to simulate evolution of Internet network
\cite{albert2002statistical}. The procedure iteratively adds vertices as long as
the final graph has $n$ vertices. There are two nonequivalent sampling
procedures that share similar concept and produce graphs with similar
properties. In the first, original version algorithm starts with complete graph
$K_{m_0}$. Then new vertex  $v$ is added, and is connected to $m_0$ already
existing vertices. Already present vertex $w$ will be a neighbor of $v$ with
probability $\deg(w)/(2|E|)$. With such procedure only simple, connected graphs
are sampled, with power-law distribution.

Unfortunately, the very first definition provided by Barab\'asi and Albert in
\cite{albert2002statistical} was not precise, and in past years many
nonequivalent definitions were proposed and utilized
\cite{durrett2006random,flaxman2005high,lightgraphs}. Because of that, result
concerning \BA model presented in Chapter~\ref{sec:complex} will be strengthen
by numerical investigations based on the model implemented in
\cite{lightgraphs}.

The directed version of \BA random graph model $\randdgn[BA](m_0)$ is defined analogically, however instead of adding edge $\{v,w\}$ for newly added $w$, arc $(v,w)$ is added \cite{lightgraphs}. For the directed version of \ER graph model $\randdgn[ER](p)$ each \emph{arc} is independently added with probability $p$.

Finally, random orientation of an undirected grap $G=(V,E)$ is a digraph $\vec G=(V,\vec E)$, where each edge $\{v,w\}\in E$ was replaced by one of arcs $(v,w)$ or $(w,v)$. Note that for random orientation of a graph we have $|\vec E|=|E|$.

\section{Stochastic and quantum dynamics} 

\subsection{Stochastic evolution}

Let us consider closed system which can be in any of canonical state $x\in X$.
Provided the system follows the rules of statistical mechanics, its state can be
defined with the probability distribution over $X$. For finite $X$ it has a
probability vector representation $p = \sum_{x\in X} p_x\ket{x}\in \R^X$, where
$p_x$ is the probability that the system is in the state $x$.

%The evolution of the system can be described by continuous-time and
%discrete-time evolution. Throughout the dissertation we will be only interested
%in time-independent, linear evolution. Thus in case of discrete stochastic
%operation
%\begin{equation}
%p(t+1) = S p(t),
%\end{equation}
%operator $S$ needs to be $|X|\times |X| $ stochastic operation. $t$ denotes the step of the evolution.

We define continuous-time stochastic evolution through differential equation
\begin{equation}
\frac{\dd p(t)}{\dd t} = -L p(t), \label{eq:ctrw}
\end{equation}
with $L$ being a transition rate matrix \cite{childs2004spatial} and $t$ being
the evolution time. Columns of the transition rate matrix needs to sum to 0,
furthermore the diagonal elements needs to be nonnegative and out-diagonal needs
to be nonpositive. Note that the matrix does not have to be symmetric.

Suppose we have two probabilistic systems that can be in canonical states $X$
and $Y$ respectively. Then the joint system can be defined through probabilistic
vectors over $X\times Y$, \ie{} on $\RR^X\otimes \RR^Y$ space. In particular if
the first system is in state $p^X$ and second in $p^Y$, and the random these
systems are independent, then the joint system's probability vector takes the
form $p^X\otimes p^Y$.

One can also define a discrete-time stochastic evolution through stochastic
matrix $P$ and the evolution
\begin{equation}
p(t+1) = Pp(t).\label{eq:dtrw}
\end{equation}
Note that $P$ has to be a stochastic matrix in this case.

\subsection{Evolution in quantum systems}
\paragraph{Closed quantum system}

Let us consider closed system which can be in any canonical state $x\in X$. If
the the system obeys the laws of quantum mechanics, its (quantum) state has
representation as a normed vector in Euclidean space $\CC^X$, called a pure
state.

The continuous-time quantum evolution of the state $\ket{\psi_t}$ is defined by
Schr\"odinger equation
\begin{equation}
\frac{\dd\ket{\psi(t)}}{\dd t} = -\ii \hbar H \ket{\psi(t)}, \label{eq:schrodinger}
\end{equation}
where $H\in \CC^{X\times X}$ is the Hamiltonian of the system and $\hbar$ is the
reduced Planck constant. In order to defined quantum state preserving evolution
one need to assume that Hamiltonian is a Hermitian operator. It is common to
assume that $\hbar =1$, which can be done by careful physics units change.
 
Suppose $H$ does not depend on the evolution time $t$, and that at time $t=0$ we
start with quantum state $\ket{\psi(0)}$. Then the Eq.~\eqref{eq:schrodinger}
can be solved  analytically into
\begin{equation}
\ket{\psi(t)} = \exp(- \ii t H ) \ket{\psi(0)}.
\end{equation}
Note that $\exp(-\ii tH)$ is a unitary matrix.

Suppose we have two quantum systems. Furthermore, suppose the first system is in
state $\ket{\phi_X}\in\CC^X$ and second in $\ket{\phi_Y}\in\CC^Y$. Then the
state of the composite system is $\ket{\phi_X}\otimes
\ket{\phi_Y}\in\CC^X\otimes \CC^Y$. General state of the quantum system lies in
a in $\CC^X\otimes \CC^Y$ space. Similarly as for probability distribution, it
may be the case that a quantum state of the composite system cannot be written
in the form $\ket{\phi}\otimes\ket{\psi}$. Such dependence has different
properties and thus earned a new term -- entanglement. Entanglement takes a
vital role in quantum communication, however in context of quantum computation
it can be seen as an extension of superposition. Thus understanding of quantum
entanglement is not vital for understanding our results.

Contrary to the statistical mechanics, measurement may affect the state of the
quantum system. Suppose we have a classical system described by probability
vector $p$. Suppose the measurement is performed but its output is ignored. Then
the description of the system has not changed from our point of view---the system
can still be described by the same probability vector $p$.

This is no longer the case for the quantum system. Suppose the system $\CC^X$ is
in the state $\ket{\psi}$ and we perform the measurement in the canonical basis.
Then the classical outcome of the measurement will be $x\in X$ with probability
$|\braket{\psi}{x}|^2$, and the system will change into $\ket{x}$. In case we
ignore the classical output, we can only say that the system will be a
statistical mixture of quantum states $\{\ket{x}:x\in X\}$. The mathematical
representation of such measurement outcome requires a density operator notation,
which will be explained later in this chapter.

While in statistical mechanics different probabilistic vector represents
different states of the system, in quantum mechanics two different vectors may
be physically indistinguishable. If $\ket{\psi}$ is the quantum state of the
system, then $\exp(-\ii \alpha)\ket{\psi}$ for any $\alpha\in\RR$ is a
description of the very same state. This means, that no measurement can
distinguish these two states, even in probability. The factor $\exp(-\ii \alpha)$ is
called a \emph{global phase}. Because of that $H$ and $H+\alpha\Id$ represents
the same quantum evolution, since
\begin{equation}
\exp(-\ii t(H+\alpha \Id)) \ket{\psi} = \exp(-\ii tH)\exp(-\ii t\alpha \Id)) \ket{\psi}.
\end{equation}

\paragraph{Open quantum system} 

Suppose we have a quantum system that is in state $\ket{\psi}\in \CC^X$ with
probability $p$ and in $\ket{\varphi}\in \CC^X$ with probability $1-p$. For such
description we use the notion of mixed states. Such states can be represented
by density matrices---the example would take the form 
\begin{equation}
\varrho = p \dyad{\psi} + (1-p) \dyad{\varphi} \in \CC^{X\times X}.
\end{equation}
In general if a quantum system is in state $\ket{\psi_i}$ with probability
$p_i$, then the state can be described by a mixed quantum state
\begin{equation}
\varrho = \sum_i p_i \dyad{\psi_i}{\psi_i}.
\end{equation}
Note that if system is in a pure state $\ket{\psi}$ with probability 1, then we
have $\varrho = \dyad{\psi}$. In this representation there is no
notion of global phase, and different density operators results in different
quantum states. Similarly as in the case of pure quantum states, if two
separable systems are in states $\varrho_1$ and $\varrho_2$, then the joint
system is in state $\varrho_1\kron \varrho_2$.

We can enrich the quantum evolution to open-system evolution, which assumes the
existence of interactions of the quantum system with the environment. In the
scope of this dissertation we will consider GKSL master equation which takes the form \cite{gorini1976completely}
\begin{equation}
\frac{\dd \varrho}{\dd t} = -\ii \hbar [H, \varrho] + \sum_{L\in \mathbb L} \gamma_L\left(  L \varrho L^\dagger - 
\frac12 \{L ^\dagger L, \varrho\} \right). \label{eq:gksl-master-equation}
\end{equation}
Here $H \in \CC^{X\times X}$ is the Hamiltonian which describes coherent, closed
system interaction. Set $\mathbb L$ consists of of Lindblad operators $L\in \mathbb L$, which may be
arbitrary complex-valued $L\in \CC^{X\times X}$ matrices. Values  $\gamma_L$ are
called intensities. Lindblad operators describe dissipative, open-system
interactions. By $[A,B]=AB-BA$ we denote the commutator, and by $\{A, B\}=AB+BA$
the anticommutator. We assume $\hbar =1$ and $\gamma_L\equiv1$ for all $L\in
\mathbb L$. Note that the GKSL master equation encapsulates both Schr\"odinger
and stochastic equation.

Eq.~\eqref{eq:gksl-master-equation} has an equivalent representation of the form
\cite{miszczak2011singular}
\begin{equation}
\frac{\dd \vecc{\varrho_t}}{\dd t} = S \vecc{\varrho_t},
\end{equation}
where
\begin{equation}
S = -\ii \left(H \kron \Id - \Id \kron \bar H \right) +
\sum_{L \in \mathbb L } \left ( L \kron \bar L - \frac{1}{2} L^\dagger L \kron \Id - \frac{1}{2} \Id
\kron L^\top \bar L \right ). \label{eq:diff-qsw-operator}
\end{equation}
We call $S$ an evolution generator. Assuming the Hamiltonian and Lindblad
operators are time-independent, we can provide a solution of the system
\begin{equation}\label{eq:integrated-qsw}
\vecc{\varrho(t)} = \exp(St) \vecc{\varrho(0)},
\end{equation}
where $\varrho(0)$ is the initial state.

Quantum measurement theory can be enriched as well, however in our case we will
be satisfied with simply generalizing the notion of von Nuemann measurement into
mixed states notation. Let $\CC^{[n]}$ be a quantum system, and let
$\{\ket{\varphi_i}\}_{i=1}^n $ be its orthonormal basis. Let $\varrho$ be the
mixed quantum state of the system being measured. Then, the measurement outputs
$i$ with probability $\bra {\varphi_i} \varrho \ket{\varphi_i}$. In the scope of
this dissertation we will ignore the quantum state coming from the measurement,
and we can consider measurement as the operation destroying the quantum system.

%Similarly as Schr\"odinger equation can be generalized into GKSL master
%equation, it is possible to extend the density-matrix von Neumann measurements
%into so-called generalized measurement or POVM. For the purpose of the
%dissertation we are satisfied with the introduced measurement.

\section{Random and quantum walk theory} \label{sec:random_quantum_walk}
\subsection{Typical random and quantum walks}

Let $\vec G=(V,\vec E)$ be a directed graph and let $\R^V$ be the space of the
walker's evolution. The continuous-time random walk (CTRW) is defined through
continuous-time stochastic evolution given by Eq.~\eqref{eq:ctrw}. The
operator defining the walk has to satisfy $\bra w L \ket v =0$ iff $(v,w) \not
\in \vec E$. This way the probability is not (directly) transported through
nonexisting arcs. A discrete-time equivalent is defined through
Eq.~\eqref{eq:dtrw} with similar condition $\bra w P \ket v =0$ iff
$(v,w)\not\in \vec E$. Note that in general it is acceptable to have $\bra w P
\ket v =0$ even if $(v,w) \in \vec E$, however in this dissertation we consider
only time-independent evolution, and such situation essentially excludes the arc
from $\vec E$. Thus such relaxation of the definition is of no interest. We call
$p_v(t)$ a probability of state being at vertex $v$ after time $t$.

There is a special class of random walks with unbiased choice of probability
transfer. For such a walk, for any $(v,w),(v,w')\in \vec E$ we have $\bra w L
\ket v= \bra{w'} L \ket v$. In the case of discrete random walk, we define
equivalent definition $\bra w P \ket v=\bra {w'} P \ket v=\frac{1}{\deg v} $. We
will call this special kind of random walks a \emph{uniform random walks}. Note
that contrary to the discrete random walk, the choice of continuous-time random
walk is non-unique. In this dissertation we mostly consider Laplacian
matrix as an evolution operator of uniform random walk.

The continuous-time quantum walk (CTQW) is defined similarly as the
continuous-time random walk \cite{childs2004spatial}. Let $G=(V,E)$ be an
undirected graph and let $\C^V$ be the space of the walker's evolution. The
evolution is defined through the Schr\"odinger equation provided in
Eq.~\eqref{eq:schrodinger}, with extra condition $\bra w H \ket v=0$ iff
$\{v,w\}\not \in E$. Since the evolution operator has to be Hermitian, CTQW is
well-defined only for undirected graphs. The probability of the walker to be at a
vertex $v$ after evolution time $t$ equals $|\braket{v}{\psi(t)}|^2$. Similarly
as for the continuous-time random walk, we call quantum walk a uniform CTQW
iff $\bra w H \ket v = \bra {w'} H \ket v$ for any choice of $\{v,w\}\in E$.
Note that the normalized Laplacian defines a valid CTQW, while adjacency matrix
and Laplacian define a valid uniform CTQW. Note that the evolution for $d$-regular
graphs are equivalent independently on chosen graph matrix. Since
\begin{equation}
\exp(-\ii t L) = \exp(-\ii t (d\Id-A))=\exp(-\ii td )\exp(-\ii t(-A)),
\end{equation}
the evolution using Laplacian is equivalent to adjacency matrix up to the global phase and sign of $t$. Similarly 
\begin{equation}
\exp(-\ii t \mathcal L) = \exp(-\ii t \left (\Id - \frac{1}{d}A \right )) = \exp(-\ii t) \exp(-\ii\left (-\frac{t}{d}\right )A),
\end{equation}
hence the normalized Laplacian is equivalent to adjacency matrix up to global phase and the time rescaling $t \mapsto -\frac{t}{d}$. Similar equivalence for nonregular graphs does not take place, which also has impact on the application of the walk \cite{wong2016laplacian}.

It is far more complicated to design a discrete-time quantum walk. Asserting
similar condition to the above one for unitary matrix leads to pathological
evolution, counter-intuitively prohibiting amplitude transfer
\cite{glos2019role}. Instead it is common to enlarge the walker's space into
$\C^{[m]}$ and provide mapping $h:[m]\to V$. This way the probability of being
at vertex $v$ equals
\begin{equation}
\sum_{j\in h^{-1}(v)} |\braket{j}{\psi(t)}|^2.
\end{equation}
While discrete-time quantum walks are beyond the topic of this dissertation,
system enlargement in order to guarantee special quantum properties will be of
use in Chapter~\ref{sec:nonmoralizing-qsw}.

Finally let us define a quantum stochastic walk (QSW). This model was proposed in
\cite{whitfield2010quantum} to encapsulate the CTRW and CTQW, but also to
provide new form of dynamics. Let as recall GKSL master equation
\begin{equation}
\frac{\dd \varrho}{\dd t} = -\ii [H, \varrho] + \sum_{L\in\mathbb L} \left(  L
\varrho L^\dagger - \frac12  \{L ^\dagger L, \varrho\} \right).
\end{equation}

Note that GKSL master equation encapsulates continuous-time stochastic and
quantum evolutions \cite{whitfield2010quantum}. Hence we can define a local
environment interaction QSW as a mixture of stochastic and quantum evolution.
\begin{definition}
	Let $\vec G = (V, \vec E)$ be a digraph and $G=(V,E)$  be its underlying graph. Let $H\in \C^{V\times V}$ be a
hermitian operator such that for $v,w \in V$ satisfying $v\neq w$
	\begin{equation}
	\{v,w\} \not \in E \iff \bra{w}H\ket v =0.
	\end{equation}
	Let $\mathbb  L= \{c_{ij}\ketbra{j}{i}: (i,j)\in\vec E, c_{ij}\neq 0 \}$ be a
collection of Lindblad operators. Then we call a GKSL master equation with
Hamiltonian $H$ and Lindblad operators collection $\mathbb  L$ a \emph{local
	enviornment interaction QSW} (LQSW).
\end{definition}
Note that condition $c_{ij}\neq 0$ can be relaxed to $c_{ij}>0$ based on the
form of GKSL master equation. If $H$ is the adjacency matrix of $G$ and $c_{ij}
= 1$ for any $(i,j)\in \vec E$, then we call a LQSW  a \emph{standard LQSW}. If
we consider a GKSL master equation  of the  form
\begin{equation}
\frac{\dd \varrho}{\dd t} = -\ii (1-\omega) [H, \varrho] + \omega\sum_{L\in \mathbb L} \left(  L \varrho L^\dagger - 
\frac12 
\gamma_L
\{L ^\dagger L, \varrho\} \right). 
\end{equation}
with extra smooth transition parameter $\omega\in [0,1]$, we call it an
\emph{interpolated LQSW}. Note that for $\omega=0$ and $\omega=1$ we recover
respectively CTQW and CTRW. The LQSW has a potential to be applied in quantum
computer science  \cite{zimboras2013quantum,sanchez2012quantum}.

Note that in \cite{whitfield2010quantum} authors proposed a more complicated
quantum walk model, which cannot be represented as a LQSW. We will focus on this
model in the Chapter~\ref{sec:nonmoralizing-qsw}.

\subsection{Walk quality measures} \label{sec:preliminaries-propagation}

In order to quantify the `usefulness' of walks, several measures can be
proposed. In the scope of the dissertation, we will focus on the propagation on
the graph and the search efficiency.

\subsubsection{Propagation speed}

A propagation of a walk is typically defined by the pace of change of the second
moment of a position of the walker in time on a infinite path graph. Let
$G=(\ZZ, E)$ be an infinite path graph with $E=\{\{k,k+1\}:k\in \ZZ\}$. The
evolution starts in a state localized at position 0, which for both CTQW and
CTRW is $\ket{0}$. Provided $p_k(t)$ is the probability of walker being measured
in $k$ after evolution time $t$, the central second  moment equals
\begin{equation}
\mu_2(t) = \sum_{k\in \ZZ} k^2 p_k(t). 
\end{equation}

For time-independent walk we always have $\mu_2(t)=\order{t^2}$. It can be
easily explained in term of discrete walks. After $t$ steps the probability of
finding an element at position $k$ such that $|k| > t$ equals zero. Thus
\begin{equation}
\mu_2(t) = \sum_{k=-t}^t k^2 p_k^t \leq t^2 \sum_{k=-t}^t  p_k^t = t^2.
\end{equation}
Provided $\mu_2(t) = \Theta(t^\alpha)$, we use $\alpha$ as a measure of
propagation of the walk. We distinguish the following propagation regimes:
\begin{enumerate}
\item if $\alpha<1$, the process is sub-diffusive,
\item if $\alpha=1$, the process obeys a normal diffusion regime,
\item if $1<\alpha<2$, the process is super-diffusive,
\item if $\alpha=2$, the process obeys a ballistic diffusion regime.
\end{enumerate}
In general, higher exponent $\alpha$ means the walker faster propagates through
a graph, which in turn may provide algorithmic speed-up. It can be shown that
time-independent random walk obeys a normal diffusion regime, as the probability
distribution can be approximated by Gaussian distribution. On the other hand, one
can show that CTQW obeys a ballistic diffusion regime~\cite{konno2005limit}.

\subsubsection{Search algorithms} 

For general search problem it is assumed that the single or multiple elements of
the database are marked, and the goal is to find the elements in as short time
as possible. In case of walk search, vertices plays the role of elements of the
database, and the evolution has to be a walk defined by some graph. Contrary to
the propagation of the walk, the way the time required to find a marked vertex
is determined differs between random and quantum walks.

In case of discrete-time random walk given by stochastic matrix $P$, the
evolution starts in its stationary distribution $p_{\rm stat}$. At each step of
random walk evolution walker is checked whether it is at the marked vertex. If
it is the case, then the procedure stop, otherwise the step is repeated. Note
that for general graph it is possible (although highly unlike for uniform random
walk) to search for a marked vertex infinitely. However expected Mean
First-Passage Time $\langle T_{v} \rangle$ is necessarily finite
\cite{noh2004random}. The only condition on $G$ is to be strongly connected.
Note there are known explicit formulas for uniform random walk search
\cite{noh2004random}. However, in Chapter~\ref{sec:complex} we will derive and
present an alternative formula based on quantities used for quantum search
\cite{chakraborty2020optimality}.

Note that for the random walk search for almost all vertices $\langle
T_{v}\rangle =\Omega(|V|)$, since within $k$ steps at most $k$ different
vertices can be visited. This does not mean that some vertices cannot be found
in shorter time. For example for star graph, i.e. a tree graph with single
vertex connected to all the others (see Fig.~\ref{fig:star-graph}), the walker
either stops at step 0 or at step 1.

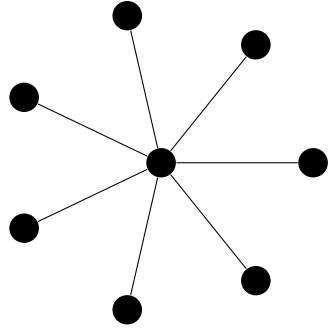
\begin{figure}\centering
	\begin{tikzpicture}
	\node[circle,fill=black] at (360:0mm) (center) {};
	\foreach \n in {1,...,7}{
		\node[circle,fill=black] at ({\n*360/7}:2cm) (n\n) {};
		\draw (center)--(n\n);}
%		\node at (0,-2*1.5) {$K_{1,7}$}; % delete line to remove label
	\end{tikzpicture}
	\caption{Star graph with 8 vertices.} \label{fig:star-graph}
\end{figure}

For the quantum walk we cannot measure the system too often, due to destructive
behavior of the measurement. Instead it is necessary to understand the graph
structure, and determine the optimal $T_{\rm opt}$ time, after which the quantum
state is measured. The choice of $T_{\rm opt}$ has to take into account the
success probability of checking the state after given evolution time. Since quantum evolution is quasi-periodic, choosing too large measurement time may also results in $\approx 0$ success probability, the same way as it happens for Grover search \cite{grover1996fast}.

In the scope of the dissertation we will focus on the original CTQW-based
quantum search algorithm. Let $G=(V,E)$ be an undirected graph and $w\in V$ be a
marked node. The Hamiltonian defining the evolution takes the form
\begin{equation}
H = - \gamma H_G - \ketbra{w}{w},
\end{equation}
and originally the evolution starts in uniform superposition $\ket
s=\frac{1}{\sqrt n}\sum_{v\in V}\ket v$. Here $H_G$ is a symmetric graph matrix
and $\gamma$ is a jumping rate which has to be determined before running the
algorithm.

Let us consider a complete graph $K_n$, \ie graph with all $n$ vertices being
connected. Since the graphs is vertex-transitive, and thus regular,u the choice
of graph matrix and marked vertex is not relevant. Note that $A+\Id = n
\ketbra{s}$. Hence, despite the fact the evolution takes places in
$n$-dimensional space, effectively it runs on a subspace spanned by $\{\ket s,
\ket w\}$. The initial state lies in this subspace. Let $\ket{\bar
	s} = \frac{1}{\sqrt {n-1}}\left ( \sqrt n \ket s - \ket w\right )$. The
Hamiltonian defined as $H_G=A$, spanned by $\ket{\bar s}$ and $ \ket w$,
equivalent to the above one takes the form
\begin{equation}
\tilde H= -\begin{bmatrix}
\gamma (n-2) & \gamma \sqrt {n-1}\\
\gamma\sqrt {n-1}& 1 \\ 
\end{bmatrix}.
\end{equation}
%\begin{equation}
%\bra{\bar s} A \ket {\bar s} = \bra{\bar s} (n\ketbra s - \Id) \ket {\bar s} = n \braket{\bar s}{s}^2 - 1 = \frac{n}{(n-1)n}(n-1)^2-1=n-2
%\end{equation}
Note for $\gamma = \frac{1}{n-2}$ the values on the diagonal equals hence is
irrelevant. For such choice of the $\gamma$, the simplified Hamiltonian takes
the form
\begin{equation}
\tilde H= -\begin{bmatrix}
0 & \frac{\sqrt {n-1}}{n-2}\\
\frac{\sqrt {n-1}}{n-2}& 0 \\ 
\end{bmatrix}.
\end{equation}
For $T_{\rm opt}\approx \pi\sqrt n /2$ the Hamiltonian transforms
state $\ket{\bar s}$ to $\approx \ket w$. Since $\braket{s}{\bar s}=1-o(1)$, the
same Hamiltonian transforms $\ket s $ to $\approx \ket w$ as well. Hence after
$\Theta(\sqrt n)$ evolution time there is $1-o(1)$ probability of finding the
marked node $w$. Alternative derivation can be also found in
\cite{wong2015nonlinear}.

Note that after $\pi \sqrt n$ evolution time the Hamiltonian will transform
$\ket s$ to $\approx \ket s$, which gives $\order{ 1/n}$ probability of finding
the node. Hence determining the measurement time in complexity is not enough to
guarantee high success probability. The situation is even worse for jumping
rate, as only $\gamma(1+\order{1/\sqrt n})$ jumping rates would give the same
result. Otherwise the diagonal elements will disturb the evolution
\cite{novo2015systematic,chakraborty2016spatial,chakraborty2020optimality}.

To determine the actual complexity of quantum search algorithms, the cost of
measurement and preparing the initial state \cite{magniez2007quantum,
	chakraborty2020optimality} should be taken into account. In this dissertation we
assume that the time required for state preparation and measurement is
significantly smaller compared to the evolution time. This approach is
frequently used in the literature
\cite{childs2004spatial,glos2018optimal,glos2018vertices,philipp2016continuous,chakraborty2020optimality,chakraborty2020finding,chakraborty2016spatial}. It is also common to choose different initial state \cite{chakraborty2020optimality}, which may depend on a chosen graph matrix $H_G$, but not marked vertex $w$.

Since we focus on the complexity of the search, we will be satisfied
with obtaining $\Theta(1)$ success probability. Note that repeating the
procedure of quantum search exponentially decrease the probability of failure.

Finally let us recall recent results concerning quantum search on general
graphs. In \cite{chakraborty2020finding} authors proposed alternative
continuous-time walk search, which was faster compared to any discrete random
search. However, the proposed method required quadratically large linear space
compared to the proposed CTQW-based search \cite{childs2004spatial}. Finally in
\cite{chakraborty2020optimality} authors proposed very general results concerning
the optimality of the CTQW-based quantum search. However in this case full
eigendecomposition of $H_G$ has to be known in order to apply their results,
which in general is a computationally difficult problem compared to quantum or
even random search. Finally, the case of \ER graph was considered in \cite{chakraborty2016spatial,chakraborty2017optimal,cattaneo2018quantum}

\section{Numerical analysis and tools}

In this section we describe a numerical procedure for determining exponent $\alpha$ for function $f=\Theta(n^\alpha)$. Furthermore, in order to simplify the reproduction of our numerical results, we described tools used in our numerical experiment and provide the link where the code can be found.

\subsection{Exponent estimation}\label{sec:exponent-estimation}

Suppose $f(t) = Ct^{\alpha \pm \varepsilon}$. Then we have
\begin{equation}
\log(f(t)) = (\alpha \pm \varepsilon)\log(t)+\log(C+o(1)).
\end{equation}
Note $\alpha $ is a slope of $\log(\mu_2(t))$ versus $\log(t)$. Several
approaches can be proposed in order to estimate $\alpha$. First, one can
estimate the slope of linear regression of pairs $(\log t_i,\log f(t_i))$. This
approach was satisfactory in case of walk search analysis, however for
propagation the values for small $t$ disturbed the actual value.

In case of walk propagation, the estimation of $\alpha$ goes as follows. For
time-points $t_1,\dots,t_k$ we calculate $f(t_1),\dots,f(t_k)$. We choose batch
size $l \ll k$ and calculate the slope $\alpha_i$ based on $(t_i,f(t_i)),
\ldots, (t_{i+l-1},f(t_{i+l-1}))$ in a way described in previous paragraph. Thus
we obtained $k-l+1$ approximations of $\alpha$. It is expected that the larger
the values of time-points are, the better the estimation of $\alpha$ is. For
convenience we choose $(t_{i+l-1}+t_i)/2$ to be the time-point corresponding to
the estimated $\alpha_i$ when plotting the results.

\subsection{Software used} 

Numerical analysis presented in this dissertation was generated with Julia programming
language version 1.5.2~\cite{bezanson2012julia}. The simulation was done using
in particular Expokit.jl \cite{expokit}, LightGraphs.jl \cite{lightgraphs} and
QSWalk.jl \cite{glos2019qswalk}. The latter package is the package co-developed
by the author of this dissertation. The code is available on GitHub under
\url{https://doi.org/10.5281/zenodo.4548423}.

%\subsection{Random and quantum walk comparison}
%\todo[inline]{general/where to find}
%\subsubsection{Ballistic propagation}
%\subsubsection{Quantum search}

%%%%%%%%%%%%%%%%%%%%%%%%%%%%%%%%%%%%%%%%%%%%%%%%%%%%%%%%%%%%%%%%%%%%%%%%%%%
%%%%%%%%%%%%%%%%%%%%%%%%%%%%%% QSW %%%%%%%%%%%%%%%%%%%%%%%%%%%%%%%%%%%%%%%%
%%%%%%%%%%%%%%%%%%%%%%%%%%%%%%%%%%%%%%%%%%%%%%%%%%%%%%%%%%%%%%%%%%%%%%%%%%%

%\include{propagation_QSW/propagation_qsw}

% !TeX spellcheck = en_US
\newcommand{\nonmoral}[1]{\underline{#1}}
\chapter{Non-moralizing Quantum Stochastic Walk} \label{sec:nonmoralizing-qsw}
\chaptermark{Non-moralizing QSW}

% !TeX spellcheck = en_US
%\chapter{Propagation of Quantum Stochastic Walk}\label{sec:prop-qsw}
Despite the evolution formula for CTRW and CTQW are very similar, the stochastic
and quantum evolution exhibits very different properties. For random walks,
based on the  Perron-Frobenious theorem, the evolution converges to a fixed
distributions for connected undirected graphs. In the case of quantum walk, we
observe quasi-periodic evolution.
\begin{lemma}[\cite{montanaro2007quantum}]
Let $\ket{\psi}\in \C^X$ be vector of a finite-dimensional linear space and let
$U\in\C^{X\times X}$ be a unitary matrix. For any $\varepsilon>0$ there exists
$n\geq 1$ such that
\begin{equation}
	|\bra \psi U^n\ket \psi |>1-\varepsilon.
\end{equation}
\end{lemma}
Based on the lemma we can see that by choosing proper $n\geq 1$ we can be
arbitrarily close to the initial state. The proof of the lemma was based on the
theorem provided in \cite{bocchieri1957quantum} for continuous-time evolution.

The lemma shows why it is complicated to define a quantum walk formula for the
directed graphs. Traditionally the space of the walk is splitted into orthogonal
subspaces, each attached to a different vertex. However based on the lemma
above, if we start at the sink vertex $v$, then based on the graph topology we
cannot move outside the sink because of the graph topology. Algebraically, it
means that $\bra v U \ket v=1$. However this means that $\bra v U \ket w=0$ for
any choice of $w\neq v$. From this we can see that one cannot amplify amplitude
on a sink vertex with the unitary evolution. For continuous-time evolution,
another obstacle is that Hamiltonian has to be a Hermitian matrix.

In this chapter we overcome the limitation be utilizing stochastic quantum
walks.  However, this kind of walk does not obey ballistic
propagation (see Fig.~\ref{fig:local-vs-global-interaction-global}). Hence
despite the fact that the LQSW may be a well-defined evolution preserving the
digraph topology (which in fact we will prove in the next Chapter), it is does
not obey super-diffusive regime. However, the quantum stochastic walk is  not
limited to Lindblad operators of the form $\ketbra{v}{w}$
\cite{whitfield2010quantum}.

\begin{figure}\centering
	\subfloat[\label{fig:local-vs-global-interaction-local}LQSW]{\includegraphics{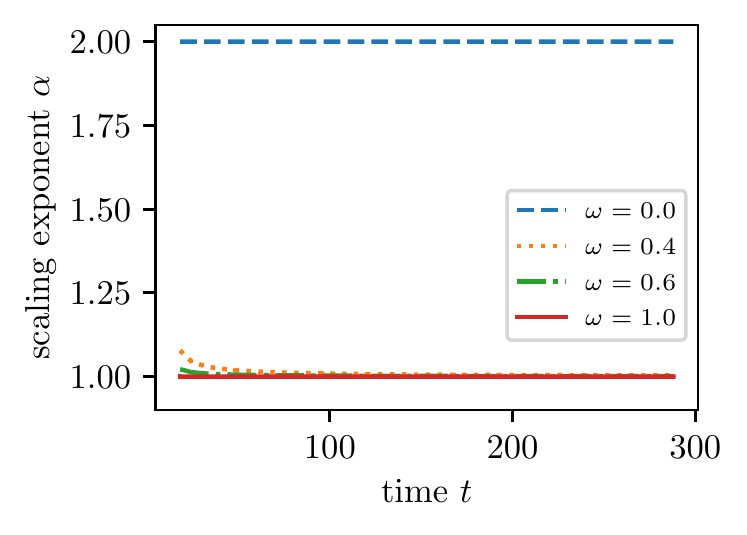}}
	\subfloat[\label{fig:local-vs-global-interaction-global}GQSW]{\includegraphics{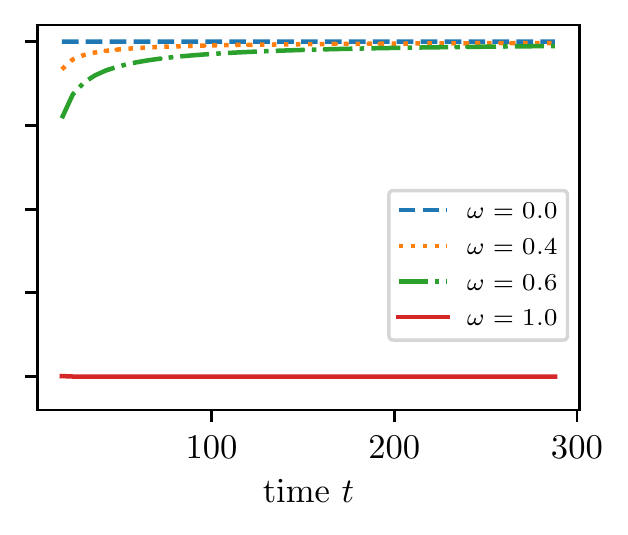}}
	
	\caption{Numerical investigation of the propagation for \protect\subref{fig:local-vs-global-interaction-local} interpolated standard GQSW  and \protect\subref{fig:local-vs-global-interaction-global} interpolated standard LQSW. We can se that for $\omega <0.75$ scaling exponent approaches 2, while decreases in time. The calculations were made for interpolating parameters $\omega=0,0.4,0.6,1$ for path graph with time-points $t=6,12,\ldots,300$. }\label{fig:local-vs-global-interaction}
\end{figure}

\begin{definition}
Let $\vec G = (V, \vec E)$ be a digraph and $G=(V,E)$ be its underlying graph.
Let $H\in \C^{V\times V}$ be a hermitian operator, and $\mathcal L =
\{L_1,\dots,L_k\} $ with  $L_i\in\C^{V\times V}$ of the form
\begin{gather}
H= \sum_{\{v,w\}\in E} c_{\{v,w\}} \ketbra{v}{w} + \bar c_{\{v,w\}} \ketbra{v}{w}, \\
L_i = \sum_{(v,w)\in\vec E} c_{(v,w),i} \ketbra{w}{v},
\end{gather}
where $c_{\{v,w\}}, c_{(w,v),i}\in \C_{\neq0}$. Then we call a GKSL master
equation with Hamiltonian $H$ and proposed Lindblad operators a \emph{global
	environment interaction QSW} (GQSW).
\end{definition}

We propose equivalent standard and interpolated GQSW. If $H$ is the adjacency
matrix of $G$ and $L$ is the adjacency matrix of $\vec G$, then we call a GQSW
with $H$ and $\{L\}$ a \emph{standard GQSW}. We define an interpolated GQSW
analogically to an interpolated LQSW.

GQSW is a quantum walk based on open system evolution, but yielding different
evolution than LQSW (see Fig.~\ref{fig:qsw-evolution}). As we will see it 
obeys a ballistic propagation. However this is achieved at cost of introducing extra edges not
existing even in the underlying graph. Thus in this chapter we start with
analytical justification of why the ballistic propagation for the GQSW. Next we
propose a modification which preserves at least superdiffusive propagation and
removed the undesired edges. The new quantum walk model will be called
\emph{nonmoralizing quantum stochastic walk}.

\begin{figure}[t]\centering 
	\includegraphics{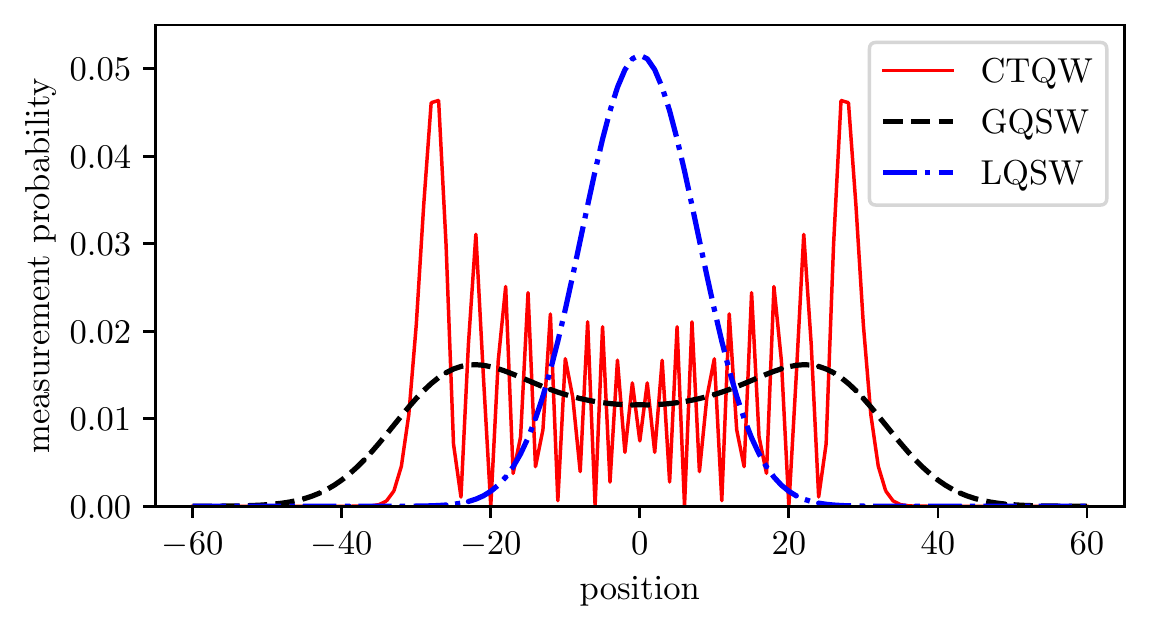}
\caption{Probability distribution of CTQW with Hamiltonian being adjacency
	matrix (red line), standard LQSW (black dashed line), and standard GQSW (blue
	dash-dotted line). We can see that the LQSW is highly concentrated around 0,
	contrary to other models. The computation were made for evolution time $t=15$
	and path graph with 121 vertices. The evolution started in state
	$\ketbra{0}$.}\label{fig:qsw-evolution}
\end{figure}

\section{Ballistic propagation for GQSW} \label{sec:prop-qsw}

In this section we will show that the GQSW model obeys ballistic propagation.
Let us consider the GQSW on path graph $P_n=([n],E_n)$, where for $i,j\in [n]$
we have $i,j\in E_n\iff |i-j| =1$. Let $A$ be an adjacency matrix of $P_n$. Then
the GQSW takes the form.
\begin{equation}
\frac{\dd \varrho}{\dd t} = -\ii (1-\omega) [A, \varrho] + \omega \left( A \varrho 
A^\dagger - 
\frac12 \{A ^\dagger A, \varrho\} \right). \label{eq:stochastic-transition}
\end{equation}
We start the evolution in $\varrho(0) = \ketbra{0}$. Based on numerical
results presented in Fig.~\ref{fig:local-vs-global-interaction-global}, we
can see that for $\omega \neq 1$ the scaling exponent approaches~2.

Below we will present the analytical derivation confirming our numerical
investigations. Let $\mu_m(t)$ be the $m$-th central moment of the probability
distribution of the computational measurement of $\varrho(t)$. We will derive
$\alpha$ for which $\mu_2(t)/t^\alpha\to c\neq 0$.

The proofs consists of several parts. First we derive the probability
distribution $p_t$ fo $\varrho(t)$ defined over path graph $P_n$. Since the adjacency matrix of the path graph is a Toeplitz matrix, its eigendecomposition is known, which in turns gives us the following theorem.
\begin{theorem}\label{theorem:prob-on-path}
	Given a interpolated standard GQSW on a path graph of order $n$  with initial state 
	$\varrho(0)=\dyad{l}$ for some $l\in\{1,\dots,n\}$, the quantum state $\varrho(t)$ at time $t$ satisfies  
	\begin{equation}\small
	\begin{split}
	\bra{k}\varrho(t)\ket{k}=&\left(\frac{2}{n+1}\right)^2 \sum\limits_{i,j=1}^n 
	\sin\left(\frac{ki\pi}{n+1}\right)
	\sin\left(\frac{kj\pi}{n+1}\right)\sin\left(\frac{li\pi}{n+1}\right)
	\sin\left(\frac{lj\pi}{n+1}\right)\times\\
	&\times 
	\exp\left[- \frac t2 \omega(\lambda_i - \lambda_j)^2\right]
	\exp\left[- \ii t (1-\omega) (\lambda_i-\lambda_j) 	\right],
	\end{split}
	\end{equation}
	where
	\begin{equation}
	\lambda_i = 2\cos\left( \frac{i\pi}{n+1} \right).
	\end{equation}
\end{theorem}
The proof of the theorem can be found in Appendix~\ref{sec:prop-qsw-prob-dist}

By increasing the order of the graph we obtain the limit of the probability
distribution. Note that in order to consider $n\to\infty$, we consider path
graph with $2n+1$ vertices labeled by $-n,\dots,n$, and we start at vertex $0$.
we determine the probability distribution of the interpolated standard GQSW on
infinite path by taking $n\to\infty$. The limiting probability distribution in
$n\to\infty $ takes the following form.
\begin{theorem}\label{theorem:prob-on-line}
Given an interpolated standard GQSW on an infinite path with initial state
$\varrho(0)=\dyad{0}$, the quantum state $\varrho(t)$ at time $t$ satisfies
\begin{equation}
\begin{split}
\bra{k}\varrho(t)\ket{k}&=\frac{1}{4\pi^2}\int_{-\pi}^\pi\int_{-\pi}^\pi 
\cos(kx)\cos(ky)\exp\left[-2\omega t (\cos(x)-\cos(y))^2\right] \\ 
&\phantom{=\ } \times \exp\left[ 2\ii t (1-\omega)(\cos(x)-\cos(y))\right] \dd 
x\dd y.
\end{split}
\end{equation}
\end{theorem}
The proof can be found in Appendix~\ref{sec:probability-infinite-path}.

Despite we have not found a simple analytical formula for the integral above, we
were able to find its Taylor expansion of the formula above. By proper summing
it turns out that the moments are simply a polynomials in $t$, which provides
the values of their scaling exponent directly. The results are concluded with
following theorem.
\begin{theorem}\label{theorem:stochastic-mixed-moments}
	Given a interpolated standard GQSW on an infinite path with initial state
	$\varrho(0)=\dyad{0}$, for odd $m$ the $m$-th moment equals zero. For even $m$
	we have
	\begin{equation}
	\mu_m(t) = \begin{cases}
	\binom{m}{\frac{m}{2}}(\omega-1)^m (1+o(1))t^m , &\omega \in [0,1) ,\\
	\frac{m!}{(m/2)!\,8^{m/2}}\binom{m}{m/2}(1+o(1)) t^{m/2}, &\omega \in 1. \\
	\end{cases}
	\end{equation}
\end{theorem}
The case $\omega >0$ is proven in Appendix~\ref{sec:prop-qsw-scaling-exp}, the $\omega=0$ is a standard CTQW evolution and the proof can be found in \cite{konno2005limit}. Formally in \cite{konno2005limit} author shown that for even $m=2k$
\begin{equation}
\lim_{t\to\infty}\mu_m(t)/t^m\to \frac{(2k-1)!!}{(2k)!!} = \frac{(2k)!}{((2k)!!)^2} = \binom{2k}{k}\frac{1}{2^{2k}}= \binom{m}{m/2}\frac{1}{2^{m}}.
\end{equation}
The difference in factor $1/2^m$ results from time rescaling $t\leftarrow \frac{t}{2}$.

As we mentioned before, the $m$-th moments are in fact polynomials in $t$ and we were able to find closed formula. The formula can be used to determine the precise form of $\mu_m(t)$ for any $\omega$. In particular for $\omega \in (0,1)$ we have $\mu_2(t) = 2\omega t +2(1-\omega)^2t^2$.

Thanks to the theorem above we can confirm the ballistic propagation of
interpolated standard GQSW with the following theorem, which is the main results
of this section.
\begin{theorem}
	Let $\alpha(\omega)$ be the scaling exponent of interpolated standard GQSW on
infinite path with initial state $\varrho(0)=\dyad{0}$, given the transition
parameter $\omega\in[0,1]$. Then
	\begin{equation}
	\alpha(\omega) = \begin{cases}
	2, & \omega \in [0,1),\\
	1, & \omega = 1.
	\end{cases}
	\end{equation}
\end{theorem}

\section{Spontaneous moralization in GQSW}

In previous section we considered a GQSW and we shown, that
despite of the presence of dissipation evolution, we observe the ballistic
propagation. At the same time Lindblad operators do not need to be symmetric, thus GKSL master equation is a natural choice for evolution model of
fast quantum walk definable on arbitrary directed graph.

Let $\vec G= (V,\vec E)$ be a directed graph and $G=(V,E)$ be its underlying
graph. Let $\vec A$ be an adjacency matrix of the digraph and $A$ be an
adjacency matrix of the underlying graph. Let us start with the interpolated
standard GQSW of the form
\begin{equation}
\frac{\dd \varrho}{\dd t} = -\ii (1-\omega) [A, \varrho] + \omega \left( \vec A \varrho
\vec A^\dagger - \frac12 \{\vec A ^\dagger \vec A, \varrho\} \right).
\end{equation}
Note that for $\omega\in(0,1)$, we expect to observe the backward propagation.
This fact is unavoidable if we plan to utilize the coherent evolution for fast
propagation. Unexpectedly, this is not the only effect that can be observed.

Let us consider the graph presented in
Fig.~\ref{fig:minimal-example-moralized-graph} with interpolating parameter
$\omega=1$ case, i.e. evolution with dissipative part only. The Lindblad
operator takes the form
\begin{equation}
L=\vec A = \dyad{v_3}{v_1}+\dyad{v_3}{v_2} = \ket{v_3}\left( \bra{v_1} + \bra {v_2}\right).
\end{equation}
Let us consider the evolution starting at state $\varrho(0) = \dyad{v_1}$. Since
Hamiltonian is not present in the evolution and we expect that the Lindblad
operator will follow the digraph topology, one should expect that $\bra{v_2}
\varrho(t)\ket{v_2} =0$ independently of chosen $t \geq0$. This is not the case,
as calculations shows
\begin{equation}
\begin{split}
\varrho(t) &=  \frac{1}{4}  \left(e^{-t}+1\right)^2 \ketbra{v_1}{v_1}
+\frac{1}{4} \left(e^{-2t}-1\right) \ketbra{v_1}{v_2} \\
&\phantom{\ =}+\frac{1}{4} \left(e^{-2 t}-1\right) \ketbra{v_2}{v_1} + 
\frac{1}{4}\left(e^{-t}-1\right )^2\ketbra{v_2}{v_2} \\
&\phantom{\ =}+  e^{-t} \sinh (t) \ketbra{v_3}{v_3}.
\end{split}
\end{equation}
In a time limit we have
\begin{equation}
\lim\limits_{t\to\infty}\varrho(t) = \frac{1}{4}\ketbra{v_1}{v_1}  - 
\ketbra{v_1}{v_2} - \ketbra{v_2}{v_1} +\ketbra{v_2}{v_2}) + \frac{1}{2} 
\ketbra{v_3}{v_3},
\end{equation}
which gives as 1/4 probability for measuring the vertex $v_2$. 

Two explanations of the phenomena can be proposed. First, let us recall the
evolution operator $S$ given in Eq.~\eqref{eq:diff-qsw-operator}
\begin{equation}
S = -\ii \left(H \kron \Id - \Id \kron \bar H \right) +
\sum_{L} \left ( L \kron \bar L - \frac{1}{2} L^\dagger L \kron \Id - \Id
\kron L^\top \bar L \right ).
\end{equation}
Let us consider propagation from $\ketbra{a}{\alpha}$ to $\ketbra{b}{\beta}$. We
have
\begin{equation}
\begin{split}
\bra{b\beta}S\ket{a\alpha} &= \delta_{\alpha\beta} \bra a\left ( -\ii H -\frac12\sum_ L   L^\dagger L \right)\ket b + \delta_{ab} \bra \beta\left ( \ii H - \frac12 \sum_L L^\top \bar L\right )\ket \alpha \\
&\phantom{\ =}+ \bra b L \ket a \bra \beta \bar L \ket{\alpha}.
\end{split}
\end{equation}
In our case we consider only real valued operator, thus conjugation can be added or remove without change on the operator. Thus we replace $L^T$ with $L^\dagger$ and $\bar L$ with $L$ to get 
\begin{equation}
\begin{split}
\bra{b\beta}S\ket{a\alpha} &= \delta_{\alpha\beta} \bra a\left ( -\ii H -\frac12\sum_ L   L^\dagger L \right)\ket b \\
&\phantom{\ =}+ \delta_{ab} \bra \beta\left ( \ii H - \frac12 \sum_L L^\dagger L\right )\ket \alpha \\
&\phantom{\ =}+ \sum_L \bra b L \ket a \bra \beta L \ket{\alpha}.
\end{split}
\end{equation}
Let us simplify the equation part by ignoring the parts that fit the graph
topology. In particular note that $\bra b L \ket a \bra \beta L \ket{\alpha}$ is
nonzero if $(\alpha,\beta)\in \vec E$ and $(a,b)\in \vec E$ which fits the graph
topology. Furthermore, Hamiltonian has nonzero impact iff $\alpha = \beta$ and
$\{a,b\}\in E$. This introduces the backward propagation if only one direction
is allowed of complex graph. However, Hamiltonian part introduces ballistic
propagation as was shown in Sec.~\ref{sec:prop-qsw}, thus we allow such
not-along the graph propagation. Moreover, this does not explains our example,
where the Hamiltonian was absent.

Let us consider the remaining part
\begin{equation}
-\frac12\delta_{\alpha\beta}\sum_ L   \bra a  L^\dagger L \ket b + \frac12 \delta_{ab} \sum_L\bra \beta L^\dagger L\ket \alpha.
\end{equation}
We will consider only the first addend, the second can be considered
analogically. Note that $\delta _{\alpha\beta}$ implies that the problem appears only
for propagation from $\dyad{a}{\alpha}$ to $\dyad{b}{\alpha}$ (or
$\dyad{a}{\alpha}$ to $\dyad{a}{\beta}$). On the other hand for standard GQSW we
have
\begin{equation}
\begin{split}
\sum_ L   \bra a  L^\dagger L \ket b  &= \bra a \vec A^\dagger \vec A \ket b = \sum_{v\in V} \bra a \vec A^\dagger \ket v \bra v \vec A \ket b  \\
&= |\{v: (a,v)\in \vec E ,(b,v)\in \vec E\}|.
\end{split}
\end{equation}
Note that this is precisely the situation observed in our example, visualized in
Fig.~\ref{fig:minimal-example-moralized-graph}. Here $a=v_1$, $b=v_2$ and
$\alpha=\beta=v_3$. Thus an additional connection is introduced between every
vertices that have common child. While this would be acceptable in case of
undirected edges, as we would introduce extra connection within radius 2, it is
not acceptable in the case of directed graphs. Due to the similarities between
graph moralization that occurs in machine
learning~\cite{cowell2006probabilistic}, we call our phenomena \emph{spontaneous
	moralization} or simply moralization.

\begin{figure}[t]
	\centering
	\subfloat[Original  graph\label{fig:intuitive-graph}]{
		\begin{tikzpicture}[node distance=2.5cm,thick]
		\tikzset{nodeStyle/.style = {circle,draw,minimum size=2.5em}}
		
		\node[nodeStyle] (A)  {$v_1$};
		\node[nodeStyle] (C) [below right of=A] {$v_3$};
		\node[nodeStyle] (B) [above right of=C] {$v_2$};
		
		\tikzset{EdgeStyle/.style   = {thick,-triangle 45}}
		\draw[EdgeStyle] (A) to (C);
		\draw[EdgeStyle] (B) to (C);
		
		\end{tikzpicture}}
	\hspace{1cm}
	\subfloat[Moral graph\label{fig:true-graph}]{
		\begin{tikzpicture}[node distance=2.5cm,thick]
		\tikzset{nodeStyle/.style = {circle,draw,minimum size=2.5em}}
		
		\node[nodeStyle] (A)  {$v_1$};
		\node[nodeStyle] (C) [below right of=A] {$v_3$};
		\node[nodeStyle] (B) [above right of=C] {$v_2$};
		
		\tikzset{EdgeStyle/.style   = {thick,-triangle 45}}
		\draw[EdgeStyle] (A) to (C);
		\draw[EdgeStyle] (B) to (C);
		
		\tikzset{additionalEdgeStyle/.style = {}}
		\draw[additionalEdgeStyle] (A) -- (B);
		%  \draw[additionalEdgeStyle] (A) to [out=90,in=150,looseness=8] (A);
		%  \draw[additionalEdgeStyle] (B) to [out=90,in=30,looseness=8] (B);
		\end{tikzpicture}}
	
	\caption{\label{fig:minimal-example-moralized-graph}Visualization of \protect\subref{fig:intuitive-graph} a directed graph  and
		its \protect\subref{fig:true-graph} moral graph.}
\end{figure}

The first explanation provides the sufficient and necessary condition for when
the moralization occurs. Second, simpler explanation suggests how to correct
this effect. Let as consider again the example provided in
Fig.~\ref{fig:intuitive-graph}. The Lindblad operator takes the form
\begin{equation}
L =  \ket{v_3}\left(\bra {v_1} + \bra {v_2}\right).
\end{equation}
Note that arbitrary state of the form $\alpha \ket{v_1} + \beta \ket{v_2}$ have
a decomposition in the basis  $\{ \frac{1}{\sqrt 2} (\ket {v_1}  +\ket
{v_2}),\frac{1}{\sqrt 2} (\ket {v_1}  -\ket {v_2})\}$. Thus the Lindblad
operator projects that part related to $\frac{1}{\sqrt 2} (\ket {v_1}  +\ket
{v_2})$ onto $\ket{v_3}$, while leaves $\frac{1}{\sqrt 2} (\ket {v_1}  -\ket
{v_2})$ unchanged. We can observe this in a quantum state as well
\begin{equation}
\begin{split}
\varrho(t) &=  \frac{1}{2} \frac{\ket{v_1}-\ket{v_2}}{\sqrt 2}\frac{\bra{v_1}-\bra{v_2}}{\sqrt 2} +\frac{e^{- 2t}}{2}  \frac{\ket{v_1}+\ket{v_2}}{\sqrt 2}\frac{\bra{v_1}+\bra{v_2}}{\sqrt 2}\\
&\phantom{\ =} +\frac{1}{2}e^{-t} \dyad {v_1} + \frac{1}{2}e^{-t} \dyad {v_2} +  e^{-t} \sinh (t) \dyad{v_3}.
\end{split}
\end{equation}
We can see that the only part spanned by $\ket{v_1},\ket{v_2}$ that remains unchanged is the part connected to $\frac{1}{\sqrt{2}}(\ket{v_1}-\ket{v_2})$ as predicted. This suggests that we need to change the operator in such a way that whole subspace spanned by the parents for each vertex is projected on the child subspace.

\subsection{Spontaneous moralization removal}

In this section we provide new QSW model called \emph{nonmoralizing global interaction
	 QSW} (NGQSW). Similarly to GQSW we construct single Lindblad
operator which describes the structure of directed graph. However in this model
we will remove the undesired moralization effect.

Let $\vec G=(V,\vec E)$ be a directed graph. We will construct new directed
graph $\vec {\nonmoral G}=( \nonmoral V, \vec{ \nonmoral{E}})$ homomorphic to
$\vec G$.  For consistency every graph object or operator that is connected to
$\vec {\nonmoral{G}} $, and thus nonmoralizing evolution, will be underlined as $\nonmoral{(\cdot )}$.

Let for each $v\in V$ define $ \nonmoral {V}_v=\{ \nonmoral v^0\dots,  \nonmoral v^{\indeg (v)-1}\}$ iff $\indeg(v)>0$. If $\indeg (v) =0$, then let $\nonmoral V_v=\{\nonmoral v^0\}$. We choose
\begin{equation}
\nonmoral V = \bigcup_{v\in V}\nonmoral V_{v}.
\end{equation} 
Furthermore, let $\vec{\nonmoral E}_{vw} = \{(\nonmoral v, \nonmoral w): \nonmoral v \in \nonmoral V_v,\nonmoral w \in \nonmoral V_w \}$. Then 
\begin{equation}
\vec {\nonmoral E} = \bigcup_{(v,w)\in E} \nonmoral {\vec E}_{vw}.
\end{equation}
Note that $f:\nonmoral V \to V$ of the form $f:\nonmoral v \mapsto v$ is a
proper homomorphism from $\vec {\nonmoral G}$ to $\vec G$. We  call
$\vec{\nonmoral G}$ \emph{demoralizing graph} and $f$ a \emph{natural
	homomorphism} from $\nonmoral{\vec G}$ to $\vec G$. Note that $f$ is also a
proper homomorphism from  underlying graphs of $\nonmoral{G}$ to $G$.

Let $\nonmoral\varrho\in \C^{\nonmoral{V}\times \nonmoral{V}}$ be a mixed state.
We define a natural measurement in terms of the vertices of the original graph $\vec
G$ as
\begin{equation}
p(v) = \sum_{\nonmoral v\in \nonmoral V_v} \bra {\nonmoral v} \nonmoral\varrho \ket{\nonmoral v}.
\end{equation}

Our goal is to define QSW on $\vec {\nonmoral G}$ that will simulate the
nonmoralizing evolution on $\vec G$. We cannot use simply an adjacency matrix of
$\vec {\nonmoral G}$ as it will lead to exactly the same moralization effect.
The solution is to choose the orthogonal vectors for columns $\nonmoral L$ in
such a way that $\bra {\nonmoral w} \nonmoral L^\dagger \nonmoral L
\ket{\nonmoral v}=0$ and $\bra{\nonmoral w} \nonmoral L \ket{\nonmoral v}=0$ for
$(v,w),(w,v)\not\in \vec E$.

\begin{lemma} \label{theorem:nonmoral-lindblad-construction}
	Let $\vec G=(V,\vec E)$ be a directed graph and $\vec {\nonmoral G}$ be its
demoralizing graph. Let $L_v\in \C^{\nonmoral V_v\times P(v)}$ be such a matrix
that columns are orthogonal, i.e. $\bra {u}L_v^\dagger L_v \ket {w} =0$ for
$u\neq w$, where $u,w\in P(v)$. Let
	\begin{equation}
		\bra{\nonmoral v^k}\nonmoral L\ket{\nonmoral w^l} \coloneqq \begin{cases}
			\bra {\nonmoral{v}^k} L_v \ket{w}, & (w,v)\in \vec E,\\
			0, & \rm otherwise.
		\end{cases}
	\end{equation}
	Then for arbitrary $\nonmoral v, \nonmoral w$ such that $v\neq w$ we have $\bra{\nonmoral v}\nonmoral L^\dagger \nonmoral L \ket{\nonmoral w}=0$.
\end{lemma}
\begin{proof}
	Let $v,w\in V$ be arbitrary different vertices and $\nonmoral v^k$ and $\nonmoral w^l$ respective vertices from $\vec{\nonmoral G}$. We have
	\begin{equation}{\label{eq::mor}}
	\begin{split}
	\bra{\nonmoral v^k}\nonmoral L^\dagger \nonmoral L\ket{\nonmoral w^l} &= \sum_{\nonmoral u\in \nonmoral
	V} \bra{\nonmoral v^k}\nonmoral L^\dagger  \ket {\nonmoral u}\bra{\nonmoral u}\nonmoral
L\ket{\nonmoral w^l}\\ 
	&=\sum_{u \in V} \sum_{\ \nonmoral u^i\in \nonmoral V_u}\overline
{\bra{\nonmoral u^i}\nonmoral 	L\ket{\nonmoral v^k}} \bra{\nonmoral u^i}\nonmoral L\ket{\nonmoral
	w^l}\\ 
	&= \sum_{u\in C(v)\cap C(w)}  \  \sum_{\ \nonmoral u^i \in
	\nonmoral V_u} \overline {\bra {\nonmoral u^i} L_{u}\ket v} \bra {\nonmoral u^i} L_{u} \ket w.\\ 
	&=
\sum_{u\in C(v)\cap C(w)}\  \sum_{\ \nonmoral u^i \in \nonmoral V_u}\ 
\bra v L_{u}^\dagger\ket {\nonmoral u^i} \bra {\nonmoral u^i} L_{u} \ket w.\\ 
	&= \sum_{u\in C(v)\cap C(w)} \bra v L_{u}^\dagger
L_{u} \ket w.\\
	\end{split}
	\end{equation}
	Since $v\neq w$ we have $\bra {v}L_u^\dagger L_u \ket w=0$, which ends the
proof.
\end{proof}
%Note that $L_v$ is still a square matrix for $|P(v)|\geq1$.

A simple remark of the lemma is that the $\nonmoral L^\dagger \nonmoral L$
impact on the evolution between the vertices is removed. Note that $\nonmoral
L^\dagger \nonmoral L$ part still have impact on the internal evolution. However,
this does not imply that the amplitude is transferred not along the digraph, only within the space defined for a single vertex in $\vec G$.

Let us recall the example the example given in Fig.~\ref{fig:intuitive-graph}.
Its nonmoralizing version is presented in Fig.~\ref{fig:nonmoral-example-graph}.
We chose $L_v$ to be Fourier matrices $F_{|\nonmoral{V}_v|}\in
\C^{|\nonmoral{V}_v|\times |\nonmoral{V}_v|}$, \ie
\begin{equation}
\bra j F_{|\nonmoral{V}_v|} \ket k  = \omega_{|\nonmoral{V}_v|}^{jk},
\end{equation}
where $\omega_n = \exp(2\pi \ii / n)$. The Lindblad operator takes the form
\begin{equation}
\nonmoral L = \ketbra{\nonmoral v_3^0}{\nonmoral v_1^0}+\ketbra{\nonmoral v_3^1}{\nonmoral
	v_1^0}+\ketbra{\nonmoral v_3^0}{\nonmoral v_2^0}- \ketbra{\nonmoral v_3^1}{\nonmoral v_2^0}.
\end{equation}
One can verify that, if $\nonmoral \varrho(0)=\dyad{\nonmoral v_1^0}$, then the
state takes the form
\begin{equation}
\nonmoral \varrho(t) = e^{-2t} \dyad{\nonmoral v_1^0} + \frac{1}{2} 
\left(1-e^{-2 t}\right)  (\ket{\nonmoral v_3^0}+\ket{\nonmoral v_3^1})(\bra{\nonmoral
	v_3^0}+\bra{\nonmoral v_3^1}).
\end{equation}
Note that for arbitrary $t \geq 0$ we have $\bra{\nonmoral{v}_2^0} \nonmoral
\varrho(t) \ket {\nonmoral v_2^0}=0$, which confirms that our correction scheme
fixes the moralization effect.

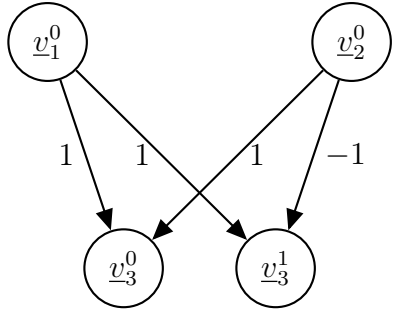
\begin{figure}
	\centering
	
	\begin{tikzpicture}[node distance=3.5cm,thick]
	\tikzset{nodeStyle/.style = {circle,draw,minimum size=2.5em}}
	
	\node[nodeStyle] (A)  at (0,0) {$\nonmoral v_1^0$} ;
	\node[nodeStyle] (C0) at (1, -3) {$\nonmoral v_3^0$};
	\node[nodeStyle] (C1) at (3,-3) {$\nonmoral v_3^1$};
	\node[nodeStyle] (B) at (4,0) [] {$\nonmoral v_2^0$};
	
	\tikzset{EdgeStyle/.style   = {thick,-triangle 45}}
	\draw[EdgeStyle] (A) to node[left] {1} (C0);
	\draw[EdgeStyle] (B) to node[right] {1}  (C0);
	\draw[EdgeStyle] (A) to node[left] {1} (C1);
	\draw[EdgeStyle] (B) to node[right] {$-1$} (C1);
	
	\end{tikzpicture}
	\caption{Visualization of the Lindblad operator in new model. The original
		graph is presented in
		Fig.~\ref{fig:intuitive-graph}.}\label{fig:nonmoral-example-graph}
\end{figure}

\subsection{Premature localization} Through numerical analysis of  the
introduced model we noticed undesired phenomenon: premature localization. For
example, let us consider the graph presented in
Fig.~\ref{fig:premature-localization-example-original}. In our model the
Lindblad operator will represent the graph presented in
Fig.~\ref{fig:premature-localization-example-increased}. Its Lindblad operator
$\nonmoral L$ has the form
\begin{equation}
\nonmoral L = \begin{bmatrix}
0 & 1 & 1 & 0 & 0 & 1 & 1\\
1 & 0 & 1 & 0 & 1 & 0 & 1\\
1 & 1 & 0 & 0 & 1 & 1 & 0\\
1 & 0 & 0 & 0 & 1 & 0 & 0\\
0 & 1 & -1 & 0 & 0 & 1 & -1\\
1 & 0 & -1 & 0 & 1 & 0 & -1\\
1 & -1 & 0 & 0 & 1 & -1 & 0
\end{bmatrix},
\end{equation} 
with order $\nonmoral v_1^0, \nonmoral v_2^0, \nonmoral v_3^0, \nonmoral v_4^0,
\nonmoral v_1^1, \nonmoral v_2^1,\nonmoral v_3^1$. It is expected that starting
from arbitrary proper mixed state (at least in a vertex, \ie{} $\nonmoral
\varrho(0)=\ketbra{\nonmoral v_i^j}{\nonmoral v_i^j}$), we should obtain
$\varrho(\infty)\coloneqq\lim_{t\to\infty}\nonmoral \varrho(t) =
\ketbra{\nonmoral v_4^0}{\nonmoral v_4^0}$. Oppositely, through numerical
simulation one can verify that
\begin{equation}
\nonmoral\varrho(\infty)= \frac{1}{16}\begin{bmatrix}
5 & 1 & 1 & 0 & -5 & -1 & -1\\
1 & 1 & 1 & 0 & -1 & -1 & -1\\
1 & 1 & 1 & 0 & -1 & -1 & -1\\
0 & 0 & 0 & 2 & 0 & 0 & 0\\
-5 & -1 & -1 & 0 & 5 & 1 & 1\\
-1 & -1 & -1 & 0 & 1 & 1 & 1\\
-1 & -1 & -1 & 0 & 1 & 1 & 1
\end{bmatrix},
\end{equation}
is a stationary state of the evolution, when starting from a state
$\nonmoral \varrho(0) = \ketbra{\nonmoral v_1^0}{\nonmoral v_1^0}$.

\begin{figure}
	\centering
	\subfloat[\label{fig:premature-localization-example-original}]{
		\begin{minipage}[b]{.35\linewidth}
			\centering
			\begin{tikzpicture}[node distance=2.5cm,thick]
			\tikzset{nodeStyle/.style = {circle,draw,minimum size=2.5em}}
			
			\node[nodeStyle] (A)  {$v_1$};
			\node[nodeStyle] (C) [above right of=A] {$v_3$};
			\node[nodeStyle] (D) [below right of=A] {$v_4$};
			\node[nodeStyle] (B) [above right of=D] {$v_2$};

			\tikzset{EdgeStyle/.style   = {thick,-triangle 45}}
			\draw[EdgeStyle] (A) to (D);
			%\draw[EdgeStyle] (B) to (D);
			\draw[thick] (A) to (C);
			\draw[thick] (A) to (B);
			\draw[thick] (B) to (C);
			\end{tikzpicture}
			
	\end{minipage}}
	\hspace{1cm}
	\subfloat[\label{fig:premature-localization-example-increased}]
	{\begin{tikzpicture}[node distance=2.5cm,thick]
			\tikzset{nodeStyle/.style = {circle,draw,minimum size=2.5em}}
			
			%zmieniać konsekwencję nazwy, to połączenia sam się ułożą
			\node[nodeStyle] (A0)  {$\nonmoral v_1^0$};
			\node[nodeStyle] (A1) [above of=A0]  {$\nonmoral v_1^1$};
			\node[nodeStyle] (C0) [above right of=A1] {$\nonmoral v_3^0$};
			\node[nodeStyle] (C1) [right of=C0] {$\nonmoral v_3^1$};
			\node[nodeStyle] (D) [below right of=A0] {$\nonmoral v_4^0$};
			
			\node[nodeStyle] (B1) [below right of=C1] {$\nonmoral v_2^1$};
			\node[nodeStyle] (B0) [below of=B1] {$\nonmoral v_2^0$};

			\tikzset{EdgeStyle/.style   = {thick,-triangle 45}};
			
			\draw[thick] (A0) to (B0);
			\draw[thick] (A0) to (C0);
			\draw[EdgeStyle] (A0) to (D);
			\draw[thick] (A0) to (B1);
			\draw[thick] (A0) to (C1);
			\draw[thick] (B0) to (C0);
			\draw[thick] (B0) to (A1);
			\draw[thick] (B0) to (C1);
			\draw[thick] (C0) to (A1);
			\draw[thick] (C0) to (B1);
			\draw[EdgeStyle] (A1) to (D);
			\draw[thick] (A1) to (B1);
			\draw[thick] (A1) to (C1);
			\draw[thick] (B1) to (C1);
			
			\end{tikzpicture}} 
		\caption{\label{fig:premature-localization-example}The original graph with the premature localization
		\ref{fig:premature-localization-example-original}, and  the graph $\nonmoral
		G$ based on the original graph
		\ref{fig:premature-localization-example-increased}. Starting from
		$\dyad{\nonmoral v_1^0}$ we obtain a stationary state which is
		not fully localised in the vertex $\nonmoral v_4^0$.}
\end{figure}
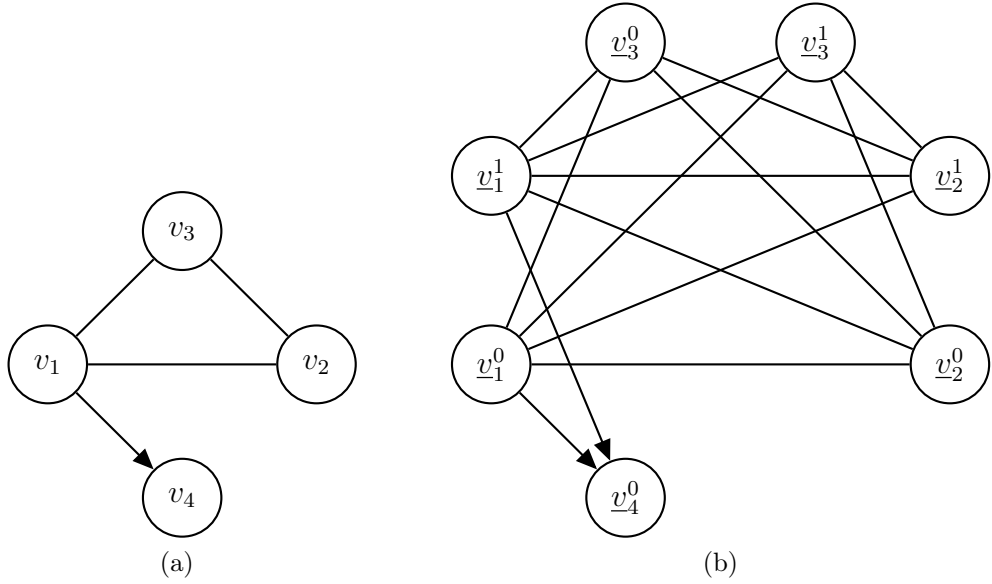

To correct this problem, we propose to add the Hamiltonian $\nonmoral
H_{\textrm{rot}}\in \C^{\nonmoral{V}\times\nonmoral{V}}$ which changes the state within the subspace corresponding to
single vertex, Let $H_{\textrm{rot},v}\in \C^{\nonmoral{V}_v\times \nonmoral{V}_v}$ be arbitrary Hamiltonian. Then $H_{\rm rot}$ will be a $\nonmoral{V}$-block diagonal operator
\begin{equation}
H_{\rm rot} = \left [\begin{array}{c|c|c|c}
H_{\textrm{rot},v_1} & & & \\\hline
& H_{\textrm{rot},v_2} & & \\\hline
&  & \ddots &  \\\hline
& & & H_{\textrm{rot},v_{|V|}}
\end{array}\right].
\end{equation} 
Now the evolution takes the form
\begin{equation}
\frac{\dd}{\dd t}\nonmoral \varrho =-\ii [\nonmoral H_{\textrm{rot}},\nonmoral  
\varrho] + 
\nonmoral L 
\nonmoral \varrho 
\nonmoral L^\dagger - \frac12 \{\nonmoral L^\dagger\nonmoral  L, \nonmoral \varrho\}.
\label{eq:GKSL-premature-localization}
\end{equation}
We call this Hamiltonian the \emph{locally rotating Hamiltonian}, since it acts
only locally on the subspaces corresponding to single vertex. We have verified
numerically that the appropriate Hamiltonian corrects the premature
localisation. In particular, if we choose the locally rotating Hamiltonian based
on
\begin{equation}
\bra{\nonmoral v^k}\nonmoral H_{\textrm{rot},v}\ket{\nonmoral v^l} = 
\begin{cases}
\ii, & l=k+1\mod \indeg (v_i), \\
-\ii, & l=k-1\mod \indeg (v_i), \\
0, & \textrm{otherwise,}
\end{cases}\label{eq:example_ham_rot}
\end{equation}
the evolution on the graph presented in Fig.
\ref{fig:premature-localization-example-increased} results in a unique
stationary state $\ketbra{\nonmoral v_4^0}{\nonmoral v_4^0}$.

Note that the quantum state of the form $\ketbra{\nonmoral v_4^0}{\nonmoral
	v_4^0}$ is a stationary state for any choice of locally-rotating Hamiltonian. This can be shown by
calculating $\frac{\dd \varrho}{\dd t}$ for the given state (we will elaborate
more on this problem in Chapter~\ref{sec:convergence-qsw}). However in case of a locally-rotating Hamiltonian being a zero matrix, we show that also other mixed states are proper stationary states of the evolution. 

To exemplify the impact of the Hamiltonians on the convergence, we analyzed the
spectrum of the evolution generator. Since there is always at least one
stationary state, it means that the multiplicity of 0 eigenvalue is at
least 1. Then the second smallest eigenvalue (ordered in absolute value)
gives us insight into the convergence property of the walk. Based on the results
 presented in Fig.~\ref{fig:loc-ham-analysis} we see that both random real- and
complex-valued rotating Hamiltonian with high probability produce high
separation between smallest eigenvalues. However, in the case of complex-valued
Hamiltonians far larger eigenvalues are obtained. Thus a random GUE Hamiltonian
may be a good candidate to be a block for the locally rotating Hamiltonian.

\begin{figure}\centering
\includegraphics{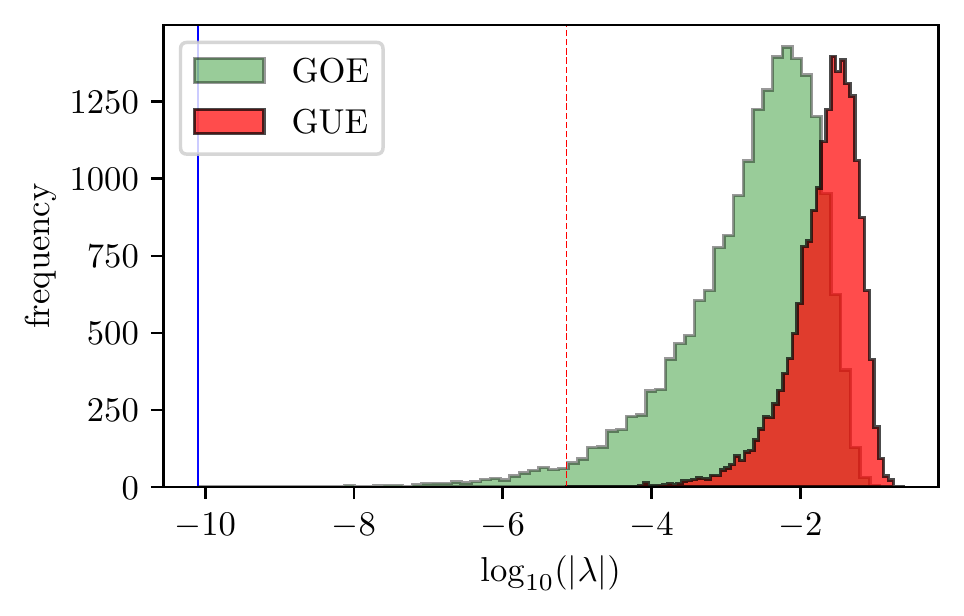}
\caption{\label{fig:loc-ham-analysis} The analysis of impact of the form of
	locally rotating Hamiltonians on the uniqueness of the stationary state for
	graph presented in Fig.~\ref{fig:premature-localization-example}. GOE (GUE)
	denotes evolution where each block of the Hamiltonian was sampled independently
	from GOE (GUE) distribution. The $\log_{10}(|\lambda|)$ the logarithm of the
	absolute value of the second smallest (in absolute value) eigenvalue of the
	evolution generator. For each GOE and GUE we sampled \num{20000} Hamiltonians.
	Vertical blue solid (red dashed) line denotes the minimum obtained value for GOE
	(GUE).}
\end{figure}

\subsection{Final model and correction cost} 
\label{sec:nonmoralizing-correction-cost}

Let us now sum up all considered corrections and define a nonmoralizing QSW. We
start with introducing a formal definitions of nonmoralizing operators. Let
$\vec G = (V, \vec E)$ be an arbitrary directed graph. Then
$\vec{\nonmoral{G}}=(\nonmoral{V}, \vec{\nonmoral{E}})$ will be a demoralizing
graph of $G$. Furthermore $G$ and $\nonmoral{G}=(\nonmoral{V}, \nonmoral{E})$
will be the underlying graphs of $\vec G$ and $\vec {\nonmoral G}$. Let $f$ will
be a natural homomorphism  from $\vec{\nonmoral{G}}$ to $\vec G$.

\begin{definition}
	Let $\vec G = (V, \vec E)$ be an arbitrary directed graph. Hermitian operator $\nonmoral{H}\in \C^{\nonmoral{V}\times \nonmoral{V}}$ is nonmoralizing Hamiltonian of QSW on $\vec G$ if for all $\nonmoral{v},\nonmoral{w}\in \nonmoral{V}$ we have
	\begin{equation}
	\{\nonmoral{v}, \nonmoral{w}\}\not\in \nonmoral E \iff \bra{\nonmoral{w}} \nonmoral H \ket{\nonmoral{v}} = 0.
	\end{equation}
\end{definition}
In other words, if vertices are not connected in the underlying graph of
$\nonmoral G$, then the corresponding  Hamiltonian element vanishes. Note that
amplitude transfer along the $\{v,w\}$ edge in $G$ corresponds to a transfer
between $|\nonmoral{V}_v|$-dimensional and $|\nonmoral{V}_w|$-dimensional space. Thus, for each edge we have $2|\nonmoral{V}_v|\cdot |\nonmoral{V}_w|$ real
degrees of freedom for each edge $e\in E$. For LQSW and GQSW we have two real
degree of freedom per edge. We say that a nonmoralizing Hamiltonian is standard if
$\{\nonmoral{v},\nonmoral{w}\}\in \nonmoral{E}$ implies $\bra{\nonmoral
	w}H\ket{\nonmoral{v}} =1$. Note that in $\nonmoral{V}$-block representation, if
blocks correspond to non-connected vertices in $V$, then the block is a zero
matrix.

\begin{definition}
	Let $\vec G = (V, \vec E)$ be arbitrary directed graph. An operator
$\nonmoral{L}\in \C^{\nonmoral{V}\times \nonmoral{V}}$ is nonmoralizing
Lindblad operator of QSW on $\vec G$ if for all $\nonmoral{v},\nonmoral{w}\in
\nonmoral{V}$ we have
	\begin{equation}
	(\nonmoral{v}, \nonmoral{w})\not\in \nonmoral {\vec E} \iff \bra{\nonmoral{w}} \nonmoral L \ket{\nonmoral{v}} = 0,
	\end{equation}
	and for any $\nonmoral v,\nonmoral w\in \nonmoral{V}$ satisfying $f(\nonmoral{v})\neq f(\nonmoral{w})$ we have
	\begin{equation}
	\bra v \bar L^\dagger \bar L\ket w=0.	
	\end{equation}
\end{definition}
Note that Lemma~\ref{theorem:nonmoral-lindblad-construction} provides a simple
construction method. The method requires a mapping $v\mapsto L_v$, where
each $L_v$ is a matrix with pairwise orthogonal columns. While the number of
degrees of freedom depends on the choice of $\vec G$, clearly the free-parameter
space is larger comparing to LQSW and GQSW. For LQSW we have only one real
degree of freedom, for the GQSW for the Lindblad operator collection $\mathcal
L_{\rm GQSW}$ we have $|\mathbb L_{\rm GQSW}|$ degrees of freedom. We say that
the nonmoralizing Lindblad operator is standard if it is constructed according
to Lemma~\ref{theorem:nonmoral-lindblad-construction} with $L_v$ being a Fourier
matrix. Similarly as is for nonmoralizing Hamiltonian, in $\nonmoral{V}$-block
representation if there is no $(v,w)$ arc, then the block corresponding to this
arc is a zero matrix.

\begin{definition}
	Let $\vec G = (V, \vec E)$ be an arbitrary directed graph. Hermitian operator $\nonmoral{H}_{\rm rot}\in \C^{\nonmoral{V}\times \nonmoral{V}}$ is locally rotating Hamiltonian of QSW on $\vec G$ if we have
	\begin{equation}
	f(\nonmoral{v}) \neq f(\nonmoral{w}) \implies \bra{\nonmoral{w}} \nonmoral H_{\rm rot} \ket{\nonmoral{v}} = 0
	\end{equation}
	for all $\nonmoral{v},\nonmoral{w}\in \nonmoral{V}$.
\end{definition}
Note that locally rotating Hamiltonian introduces
$2|\nonmoral{V}_v|^2-|\nonmoral{V}_v|$ real degrees of freedom for each vertex
$v\in V$. We say that locally rotating Hamiltonian is standard iff block are
defined as in Eq.~\eqref{eq:example_ham_rot}.

Finally we define a nonmoralizing global interaction QSW as follows.
\begin{definition}
Let $\vec G = (V, \vec E)$ be an arbitrary directed graphs. Let $\nonmoral H$ be
a nonmoralizing Hamiltonian, $\nonmoral {\mathbb L}$ be a collection of
nonmoralizing Lindblad operators and $\nonmoral H_{\rm rot}$ be a locally
rotating Hamiltonian. Then the evolution
\begin{equation}
\frac{\textrm{d}}{\textrm{d}t}\nonmoral \varrho =-\ii [\nonmoral H +\nonmoral{H}_{\rm rot} ,\nonmoral \varrho]
+\sum_{\nonmoral L \in
	\nonmoral {\mathcal L}}\left (\nonmoral L \nonmoral \varrho \nonmoral L^\dagger -
\frac12 \{\nonmoral L^\dagger\nonmoral  L, \nonmoral \varrho\} \right ), \label{eq:nonmoral-qsw-def}
\end{equation}
is called a nonmoralizing global interaction QSW (NGQSW).
\end{definition}
 We say that a NGQSW is standard if all operators defining the evolution are standard.
 
In the previous section we introduced a interpolating parameter $\omega$, which was
responsible for adjusting a relative strength of closed- and open-system
evolution. Note that for the model defined above locally rotating Hamiltonian
should be introduced only in the presence of open-system evolution. Hence an
interpolated  NGQSW with the parameter $\omega$ will be of the form
\begin{equation}
\frac{\textrm{d}}{\textrm{d}t}\nonmoral \varrho =-\ii [ (1-\omega )\nonmoral H
+\omega\nonmoral{H}_{\rm rot} ,\nonmoral \varrho] +\omega \sum_{\nonmoral L \in
	\nonmoral {\mathbb L}}\left (\nonmoral L \nonmoral \varrho \nonmoral L^\dagger -
\frac12 \{\nonmoral L^\dagger\nonmoral  L, \nonmoral \varrho\} \right ).\label{eq:nonmoral-qsw-def-weigthed}
\end{equation}

The presented correction scheme enlarges the Hilbert space used. We can bound
from above the dimension of the constructed space. If the original graphs
consists of $n$ vertices, with indegree $\indeg(v)$ for vertex $v$, then the
dimension of enlarged Hilbert space equals $\sum_{v\in
	V}\indeg(v)+|\{v:\textrm{indegree}(v)=0\}|=|\vec
E|+|\{v:\textrm{indegree}(v)=0\}|$. Note that in the worst case scenario of a
complete digraph, the dimension of the enlarged Hilbert space is
$\order{|V|^2}$. In the term of number of qubits the additional qubit number is
$O(\log n)$, hence in our opinion the correction scheme is efficient. Comparing
to other models \cite{taketani2016physical}, where for each vertex there is
corresponding qubit, size of our Hilbert space is still small.

%Furthermore, our corrections scheme preserves the geometry of the graph, since
%we do not create long range interactions: $\nonmoral v_i^k$ and $\nonmoral
%v_j^k$ are connected iff $v_i$ and $v_j$ are connected or $v_i=v_j$. Hence from
%the geometry point of view, our model may not be much harder to implement
%comparing to original global environment interaction case.

\section{Propagation of standard NGQSW}

In Chapter~\ref{sec:prop-qsw} we have shown that standard GQSW yields a
ballistic propagation. However, because of the spontaneous moralization, the
graph which was actually analysed was an undirected line with additional edges
between every two vertices as in  Fig.~\ref{fig:line}. Hence, in this section we
analyze NMQSW, to verify if the fast propagation recorded in moralizing quantum
stochastic walk for the global interaction case is due to the additional
amplitude transitions or due to the global interactions.

\subsection{Lack of symmetry on infinite path} \label{sec:nonmoralizing-symmetry}

\begin{figure}
	\centering 
	\subfloat[\label{fig:nonsymmetry-line} before symmetrization]{
		\includegraphics{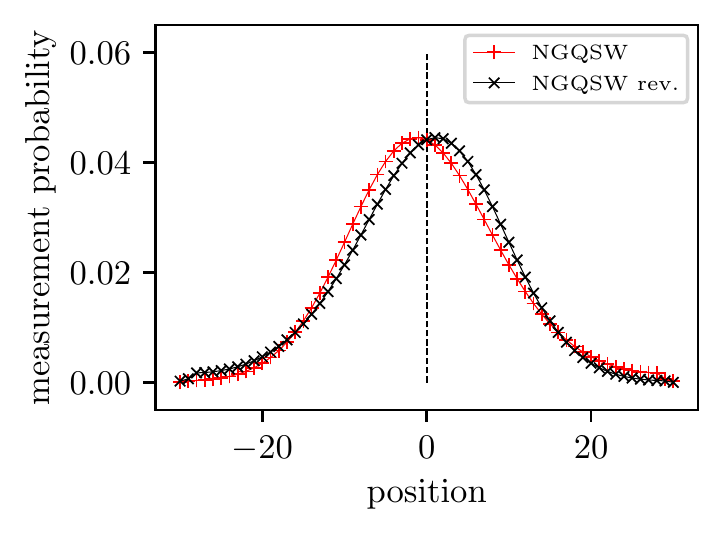}}
	\subfloat[\label{fig:symmetry-line} after symmetrization]{\includegraphics{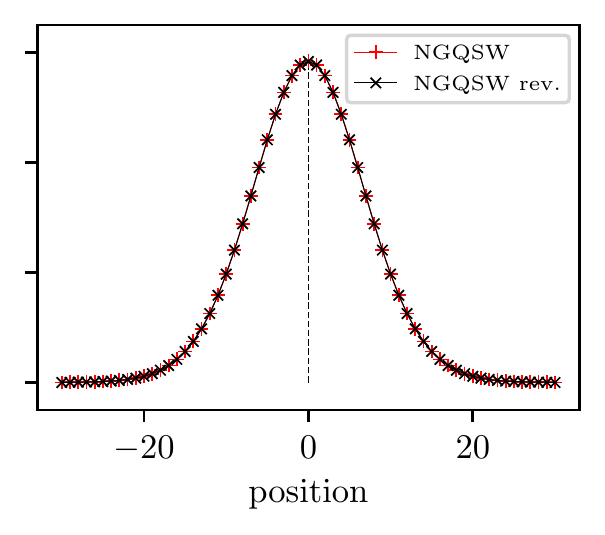}}
	
	\caption{ Probability distribution of the measurement in the standard basis and
	its reflection by the initial position for $t=100$ on the line segment of
	length 61, \protect\subref{fig:nonsymmetry-line}  before procedure application
	and \protect\subref{fig:symmetry-line} after procedure application described in
	Sec.~\ref{sec:nonmoralizing-symmetry}. In both cases we start in
	$\frac{1}{2}(\ketbra{\nonmoral v_0^0}{\nonmoral v_0^0}+\ketbra{\nonmoral
		v_0^1}{\nonmoral v_0^1})$. }
\end{figure}

Let us analyse the standard NGQSW on undirected path graph. A further undesired
effect  has occurred for some symmetric graphs, where we observe the lack of
symmetry of the probability distribution.  In Fig.~\ref{fig:nonsymmetry-line} we
present the probability distribution of the position measurement and the
reflection of distribution according to the initial position. We observe that
probability distribution is not symmetric with respect to the initial position.

Removing the locally rotating Hamiltonian does not remove the asymmetry and the
nonmoralizing Hamiltonian is a symmetric operator. Hence, the lack of symmetry
comes from asymmetry of nonmoralizing Lindblad operators. Since by construction
the columns of matrices $L_{v}$ need to be orthogonal, the matrices $L_v$ will
not be symmetric in ge`neral.

We propose to add another global interaction Lindblad operator, with different
$L_v$ matrices which will remove the side-effect. In the case of undirected
segment, we choose $\nonmoral {\mathbb L}=\{\nonmoral{L}_1,\nonmoral{L}_2\}$
defined through matrices $L_{v}^{(1)}$ and $L_{v}^{(2)}$ for the vertex $v$. For
$\nonmoral{L}_1$ we choose for each vertex
\begin{equation}
L_v ^{(1)}=\begin{bmatrix}
1 & 1 \\1 & -1
\end{bmatrix}
\end{equation}
and for $\nonmoral{L}_2$ we choose for each vertex 
\begin{equation}
L_v ^{(2)}=\begin{bmatrix}
1 & 1 \\ -1 & 1
\end{bmatrix}.
\end{equation}
The evolution takes the form
\begin{equation}
\frac{\textrm{d}}{\textrm{d}t}\nonmoral \varrho =-\ii [\nonmoral 
H_{\textrm{rot}},\nonmoral  \varrho] + 
\sum_{\nonmoral L \in \{\nonmoral L_1,\nonmoral L_2\}}\left (\nonmoral L 
\nonmoral \varrho 
\nonmoral L^\dagger - \frac12 \{\nonmoral L^\dagger\nonmoral  L, \nonmoral \varrho\} 
\right ).
\label{eq:GKSL-symmetric-segment}
\end{equation}
Note that the evolution operator is defined through two asymmetric, correcting each other Lindblad operators. Thus, the state has to be symmetric for the whole evolution. Numerical analysis confirms this conclusion (see Fig.~\ref{fig:symmetry-line}).

%It is worth noting that for some graphs the symmetrization procedure is not
%necessary. For example, if we evolve quantum stochastic walks on undirected
%perfect $k$-ary tree\todo{figure?}, with the initial state localized in the root,
%we can observe that the distribution is symmetric due to arbitrary graph automorphism. \todo{define automorphism}

\begin{figure}[t]
	
	\begin{tikzpicture}[node distance=1.5cm,thick,scale=0.9,transform shape]

	\node[nodeStyle] (A0)   {$-3$};
	\node[nodeStyle] (A1) [right of=A0] {$-2$};
	\node[nodeStyle] (A2) [right of=A1] {$-1$};
	\node[nodeStyle] (A3) [right of=A2] {$0$};
	\node[nodeStyle] (A4) [right of=A3] {$1$};
	\node[nodeStyle] (A5) [right of=A4] {$2$};
	\node[nodeStyle] (A6) [right of=A5] {$3$};
	\node[] (Eright)   [right=0.5cm of A6] {$\cdots$};
	\node[] (Eleft)   [left=0.5cm of A0] {$\cdots$};
	\draw (A6) -- (Eright) ;
	\draw (A0) -- (Eleft) ;
	
	\node (Am1) [left of=A0] {};
	\node (Am2) [left of=Am1] {};
	\node (Ap1) [right of=A6] {};
	\node (Ap2) [right of=Ap1] {};

	\draw (A0) -- (A1);
	\draw (A1) -- (A2);
	\draw (A2) -- (A3);
	\draw (A3) -- (A4);
	\draw (A4) -- (A5);
	\draw (A5) -- (A6);
	
	\draw[dashed] (A0) to [out=60,in=120,looseness=1] (A2);
	\draw[dashed] (A2) to [out=60,in=120,looseness=1] (A4);
	\draw[dashed] (A4) to [out=60,in=120,looseness=1] (A6);
	\draw[dashed] (A1) to [out=-60,in=-120,looseness=1] (A3);
	\draw[dashed] (A3) to [out=-60,in=-120,looseness=1] (A5);
	
	\draw[dashed,shorten >=1cm] (A1) to [out=-120,in=-60,looseness=1] (Am1);
	\draw[dashed,shorten >=1cm] (A0) to [out=120,in=60,looseness=1] (Am2);
	\draw[dashed,shorten >=1cm] (A6) to [out=60,in=120,looseness=1] (Ap2);
	\draw[dashed,shorten >=1cm] (A5) to [out=-60,in=-120,looseness=1] (Ap1);
	\end{tikzpicture}
	
	\caption{Line graph. Dashed lines correspond to additional amplitude 
		transitions every two nodes coming from the GKSL model.\label{fig:line}}
\end{figure}
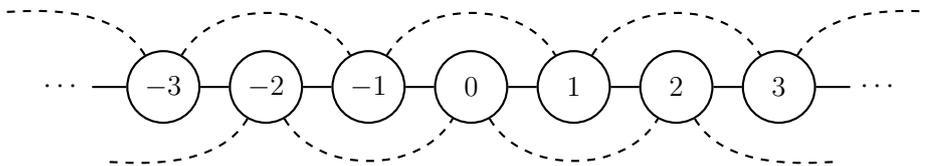

\subsection{Propagation analysis}
In this section we present a numerical analysis of interpolated standard NGQSW. To do so we analyse use the scaling exponent $\alpha$ of the variance $\mu_2(t)= \Theta(n^\alpha)$. In Section~\ref{sec:prop-qsw} we have analytically shown that the scaling exponent in time limit equals 2.

We consider the model based on the symmetrized quantum stochastic walk given by
Eq.~\eqref{eq:GKSL-symmetric-segment}. We use a standard nonmoralizing 
Hamiltonian $\nonmoral{H}$ and locally rotating Hamiltonian $\nonmoral H_{\rm
	rot}$. We choose $\nonmoral{\mathbb L}=\{L_1,L_2\}$ to be a collection of two Lindblad
operator defined in previous section. The evolution takes the form
\begin{equation}
\begin{split}
\frac{\textrm{d}}{\textrm{d}t}\nonmoral \varrho &=-\ii (1-\omega)[\nonmoral H
,\nonmoral \varrho] +\omega\left (\ii[\nonmoral H_{\textrm{rot}},\nonmoral 
\varrho] + \sum_{\nonmoral L \in \{\nonmoral L_1,\nonmoral L_2\}}\left
(\nonmoral L \nonmoral \varrho \nonmoral L^\dagger - \frac12 \{\nonmoral
L^\dagger\nonmoral  L, \nonmoral \varrho\} \right )\right )
\label{eq:weighted-line-GKSL-model}.
\end{split}
\end{equation}

\begin{figure}[t]
	\centering
	\includegraphics{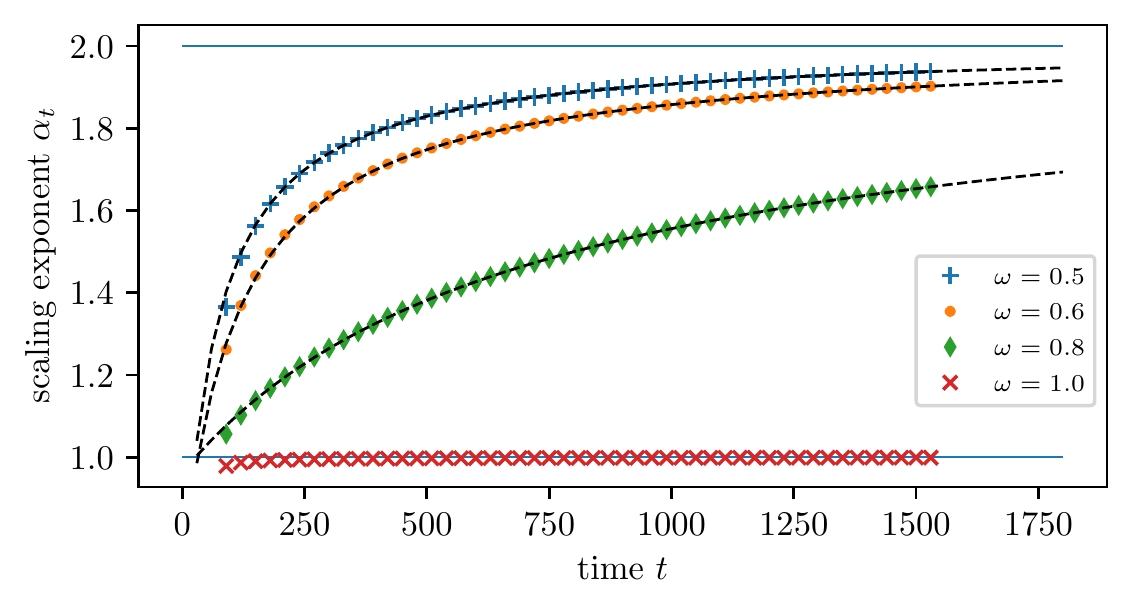}
	
	\caption{Slope of the local regression line for various values of $\omega\leq1$
	for interpolated standard NGQSW. The scaling plots were computed for timepoints
	$t=30,60,\ldots, 1590$ and we chose batch size $l=5$, according to the method
	presented in Sec.~\ref{sec:exponent-estimation}. Vertical solid lines represents
	values $\alpha=1,2$. Dashed lines for $\omega <1$ presents fit line according
	to model $f(t;p)\coloneqq p_1 - p_2\frac{1}{(t-p_3)^{p_4}}$. The fit was
	calculated using Levenger-Marquardt algorithm implemented in \texttt{LsqFit.jl}
	based on 30 estimations $\alpha_t$ for largest values of~$t$.}
\label{fig:hurst-on-line}
	
\end{figure}

To determinate the scaling exponent we used a method presented in
Sec.~\ref{sec:exponent-estimation}.  The scaling exponents were derived for
$\omega = 0.5, 0.6, 0.8, 1.0$. Results are shown in
Fig.~\ref{fig:hurst-on-line}. We can see that for $\omega=1$ the scaling
exponent converged to $1$. On the other hand for $\omega< 1$ the slope increases
in time and exceeds~$1$, which is the upper bound for classical propagation. To determine the limiting value of $\alpha_t$, we fitted pairs $(t,\alpha_t)$ to model function $f(t;p)\coloneqq p_1 - p_2\frac{1}{(t-p_3)^{p_4}}$. Note that $\lim_{t\to\infty} f(t;p) =p_1$ for positive values of $p_4$, which was assumed during the optimization. For value $\omega=0.5,0.6,0.8$ we obtained values $p_1=2.00,2.01, 2.04$, which shows that the propagation is in fact ballistic.

The results confirm that the fast propagation (ballistic or at least
super-diffusive) is the property of global interactions present in quantum
stochastic walks and not from the fact that the original model allows additional
transitions not according to the graph structure. It remains an open question,
whether the evolution is convergent to a subspace corresponding to sink
vertices. In the next chapter, we will show that this convergence occurs, and in
fact, it is stronger compared to LQSW.

% !TeX spellcheck = en_US
\chapter{Convergence of quantum stochastic walks}\label{sec:convergence-qsw}
\chaptermark{Convergence of QSW}

In chapter~\ref{sec:nonmoralizing-qsw} we proposed a new model of quantum stochastic
walk. We showed that in the case of absence of the Hamiltonian, which defines
the walk on the underlying graph, the structure of the directed graph is
perfectly preserved. However in order to introduce superdiffusive propagation,
one need to introduce the Hamiltonian. Thus the model in fact may present a
trade-off between the preservation of the arcs direction and the propagation
speed. Furthermore NGQSW may have different converging property comparing to
LQSW and GQSW.

In this chapter we investigate the directed-graph preservation for various QSW
models. We analyze the direction preservation for standard LQSW, standard
GQSW, and standard NGQSW. We achieve this by analyzing the limiting behavior of
the quantum walks.

Various convergence classes can be considered. We will say that the evolution is
convergent, if for arbitrary initial state $\varrho_0$ there exists stationary
state $\varrho_\infty$ such that $\varrho_0$ will converge to $\varrho_\infty$
according to the evolution. Otherwise we say that the model is not convergent.
Note that any evolution based on a time-independent GKSL master equation has at
least one stationary state \cite{spohn1977algebraic}. Hence, it is not possible
to define an evolution that does not converge for any initial state.

Evolution may have a special property called \emph{relaxing property}. It is
defined as an evolution which has a unique stationary state. Uniqueness of the
stationary states implies that the evolution is convergent for any choice of
initial state~\cite{schirmer2010stabilizing}.

All this properties can be checked without the simulation of the GKSL master
equation. Let us recall the GKSL master equation presented in
Eq.~\eqref{eq:gksl-master-equation}
\begin{equation}
\frac{\dd \varrho}{\dd t} = -\hbar [H, \varrho] + \sum_{L\in \mathbb L} \left(
\gamma_L L \varrho L^\dagger - \frac12 \gamma_L \{L ^\dagger L, \varrho\}
\right).
\end{equation}
Since we consider time-independent Hamiltonian and Lindblad operators, the
solution to the equation above takes the form
\begin{equation}
\vecc{\varrho_t} = \exp(St) \vecc{\varrho_0},
\end{equation}
where $S$ is an evolution generator of the form presented in
Eq.~\eqref{eq:diff-qsw-operator}. It can be shown that the eigenvalues $\lambda$
of $S$ satisfies $\Re\lambda \leq 0$. Since there exists at least one stationary
state, $S$ possess at least one zero eigenvalue. Furthermore, $S$ has single
zero eigenvalue iff the evolution represented by $S$ is relaxing
\cite{schirmer2010stabilizing}. Finally, if $S$ possess purely imaginary
eigenvalues, then one can suspect the existence of an initial state that result
in periodic or quasi-periodic evolution.

\section{Convergence of LQSW} \label{sec:local-convergence}

\paragraph{Strongly connected digraphs and single sink condensation graphs}

The local environment interaction case is relaxing for all connected undirected
graphs and arbitrary Hamiltonian \cite{liu2016continuous}. The proof presented
therein is based on the Spohn theorem~\cite{spohn1977algebraic}, which requires
the self-adjointess of the set of Lindblad operators $\mathbb L$, hence its
applications is limited to the undirected graphs case. Nevertheless we show,
that the result can be extended to strongly connected digraphs and weakly
connected graphs with single sink vertex. Our proofs utilize Conditions 2. and
3. from \cite{schirmer2010stabilizing}, recalled below as
Lemma~\ref{lemma:condition2} and \ref{lemma:condition3}. By the interior we mean set
of density matrices with full rank.

\begin{lemma}[\cite{schirmer2010stabilizing}] \label{lemma:condition2}
	Let $\mathcal H$ be a Hilbert space. If there is no proper subspace
$S\subsetneq\mathcal H$, that is invariant under all Lindblad generators
$L\in\mathbb L$  then the system has a unique steady state in the interior.
\end{lemma}

\begin{lemma}[\cite{schirmer2010stabilizing}] \label{lemma:condition3}
	If there do not exists two orthogonal proper subspaces of $\mathcal H$ that are
simultaneously invariant under all Lindblad generators $L\in\mathbb L$, then
the system has unique fixed point, either at the boundary or in the interior.
\end{lemma}

\begin{theorem} \label{theorem:local-strongly}
Let $\vec G=(V,\vec E)$ be a strongly connected digraph and let $\mathbb L=
\{L_{vw}=c_{(v,w)}\ketbra{w}{v}\colon (v,w)\in \vec E,c_{(v,w)}>0\}$. Then the LQSW
with $\mathcal{L}$ is relaxing for arbitrary Hamiltonian $H$ with stationary
state in the interior.
\end{theorem}
\begin{proof}
Let $\mathcal H$ be a Hilbert space  spanned by $\{\ket v: v\in V\}$ and $S\neq
\{0\}$ be arbitrary subspace of $\mathcal H$ invariant under $\mathbb L$. We
will show that $S=\mathcal H$, which by Lemma~\ref{lemma:condition2} will end
the proof.
	
Let $v\in V$ be arbitrary vertex. Let $\ket \psi\in S$ be a nonzero vector.
Since $G$ is strongly connected, for each $w$ there is a directed path
$P_{vw}=(v_1=v,v_2,\dots,v_k,w)$. Then
\begin{equation}
c_w \ket w =L_{v_kw} L_{v_{k-1}v_k}\cdots L_{v_2,v_1}\ket{\psi}
\end{equation}
for some $c_w\in\C_{\neq 0}$. Hence we have $\ket{w}\in S$ for all $w\in V$.
Hence $S\supseteq\operatorname{span}(\{\ket{w} : w\in V \})=\mathcal H$ and by
this, $S=\mathcal H$.
\end{proof}

Let $\vec  G=(V,\vec E)$ be an arbitrary weekly connected digraph. Let
$V'\subset V$ be such a set that for each $v,w\in V'$ there exist paths from $v$
to $w$ and from $w$ to $v$. Let any superset of $V'$ does not have this
property. Then we call induced subgraph $\vec G'=( V',\vec E')$ a strongly
connected component.

We define a \emph{condensation graph} $\vec G^c=(V^c,\vec E^c)$ as follows. Let
$V^c=\{V_1,\dots,V_k\}$ be a partition of $V$ such that each $V_i$ constructs a
strongly connected component of $\vec G$. Furthermore let $(V_i,V_j)\in \vec E^c
$ iff there exist $v_i\in V_i$ and $v_j\in V_j$ such that $(v_i,v_j)\in \vec E$.
Then we call $\vec G^c$ a condensation graph of $\vec G$. Note that $\vec G^c$
is a directed acyclic graph. Let us denote $L(\vec G)$ the collection of leaves
(sinks) of a digraph $\vec G$.

\begin{theorem} \label{theorem:local-one-sink}
	Let $\vec G=(V,\vec E)$ be a weakly connected digraph such that
	$|L(\vec G^c)|=1$ and let $\mathbb L= \{c_{(v,w)}\ketbra{w}{v}\colon
	(v,w)\in\vec  E\}$ for some $c_{(v,w)}>0$. Then the LQSW with $\mathbb{L}$ is relaxing for arbitrary Hamiltonian $H$.
\end{theorem}

\begin{proof}

Suppose $S_1\neq \{0\},S_2\neq \{0\}$ are two subspaces of $\mathcal H$ and let
$\ket{\psi_1}\in S_1,\ket{\psi_2}\in S_2$. Let $w$ be an element from the sink
vertex of $\vec G^c$. Similarly to method in
Theorem~\ref{theorem:local-strongly} one can show that there exist
$L_1^1,\dots,L_k^1\in \mathbb L$ and $L_1^2,\dots,L_{k'}^2\in \mathbb L$ such
that $c_w^1\ket{w} = L_k^1L_{k-1}^1\dots L_1^1 \ket{\psi_1}$ and $c_w^2\ket{w} =
L_{k'}^2L_{k'-1}^2\dots L_1^2 \ket{\psi_2}$ for some $c_w^1,c_w^2\in \C_{\neq
	0}$. Hence $S_1$ and $S_2$ are not orthogonal and by
Lemma~\ref{lemma:condition3} the theorem holds.
\end{proof}
Several interesting things can be pointed. First note that since strongly
connected digraphs satisfies the assumptions of
Theorem~\ref{theorem:local-one-sink}, the Theorem~\ref{theorem:local-strongly}
can be considered as a special case of the former one. However
Theorem~\ref{theorem:local-strongly} provides that the stationary state has full
rank. This is not achievable in general for directed graphs that are not
strongly connected and such with $L(\vec G^c)=\{V_i\}$ being a singleton. Let us
consider an evolution with no Hamiltonian. Then the evolution is a CTRW, and the
stationary state is spanned by vertices from the sink from the condensation
graph.

Note that the Hamiltonian has no influence on the for of the convergence.
However it does have an impact on a form of stationary state. Let us consider a
directed path graph $P_2 = (\{1,2\}, \{(1,2)\})$. Let us take a Lindblad
operator collection $\mathbb L = \{\ketbra{2}{1}\}$. If there is no Hamiltonian,
then
$\left [\begin{smallmatrix}
0 & 0 \\ 0 & 1
\end{smallmatrix}
\right ]$
is the unique stationary state. However if we apply Hamiltonian $H = \frac{1}{2}\left [\begin{smallmatrix}
0 & 1 \\ 1 & 0
\end{smallmatrix}\right ]$, then the stationary state changes into
$\frac{1}{3}\left [\begin{smallmatrix}
1 & -\ii \\ \ii & 2
\end{smallmatrix} \right ]$.

\paragraph{Multi-sink condensation graphs}
\begin{figure}
	\centering
%	\subfloat[][\label{fig:directed-k12}directed $K_{1,2}$ graph]{
		\begin{tikzpicture}[node distance=2cm]
		\tikzset{nodeStyle/.style = {circle,draw,minimum size=2.5em}}
		
		\node[nodeStyle] (A)  {1};
		\node[] (C) [below of=A] {};
		\node[nodeStyle] (B) [left of=C] {2};
		\node[nodeStyle] (D) [right of=C] {3};
		
		\tikzset{EdgeStyle/.style   = {->,>=latex}}
		\draw[EdgeStyle] (A) to (B);
		\draw[EdgeStyle] (A) to (D);
		
		\end{tikzpicture}
%	}\hspace{1cm}
%	\subfloat[][\label{fig:directed-k13}directed $K_{1,3}$ graph]{
%		\begin{tikzpicture}[node distance=2cm]
%		\tikzset{nodeStyle/.style = {circle,draw,minimum size=2.5em}}
%		
%		\node[nodeStyle] (A)  {1};
%		\node[nodeStyle] (C) [below of=A] {3};
%		\node[nodeStyle] (B) [left of=C] {2};
%		\node[nodeStyle] (D) [right of=C] {4};
%		
%		\tikzset{EdgeStyle/.style   = {->,>=latex}}
%		\draw[EdgeStyle] (A) to (B);
%		\draw[EdgeStyle] (A) to (C);
%		\draw[EdgeStyle] (A) to (D);
%		
%		\end{tikzpicture}}
	\caption{\label{fig:local-graphs-multisink} An oriented $K_{1,2}$.}
\end{figure}
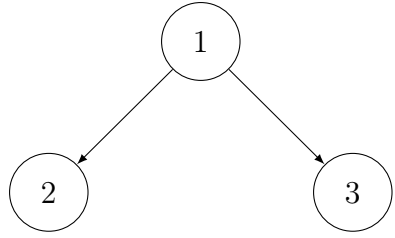

The remaining class of weakly connected graphs is those for which $|L(\vec
G^c)|>1$. Note that in this case one cannot expect relaxing property for a
general Hamiltonian. In particular, let as consider a purely classical CTRW on an
oriented $K_{1,2}$ as in Fig.~\ref{fig:local-graphs-multisink}. Note that both
$\ketbra{2}$ and $\ketbra{3}$ are a proper stationary states.

Let $\vec G=(V, \vec E)$ be digraph. Let us consider a CTRW with Lindblad
operators $\mathcal L=\{\ketbra{w}{v}: (v,w)\in \vec E\}$ and Hamiltonian $H$
being an adjacency matrix of its underlying graph $G$. Let us consider an
interpolated LQSW with interpolating parameter $\omega\in[0,1]$. Note that for
$\omega=0$ we obtain a Schr\"odinger equation, thus if $H\neq 0$ we have a
non-convergent evolution. For $\omega = 1$ we have a CTRW, hence states
localized in different strongly connected components which are sinks will
converge to two different stationary states. However the evolution will be in
general convergent.

%For $\omega \in (0,1)$ the behavior of LQSW on a graph with multi-sink
%condensation graph depends  on its topology. Let $S$ be evolution generator and
%$P_\omega(x) = \det(S-x\Id)$ be its characteristic polynomial. If $S$ is of
%order $n^2\times n^2$, then the polynomial takes the form
%\begin{equation}
%P_\omega(x) = \sum_{i=0}^{n^2}P^{(i)}(\omega)x^i.
%\end{equation}
%Note that $P^{(0)}\equiv 0$, since $S$ has at least one zero-eigenvalue. Furthermore, $P^{(i)}(\omega)$ is a polynomial in $\omega$.
%
%The evolution with fixed $\omega$ is relaxing iff $P^{(1)}(\omega) \neq 0$. Let us
%first consider the directed $K_{1,2}$ presented in Fig.~\ref{fig:directed-k12}. Then
%we have
%\begin{equation}
%\begin{split}
%P^{(1)}(\omega) &= -\frac{425}{16} \omega^8 + 135 \omega^7 -\frac{591}{2}\omega^6 + 356 \omega^5 -249 \omega^4 + 96 \omega^3 - 16 \omega^2.
%\end{split}
%\end{equation}
%The only real root is 0, which means $P^{(1)}(\omega) \neq 0$ for any $\omega \in
%(0,1)$. This means, that for $\omega \in (0,1)$ the evolution is relaxing. At
%the same time for directed graph $K_{1,3}$ graph presented in
%Fig.~\ref{fig:directed-k13}. we have $P^{(1)}  \equiv 0$. Hence, the
%evolution is no longer relaxing. The results above were generated with
%\texttt{SymPy}, a Python library for symbolic
%computations~\cite{meurer2017sympy}. This confirms that in the case of multi-sink condensation graph, the convergence property depends on the chosen topology, even for intermediate values of scaling parameter $\omega$.

We conclude our analysis with numerical investigation. We have analyzed a standard
LQSW on various random graph models in context of its convergence properties. We
have chosen only graphs with multi-sink condensation graph. The results of the
numerical experiment can be found in Fig.~\ref{fig:convergence-local}. The red
bar presents an amount of graphs which yield a relaxing evolution. Black bar
presents an amount of graphs which does not yield a relaxing evolution, but were
still convergent. Finally blue bar (not present in given figure) yielded the
graph for which evolution generator had purely imaginary eigenvalues.

\begin{figure}[t]
	\captionsetup[subfigure]{oneside,margin={0.8cm,0cm}}
	\subfloat[$\mathcal G_n^{\rm BA}(1)$ \label{fig:convergence-local-ba1}] {\includegraphics{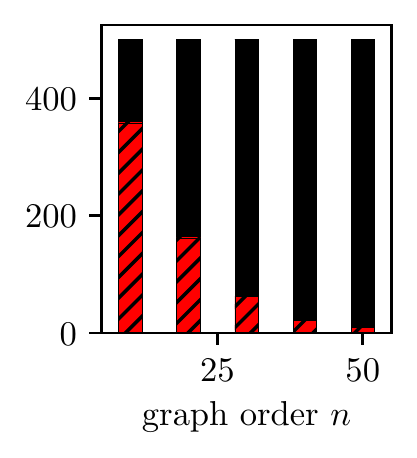}}
	\captionsetup[subfigure]{oneside,margin={0.2cm,0cm}}
	\subfloat[$\mathcal G_n^{\rm BA}(3)$ \label{fig:convergence-local-ba3}] {\includegraphics{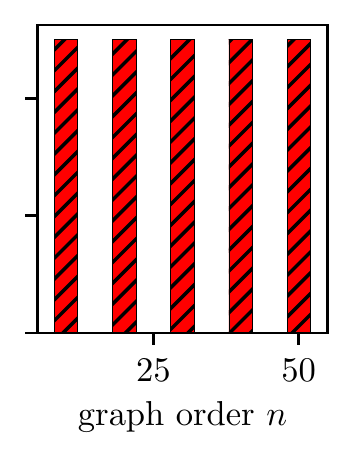}}
	\captionsetup[subfigure]{oneside,margin={-2.2cm,0cm}}
	\subfloat[$\vec {\mathcal G}_n^{\rm ER}(0.4)$ \label{fig:convergence-local-er}] {\includegraphics{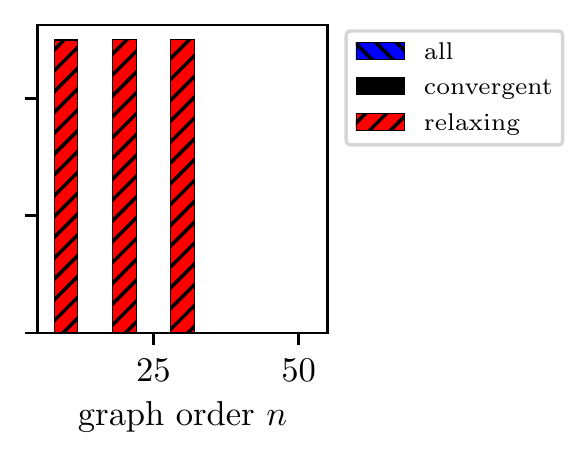}}
	\caption{\label{fig:convergence-local}The convergence statistics of LQSW for
		various directed random graph models. For each $\randgn[BA](m_0)$
		models a sample $G$ was chosen, then its random orientation $\vec G$ was
		chosen. For randomly  directed \ER graphs only weakly connected were considered.
		We have chose at random 500 graphs for each model and $n$ s.t. the corresponding condensation graph has at least two sinks.
		Since it was extremely difficult to find such for \ER graphs, numerical results
		are limited to at most $30$ nodes. The analysis was done through analysis of
		eigenvalues of evolution generator $S$. We considered eigenvalue $\lambda$ to
		be 0 if $|\lambda|<10^{-10}$. We assumed $\lambda$ to be purely imaginary if
		$|\real \lambda|<10^{-10}$ and $|\imaginary \lambda| > 10^{-10}$.}
\end{figure}

We haven't found a single graph for which the evolution operator had a purely
imaginary eigenvalues, which would suggest quasi-periodic evolution. For randomly oriented $\randgn[BA]$ models and $\randdgn[ER](0.4)$ all graphs
yield relaxing property. Contrary, for randomly oriented trees the number of
graphs yielding relaxing evolution decreases with the order of the graph.

\section{Convergence of GQSW} \label{sec:global}

\subsubsection{Undirected graphs} We start this section with providing the general
result for the commuting operators.

\begin{proposition}\label{theorem:commuting}
	Let us consider GKSL master equation in the case of commuting Lindbladian
operators $\mathbb L$ and Hamiltonian $H$. Then the evolution operation is of
the form
	\begin{equation}
	(U\kron \bar U)D_{S}(U\kron \bar U)^\dagger,
	\end{equation}
	where
	\begin{equation}
	D_{S}= -\ii (D_H\kron \Id -\Id\kron D_H) + \sum_{L\in\mathcal L}\left (D_L\kron
\bar D_L - \frac{1}{2} \bar D_L D_L\kron \Id-\frac{1}{2} \Id \kron D_L \bar 
D_L\right )  \label{eq:diagonal-commuting}
	\end{equation}
	is a diagonal matrix. Here we assume that $U$ is a unitary operator and $D_H,D_L$ are
	diagonal operators such that $H=UD_HU^\dagger$ and $L=UD_LU^\dagger$.
\end{proposition}
\begin{proof}
	The proof comes directly from the eigendecompositions of the operators. Since
all operators commute, it is possible to find common eigendecomposition with
the same unitary matrix. By this we can easily find the result.
\end{proof}
The standard GQSW on undirected graphs is a special case of the evolution
described in the theorem above, where we choose only single Lindbladian operator
$L=H$.

\begin{theorem} \label{theorem:undirected}
	The stationary states of the standard interpolated GQSW are precisely the
stationary states of the CTQW.  The evolution is convergent for
$\omega\in(0,1]$, but not relaxing iff the system size is greater than one.
\end{theorem}
\begin{proof} By the model construction  we can choose $\mathbb  L=\{\sqrt\omega  A\}$ and $H=(1-\omega )A$ and apply the Theorem~\ref{theorem:commuting}. The 
diagonal matrix takes the form
	\begin{equation}
	D_{S_{\omega}}=-\ii (1-\omega) (D\kron \Id -\Id\kron D) + \omega \left (D \kron D -
\frac{1}{2} D^2\kron \Id-\frac{1}{2} \Id \kron D ^2\right).
	\end{equation}
	Here we assume $A=UDU^\dagger$. Since $A$ is hermitian, operator $D$ is a
	real-valued diagonal matrix. The diagonal entries of operator $D_{S_{\omega}}$
	are eigenvalues which characterize the evolution. Let $d_i \coloneqq \bra i D \ket i$. Then we have
	\begin{equation}
	\begin{split}
	\bra {i,j} D_{S_{\omega}}\ket {i,j} &=  -\ii(1-\omega) (d_i- d_j)+\omega \left
( d_id_j  - \frac{1}{2}d_i^2 -\frac{1}{2}d_j^2 \right) \\ &= -\ii(1-\omega)
(d_i- d_j)- \frac\omega 2 (d_i - d_j)^2.
	\end{split}
	\end{equation}

	Here $-\ii (1-\omega ) (d_i-d_j)$ corresponds to
purely Hamiltonian evolution, and hence to CTQW. Since 0-eigenvalues of
$S_\omega$ correspond to 0-eigenvalues of Hamiltonian part of the system,
which furthermore correspond to the stationary states of the CTQW, we obtained
the first part of the theorem.
	
	Note that $S_\omega$ does not have purely imaginary eigenvalues for $\omega>0$.
Hence, we have that the evolution is convergent. Since the set of stationary
states of CTQW for graph with $n$ vertices has at least $n$ elements, we obtain
that the presented evolution is never relaxing.
\end{proof}
The result from the above theorem implies that we can generate the
stationary states of the CTQW by adding proper Lindbladian operator.

\subsubsection{Directed graphs} In this section we provide an example of standard
GQSW on a directed graph for which the evolution is no longer convergent.

\begin{theorem}\label{theorem:global-directed}
	There exist an infinite number of digraphs $\vec G$ with corresponding initial
states $\varrho_0$ for which the interpolated standard GQSW is non-convergent
for an arbitrary value of the smoothing parameter $\omega$.
\end{theorem}
\begin{proof}
	Case $\omega =0$ comes directly from the properties of continuous-time quantum evolution. Let us consider $\omega>0$.
	\begin{figure}
		\centering
		\begin{tikzpicture}
		
		\def \n {8}
		\def \radius {2.5cm}
		\def \margin {8} % margin in angles, depends on the radius
		
		\foreach \s in {7,...,0}
		{
			\node[draw, circle] (\s) at ({360/\n * (\s - 1)}:\radius) {$\s$};
			\draw[] ({360/\n * (\s - 1)+\margin}:\radius)
			arc ({360/\n * (\s - 1)+\margin}:{360/\n * (\s)-\margin}:\radius);
		}
		\draw [<-,>=latex] (0)  edge (2);
		\draw [<-,>=latex] (1)  edge (3);
		\draw [<-,>=latex] (2)  edge (4);
		\draw [<-,>=latex] (3)  edge (5);
		\draw [<-,>=latex] (4)  edge (6);
		\draw [<-,>=latex] (5)  edge (7);
		\draw [<-,>=latex] (6)  edge (0);
		\draw [<-,>=latex] (7)  edge (1);

		\end{tikzpicture} 
		\caption{An example of strongly connected directed graph, for which the global interaction case evolution is not convergent.}\label{fig:example-small-directed}
	\end{figure}
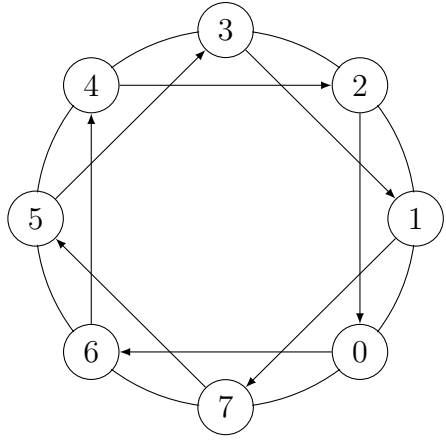
	We choose a circulant graph of size $4k$ for $k>1$ and with extra jump every
two vertices. An example for $k=2$ is presented in
Fig.~\ref{fig:example-small-directed}. The graph and its underlying graph are
circulant matrices. Therefore, we can use Eq.~\eqref{eq:diagonal-commuting} to
find out that there exists an eigenvalue of the form $2(1-\omega)\ii$ with
corresponding eigenvector $\ket{C_{k}}\overline {\ket{C_{2k}}}$, where $\ket
{C_i}$  is the $i$-th eigenvector of a circulant matrix of the
form~\cite{gray2006toeplitz}
	\begin{equation}
	\ket {C_i}= \frac{1}{2\sqrt k}\sum_{j=0}^{4k} \exp \left(\frac{2\pi \ii 
		ij}{4k-1}\right) \ket i.
	\end{equation} 
	The initial state takes the form
	\begin{equation}
	\begin{split}
	\varrho(0)  &= \frac{1}{2}(\ket{C_k}+\ket{C_{2k}})(\bra{C_k}+\bra{C_{2k}}),
	\end{split}
	\end{equation}
	and the $\varrho(t)$ takes the form
	\begin{equation}
	\varrho(t) = \frac{1}{2}(\ketbra{C_k}{C_k} + \ketbra{C_{2k}}{C_{2k}} + 
	e^{2\ii(1-\omega)t}\ketbra{C_k}{C_{2k}}+ e^{-2\ii(1-\omega)t}\ketbra{C_{2k}}{C_k}).
	\end{equation}
	Since $\varrho(t)$ is periodic with period $\frac{\pi}{(1-\omega)}$, we obtain 
	the result.
\end{proof}
Note, that for different $t$ we can obtain different state in the sense of
possible measurement output. For example we have $\bra 0 \varrho(0)\ket 0 =
\frac{1}{2k}$, but at the same time we have $\bra 0 \varrho
(\frac{\pi}{2(1-\omega)}) \ket 0 = 0$.

Circulant graphs provide an infinite collection of directed graphs for which the
convergence does not hold. Note that the example used in the proof of
Theorem~\ref{theorem:global-directed} is a strongly connected directed graph.
This shows that the convergence in the local interaction case does not imply
the convergence in the global interaction case.

We finalize our analysis of standard GQSW with numerical investigations of
random digraphs. We have sampled $500$ weakly connected directed graphs for each
model and order of the graph. The statistics are presented in
Fig.~\ref{fig:convergence-global}. As in LQSW, none of sample graphs had a
purely imaginary eigenvalue, although based on the theorem above we know such
graphs exist. Almost all $\randdgn[BA](3)$ and  $\randdgn[ER](0.4)$ graphs
yielded relaxing evolution. For $\randdgn[BA](1)$ the number of relaxing GQSW
decreased with the graph order, as it was in LQSW model. Thus GQSW and LQSW have
statistically  similar convergence properties.

\begin{figure}
	\captionsetup[subfigure]{oneside,margin={0.8cm,0cm}}
	\subfloat[$\vec {\mathcal G}_n^{\rm BA}(1)$ \label{fig:convergence-global-ba1}] {\includegraphics[]{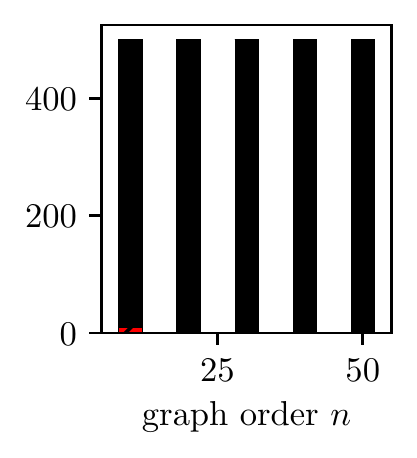}}
	\captionsetup[subfigure]{oneside,margin={0.2cm,0cm}}
	\subfloat[$\vec {\mathcal G}_n^{\rm BA}(3)$ \label{fig:convergence-global-ba3}] {\includegraphics[]{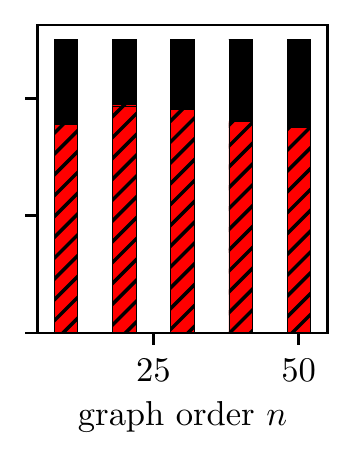}}
	\captionsetup[subfigure]{oneside,margin={-2.2cm,0cm}}
	\subfloat[$\vec {\mathcal G}_n^{\rm ER}(0.4)$ \label{fig:convergence-global-er}] {\includegraphics[]{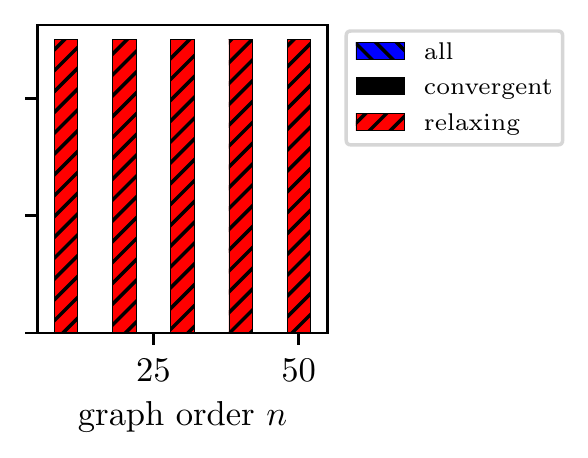}}
	\caption{\label{fig:convergence-global}The convergence statistics of GQSW for
	various directed random graph models. Only weakly connected \ER graphs were considered. We applied the same conditions for eigenvalues as in Fig.~\ref{fig:convergence-local}.}
\end{figure}

\section{Convergence of standard NGQSW}\label{sec:ngqsw-convergence}

Contrary to previous results, NGQSW is nonconvergent evolution even for an undirected graphs. 
\begin{theorem} \label{th:global-nonmoralizing}
	Let us consider the standard NGQSW. Then there exists a digraph
$\vec G$ and initial state $\varrho(0)$ for which the evolution is periodic in
time for an arbitrary value of the smoothing parameter $\omega\in(0,1]$.
\end{theorem}

\begin{proof}
	\begin{figure}
		\centering
		\begin{tikzpicture}
		\node[draw,circle] (0) {$v_0$};
		\def \n {5}
		\def \radius {2cm}
		\def \margin {8} % margin in angles, depends on the radius
		
		\def \s {1}
		\node[draw, circle] (\s) at ({360/\n * (\s - 1)}:\radius) {$v_\s$};
		\def \s {2}
		\node[draw, circle] (\s) at ({360/\n * (\s - 1)}:\radius) {$v_\s$};
		\def \s {3}
		\node[draw, circle] (\s) at ({360/\n * (\s - 1)}:\radius) {$v_3$};
		\def \s {4}
		\node[draw, circle] (\s) at ({360/\n * (\s - 1)}:\radius) {$v_4$};
		\def \s {5}
		\node[draw, circle] (\s) at ({360/\n * (\s - 1)}:\radius) {$v_\s$};

		\draw [-,>=latex] (0)  edge (1);
		\draw [-,>=latex] (0)  edge (2);
		\draw [-,>=latex] (0)  edge (3);
		\draw [-,>=latex] (0)  edge (4);
		\draw [-,>=latex] (0)  edge (5);
		\draw [-,>=latex] (5)  edge (4);
		\end{tikzpicture} 
		\caption{An example of graph for which non-moralizing global interaction evolution is not convergent.}\label{fig:example-big-directed}
	\end{figure}
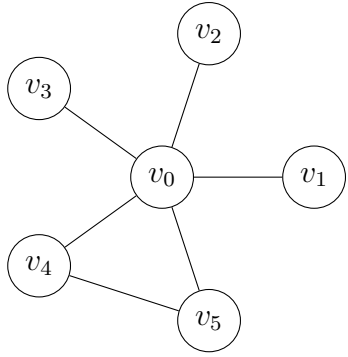
	
Let us consider a graph presented in Fig.~\ref{fig:example-big-directed}. Using
the scheme presented in Chapter~\ref{sec:nonmoralizing-qsw}, new graph will
consist of 5 copies of vertex $v_0$, two copies of vertices $v_4$ and $v_5$, and
single copy of other vertices. Let us consider standard NGQSW
	
	Let us choose two eigenvectors of the standard rotating Hamiltonian
	\begin{gather}
	\ket{\varphi} = \frac{1}{2\sqrt 3}\ket{v_0^0} 
	-\frac{\ii}{2}\ket{v_0^1}- \frac{1}{\sqrt 3}\ket{v_0^2} + \frac{\ii}{2} 
	\ket{v_0^3}+\frac{1}{2\sqrt 3}\ket{v_0^4},\\
	\ket{\psi}= \frac{1}{2\sqrt 3}\ket{v_0^0} 
	-\frac{\ii}{2}\ket{v_0^1}- \frac{1}{\sqrt 3}\ket{v_0^2} + \frac{\ii}{2} 
	\ket{v_0^3}+\frac{1}{2\sqrt 3}\ket{v_0^4}.
	\end{gather}
	One can show that the vectors
$\ket{\varphi}\overline{\ket{\varphi}}$, $\ket {\varphi} \overline{\ket{\psi}}$, $\ket {\psi} \overline{\ket{\varphi}}$, $\ket {\psi} \overline{\ket{\psi}}$
	are eigenvectors of the increased evolution operator $\tilde S_{t,\omega}$ for
arbitrary $\omega\in(0,1]$. Corresponding eigenvalues are respectively
$0,-2\ii\sqrt 3\omega,2\ii\sqrt 3\omega,0 $. Similarly to the example presented
in the previous section, the state
	\begin{equation}
	\tilde \varrho_0 = \frac{1}{2}\left (\ket{\varphi}+\ket{\psi}\right )\left (\bra{\varphi}+\bra{\psi}\right )
	\end{equation}
	is the required initial state. The state after time $t$ takes the form
	\begin{equation}
	\begin{split}
	\tilde \varrho_t &= \frac{1}{2}\left (\ket{\varphi}\bra{\varphi}+e^{-2\ii t\sqrt{3}\omega}\ket{\varphi}\bra{\psi}+ e^{2\ii t\sqrt{3}\omega}\ket{\psi}\bra{\varphi}+\ket{\psi}\bra{\psi}\right ).
	\end{split}
	\end{equation}
	The function $\tilde\varrho_t$ is periodic with period 
	$\frac{\pi}{\sqrt3\omega}$, hence we obtained the result.
\end{proof}

\begin{figure}
	\captionsetup[subfigure]{oneside,margin={0.8cm,0cm}}
	\subfloat[$\mathcal G_n^{\rm BA}(1)$ \label{fig:convergence-nonmoral-ba1}] {\includegraphics{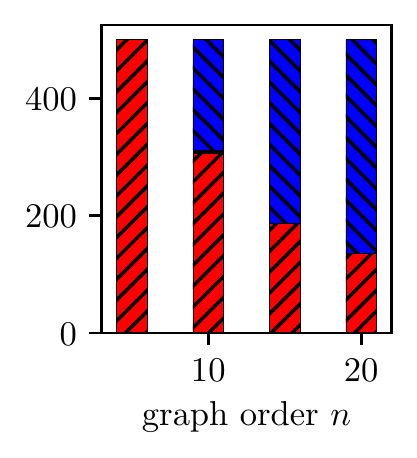}}
	\captionsetup[subfigure]{oneside,margin={0.2cm,0cm}}
	\subfloat[$\vec{\mathcal G}_n^{\rm BA}(3)$ \label{fig:convergence-nonmoral-ba3}] {\includegraphics{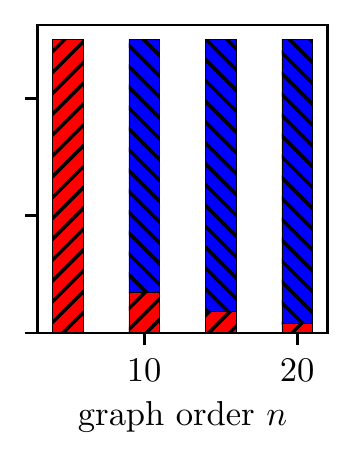}}
	\captionsetup[subfigure]{oneside,margin={-2.2cm,0cm}}
	\subfloat[$\vec{\mathcal G}_n^{\rm ER}(0.4)$ \label{fig:convergence-nonmoral-er}] {\includegraphics{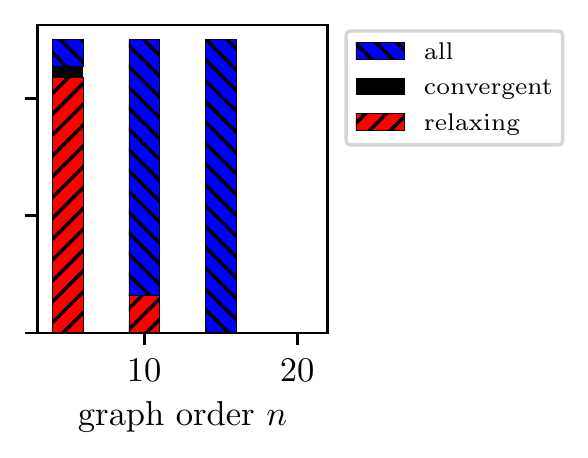}}
	\caption{\label{fig:convergence-nonmoral}The convergence statistics of NGQSW for
		various directed random graph models. Only weakly connected \ER graphs were considered. We applied the same conditions for eigenvalues as in Fig.~\ref{fig:convergence-local}. We were not able to perform the statistics for the \ER model for $n=20$, due to the size of the evolution generator.}
\end{figure}

Contrary to the LQSW and GQSW, it seems that such situation may occur quite
frequently for standard NGQSW, see Fig.~\ref{fig:convergence-nonmoral}. It turns
out that for majority of graphs the evolution generator have a purely imaginary
eigenvalues which suggest that the evolution will be periodic. However, provided
there is no imaginary eigenvalues, the evolution turned out to be relaxing.

Note that in the example above the probability distribution coming from the
measurement in canonical basis in the enlarged Hilbert space will differ.
However, independently on the chosen measurement time, the probability
distribution coming from the natural measurement of NGQSW remains unchanged.
This suggests that different measure of convergence has to be chosen.

Let $p(t;\nonmoral \varrho)$ be a probability distribution of measurement of the
NGQSW with initial state $\nonmoral \varrho$ after evolution time $t$,
according to its natural measurement. We will be interested, whether given
initial state, its probability distribution will converge. Formally, we are
interested whether there exists $p(\infty;\nonmoral \varrho)$ s.t.
\begin{equation}
\lim_{t\to\infty} \|p(t;\nonmoral \varrho)-p(\infty;\nonmoral \varrho)\| = 0.
\end{equation}
The spectral analysis is no longer useful here, as imaginary
eigenvalues may imply local evolution within subspace attached to $\nonmoral
V_{v}$. Instead, we made numerical analysis for a special choice of input
state of the form
\begin{equation}
\varrho = \frac{1}{|V|}\sum_{v\in V} \frac{1}{|\nonmoral V_v|}\sum_{\nonmoral{v} \in \nonmoral V_{v}} \ketbra{\nonmoral v}. \label{eq:nmgqsw-init-state}
\end{equation}
It turns out that difference between $p(t;\nonmoral{\varrho})$ and
$p(\num{10000};\nonmoral{\varrho})$ was almost monotonically decreasing as $t$
approached \num{10000}, see Fig.~\ref{fig:nonmoral-prob}. Hence we conclude that
at least for the proposed initial state the evolution was convergent in
probability.

\begin{figure}\centering
	\includegraphics{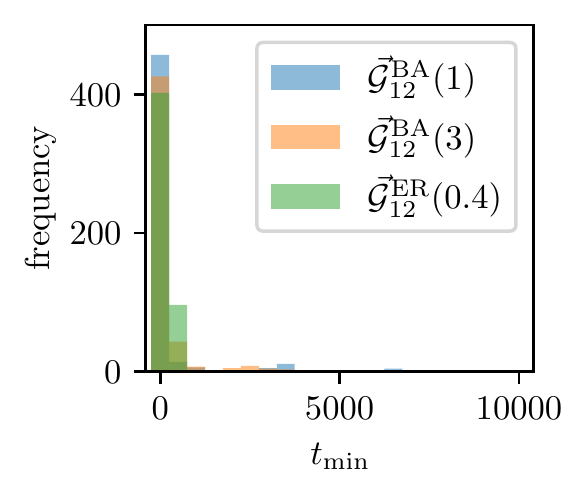}\ 		\includegraphics{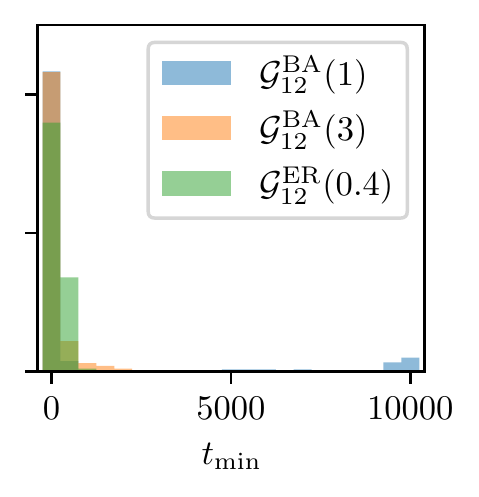}
	\caption{\label{fig:nonmoral-prob} Convergence for various directed random graph models for standard NGQSW. For (directed) \ER models only (weekly) connected graphs were chosen. For $t=0,100,\ldots,\num{10000}$ we calculated $p(t,\varrho)$ with $\varrho$ define as in Eq.~\eqref{eq:nmgqsw-init-state}. Then for given $p$ looked for minimal $t_{\min}$ such that for all $t'\geq t_{\min}$ we have $\|p(t'+100,\varrho)-p(\num{10000},\varrho)\| \leq \|p(t',\varrho)-p(\num{10000},\varrho)\|$. We repeated the procedure for 500 graphs for each graph model. Note that for some graphs we observed that the convergence were monotonic starting from very large values of $t$. However, for this samples the difference in norms for last 30 pairs of timepoints were (except single case) below $10^{-10}$. Hence, we claim that this deviations are due to a numerical error of estimating $p(t,\varrho)$.}
\end{figure}

\begin{figure}
	\centering
	\subfloat[graph with different limit distribution \label{fig:spacial-nonmoral-graph}] {\centering
\begin{tikzpicture}
%[0 0 0 0 1 1 0; 
% 0 0 0 0 1 0 0; 
% 0 0 0 0 1 1 0; 
% 0 0 0 0 1 1 0; 
% 1 1 1 1 0 1 0; 
% 1 0 1 1 1 0 1; 
% 0 0 0 0 0 1 0]
\def \n {7}
\def \radius {2cm}
\def \margin {8} % margin in angles, depends on the radius

\foreach \s in {7,...,1}
{
	\node[draw, circle] (\s) at ({360/\n * (\s )}:\radius) {$\s$};
	
}
\draw  (1)  edge (5); 
\draw  (1)  edge (6); 
\draw  (2)  edge (5);
\draw  (3)  edge (6);
\draw  (4)  edge (6);
\draw  (5)  edge (6);
\draw  (6)  edge (7);
\end{tikzpicture}	
}\hspace{.5cm}
	\subfloat[the probability of measuring vertex 5 for different inital state] {\includegraphics{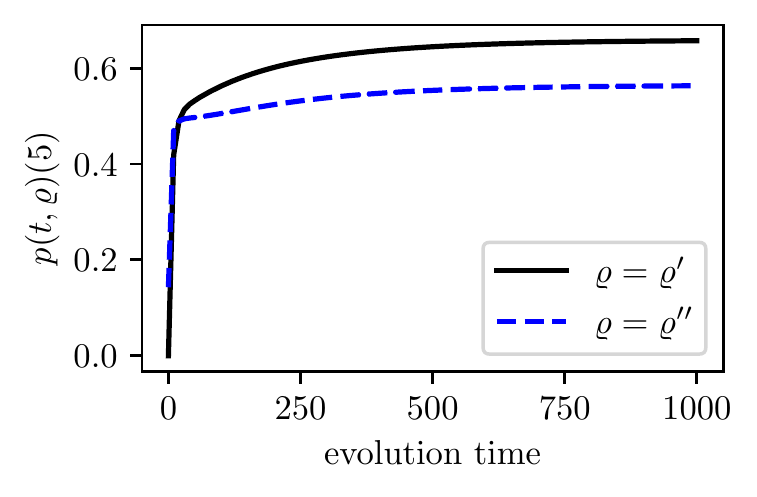}}

	\caption{Graph for which there exists two different stationary states in the
	sense of the natural measurement for NGQSW. The states can be obtained by
	starting in the state $\varrho'=\frac{1}{3}\sum_{\nonmoral v \in \nonmoral{V}_5} \ketbra{\nonmoral v}$ and the $\varrho''$ defined in Eq.~\eqref{eq:nmgqsw-init-state}.}\label{fig:different-measurements}
%\end{figure}
%\begin{figure}[th!]
	\centering
	\includegraphics[width=0.8\textwidth]{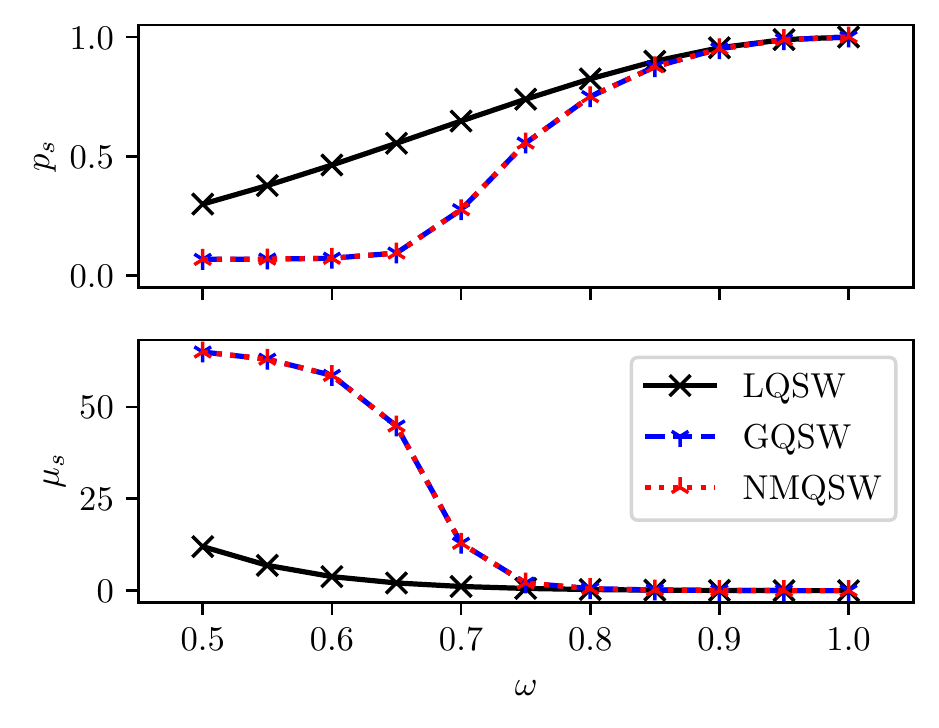}
	
	\caption{\label{fig:structure-observance-path} Measures $p_s$ and $\mu_s$ in term of $\omega=.5,.55,\ldots,1.$ for directed path with 15 vertices. The evolution starts in the initial state described in Eq.~\eqref{eq:nmgqsw-init-state}, and the evolution time equals \num{10000}. Note that the plots for GQSW and NGQSW coincides -- this comes from the fact that for each vertex in directed path the indegree is at most 1, hence the standard GQSW and NGQSW are indistinguishable
	}
\end{figure}

The limit probability distribution depends in general on the initial
state. Let us analyse the graph presented in
Fig.~\ref{fig:different-measurements}. For the standard NGQSW and two initial
states, we see that the limiting probability distribution differ.

\section{Digraph structure observance} \label{sec:digraph-structure-observance}

Let us consider standard interpolated QSW models with $\omega = 1$, \ie evolution defined for a digraph. Let $\vec G=(V,\vec E)$ be a digraph and let $v\in V$ be a sink vertex. Independently of chosen QSW model, mixed state defined over the space of $v$ is a stationary state. For LQSW we have 
\begin{equation}
\begin{split}
\frac{\dd \ketbra v}{\dd t} &= \sum_{(i,j)\in \vec E} \left (|c_{i,j}|^2\ketbra{j}{i} \cdot \ketbra{v}{v} \cdot \ketbra{i}{j} - \frac{1}{2}|c_{i,j}|^2 \{ \ketbra{i}{j} \cdot \ketbra{j}{i}, \ketbra{v}{v}\} \right ) \\
&=\sum_{(v,j)\in \vec E} \left (|c_{v,j}|^2\ketbra{j}{j} - |c_{v,j}|^2  \ketbra{v} \right ) =0,
\end{split}
\end{equation}
because there is no arc of the form $(v,j)$. Similarly, for GQSW with set $\mathbb L$ of Lindblad operators we have
\begin{equation}
\begin{split}
\frac{\dd \ketbra v}{\dd t} &= \sum_{L\in\mathbb L} \left (L \ketbra{v}{v} L^\dagger - \frac{1}{2} \{ L^\dagger L, \ketbra{v}{v}\} \right ) = 0,
\end{split}
\end{equation}
because $L\ket v$ is a zero vector.

The case of NGQSW is more complicated because the subspace connected to the sink
vertex is $\indeg(v)$-dimensional. Hence, based on the results from
Sec.~\ref{sec:ngqsw-convergence}, we should allow the state to evolve within the
subspace $S_v$ attached to $\nonmoral{V}_v$. Let $\nonmoral{v}\in
\nonmoral{V}_v$. Note that $H_{\rm rot}\ket{\nonmoral{v}}$ is a vector spanned
by $S_v$, and for any nonmoralizing Lindblad operator $L$ we have
$L\ket{\nonmoral v}=0$. This means, by linearity, that for any mixed state
defined over $S_v$ is evolving within the space spanned by $\nonmoral V_v$,
hence the probability distribution coming from the natural measurement is
stationary.

However, as it was shown in Sec.~\ref{sec:local-convergence}, even for a very
simple directed path $(\{1,2\}, \{(1,2)\})$ the amplitude for stationary state
may be localized outside the sink vertices in the presence of the Hamiltonian.
Still we expect, that as $\omega \to 1$, the more amplitude should be localized
in the subspaces attached to the sink vertices.

Let $\vec G$ be a directed graph and let $\vec G^c=(V^c,\vec E^c)$ be its
condensation graph with unique sink vertex $V^c_s \in V^c$. We propose two
measures of how much the state is localized in the sink vertex or its
neighborhood. First, we can determine the probability of being at any vertex
from $V^c_s$, \ie
\begin{equation}
p_s(t;\varrho) = \sum_{v\in V^c_s} p(t; \varrho_0)(v).
\end{equation}
Similarly we proposed measure based on the second moment. Let $w\in V$ and $v\in
V^c_s$. Let $d(w,V^c_s)= \min_{v\in V^c_s} d(w,v)$. Note that the function $d$
is well-defined if there is a unique sink vertex in the digraph $\vec G^c$.
\begin{equation}
\mu_{s}(t;\varrho) = \sum_{v\in V} d^2(v,w) p(t; \varrho_0)(v).
\end{equation}

For the evolution preserving the digraph structure, we expect
$p_s(\infty;\varrho)=1$ and $\mu_s(\infty;\varrho)=0$. In the case of QSW walk
we expect $p_s\to 1$ and $\mu_s \to 0$ as $\omega \to 1$. As we
can observe on Fig.~\ref{fig:structure-observance-path}, the measures converge
to proper values for all QSW models as $\omega\to1$.

We repeated the experiment for $ \randdgn[BA](1)$ and $ \randdgn[BA](2)$, see
Fig.~\ref{fig:structure-observance-ba}.  For GQSW model, independently of
$\omega$ the values of $p_s$ and $\mu_s$ were far from their optimal values $1$
and $0$. This is expected, as even for the undirected graphs the model is
projecting the initial state to stationary state of the unitary evolution. Note
that as $\omega\to 1$, LSQW model acquire $p_s=1$ and thus $\mu_1=1$. However,
for $\omega<1$ there is a clear gap between obtained and limit value. For NGQSW we
observe that independently of chosen $\omega$ the model converged almost fully
to the vertices from sink of the condensation graph. However, for 3 out of 50
graphs the value $p_s$ was below .99. This may be due to invalid choice of the
rotating Hamiltonian.

\begin{figure}[t!]
	\centering
\includegraphics{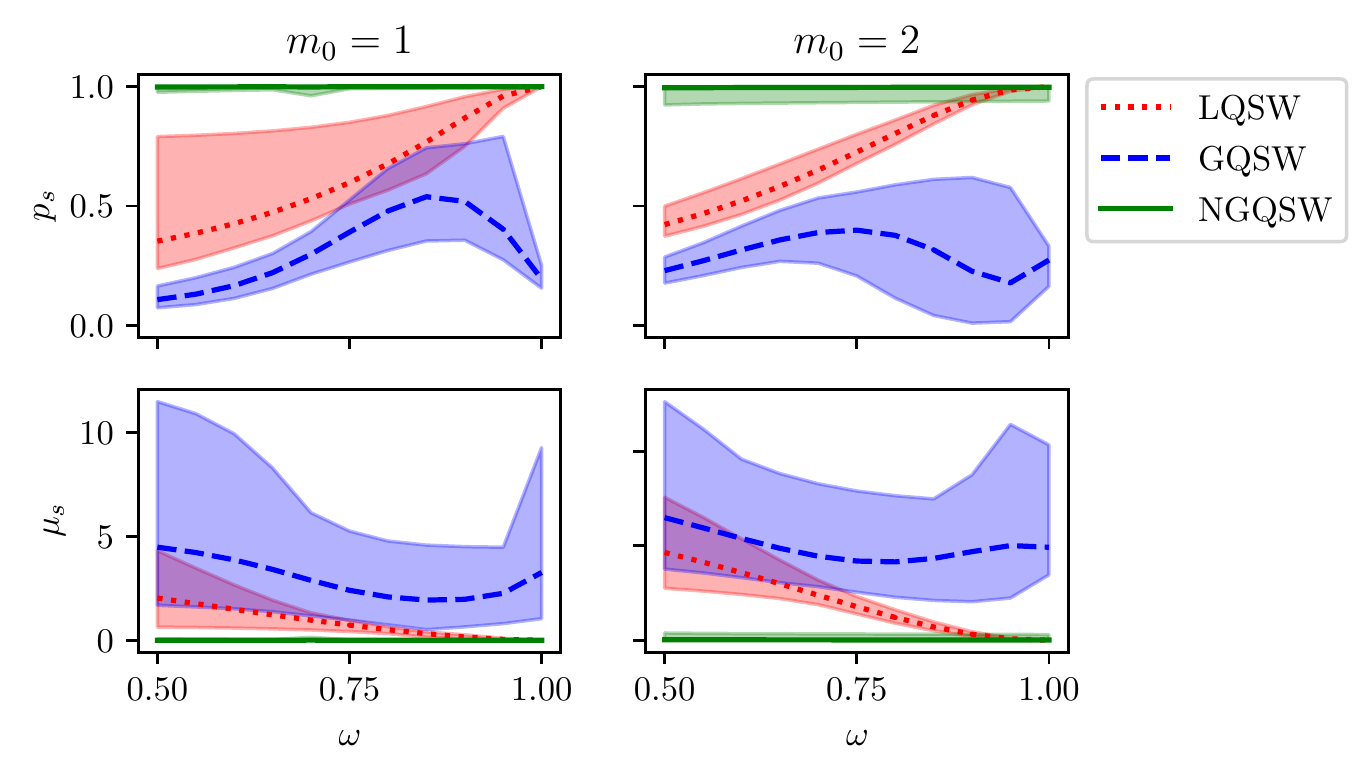}
	
	\caption{\label{fig:structure-observance-ba}Measures $p_s$ and $\mu_s$ in term of $\omega=.5,.55,\ldots,1.$ for \BA random digraphs with 15 vertices for $m_0=1,2$. For each value of parameter $m_0$ and each QSW model we sampled 50 graphs. The evolution starts in the initial state described in Eq.~\eqref{eq:nmgqsw-init-state}, and the evolution time equals \num{10000}. For NGQSW we chosen random rotating Hamiltonian, s.t. for each block we sampled independently $X+X^\top + \ii (Y-Y^T)$, where $X,Y$ are random matrices with entry sampled independently according to uniform distribution over $[0,1]$.
	}
\end{figure}

\newcommand{\erfc}{\operatorname{erfc}}

\chapter{Hiding vertices for quantum spatial search} \label{sec:hiding}
\chaptermark{Hiding vertices for quantum\ldots}

In Sec.~\ref{sec:preliminaries-propagation} we analyzed a complete graph in the context of
efficiency of quantum search. The analysis was simple, because the procedure
does not depend on the marked node. This comes from the fact that complete
graphs are vertex-transitive, \ie the vertex can be distinguished only by its
label.

However, in general one could expect that the transition rate and measurement
time may depend not only on the chosen graph, but also on the marked vertex. Let
us consider an adjacency matrix of a star graph $K_{n-1,1}$, with vertex 0 being
connected to all the other vertices. The eigenvalues of the adjacency matrix of the graph are
$-\sqrt{n-1}$, $0$, $\sqrt{n-1}$ with multiplicity $1$, $n-2$, and $1$
respectively check numbers. The eigenvectors corresponding to $-\sqrt{n-1}$ and $\sqrt{n-1}$ are:
\begin{gather}
\ket{-\sqrt{n-1}} = \frac{1}{\sqrt 2} \ket{0} - \frac{1}{\sqrt{2(n-1)}} \sum_{v=1}^n \ket{i}, \\
\ket{\sqrt{n-1}} = \frac{1}{\sqrt 2} \ket{0} + \frac{1}{\sqrt{2(n-1)}} \sum_{v=1}^n \ket{i} .
\end{gather}
Let us consider an initial state $\ket{\psi_0}= \ket{\sqrt{n-1}}$. If $w=0$ is the marked vertex, then the success probability at time $t=0$ of $w$ is $1/2$. This shows that the optimal choice is to not move at all, which gives the complexity $\order{1}$. 

Since the full eigendecomposition is known, one can estimate manually the proper measurement time to find any of the sink vertex. However, it is possible to use a lemma proved in \cite{chakraborty2016spatial} and improved in \cite{chakraborty2020optimality} instead.
\begin{lemma}[\cite{chakraborty2016spatial,chakraborty2020optimality}] \label{lem:the-search-lemma}
	Let $H$ be a Hamiltonian with eigenvalues $\lambda_1\geq\dots\geq\lambda_n$
satisfying $\lambda_1=1$ and $c\coloneqq \max_{i\geq 2}|\lambda_i| <1$ for all
$i>1$ with corresponding eigenvectors $\ket
{\lambda_1},\ket{\lambda_2},\dots,\ket{\lambda_n}$ and let $\ket{w}$ be another
state lying in the same quantum system. For an appropriate choice of
$r\in[-\frac{c}{1+c},\frac{c}{1-c}]$, the starting state $\ket{\lambda_1}$
evolves by the Schr\"{o}dinger's equation with the Hamiltonian $(1+r)H
+\ketbra{w}{w}$ for time $t = \Theta(\frac{1}{\braket{\lambda_1}{w}})$ into the
state $\ket f $ satisfying $|\braket{w}{f}|^2\geq \frac{1-c}{1+c}+o(1)$.
\end{lemma}

The adjacency matrix $A= A(K_{n-1,1})$ does not fulfill the requirement of the
lemma, because the largest eigenvalue is not equal to 1. In order to satisfy
$\lambda_1=1$ one can simply take $\frac{1}{\sqrt{n-1}} A$, however still one
does not have a separation between $\lambda_1$ and $\lambda_i$, since
$\lambda_{n}=-1$ and by this $c=1$. However, adding a scaled identity matrix does
not change the quantum evolution. By this,  transformation
$\frac{2}{3\sqrt{n-1}}(A+ \frac{\sqrt{n-1}}{2} \Id)$ maps eigenvalues
$-\sqrt{n-1}$, $0$, $\sqrt{n-1}$ to $-1/3$, $1/3$, $1$ giving $c=1/3$. In
general, applying a shifting and rescaling transformation
\begin{equation}
H_G \mapsto H_G'\coloneqq \frac{H_G - \frac{\lambda_2(H_G)+\lambda_n(H_G)}{2}\Id}{\lambda_1(H_G) - \frac{\lambda_2(H_G)+\lambda_n(H_G)}{2}}
\end{equation}
transforms $H_G$ to a new Hermitian operator with $|\lambda_2(H_G')|=|\lambda_n(H_G')|$ and $\lambda_1(H_G')=1$.

Using this fact we can finally show the optimality of the star graph for leaves.
Using the shifting and rescaling transformation we have $c=1/3$. Based on
Lemma~\ref{lem:the-search-lemma} after time $T=\Theta(2\sqrt{n-1}) =
\Theta(\sqrt{n})$ we obtain a state with the probability of measuring the marked
state at least $|\braket{w}{f}|^2\geq \frac{1-c}{1+c}+o(1) = \frac{1}{2}+o(1)$.

The star graph is an example of a graph where all vertices can be found within the
time $\order{ \sqrt{N}}$, except the single vertex which can be found in
$\order{1}$ time. Note that it does not violate the $\Omega(\sqrt{N})$ bound for
quantum search \cite{boyer1998tight,grover1996fast}, as there is only $1=o(n)$
vertex with the time complexity below the bound. On the other hand, the result
from \cite{boyer1998tight} is applicable only for uniformly random chosen
vertex, and in such case for the star graph, the vertices can still be found in expected time
$\Theta(\sqrt{n})$.

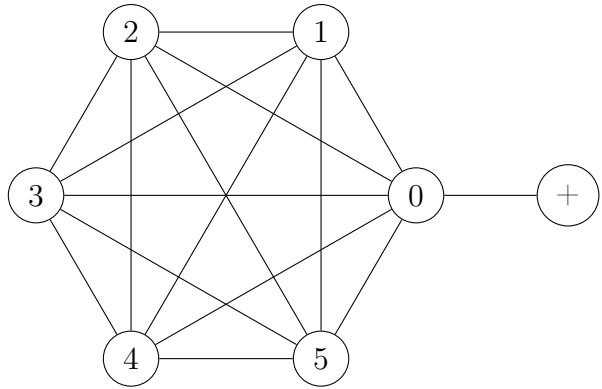
\begin{figure}\centering
	\begin{tikzpicture}
	\def \n {6}
	\def \radius {2.5cm}
	\def \margin {8} % margin in angles, depends on the radius
	
	\foreach \s in {5,...,0}
	{
		\node[draw, circle] (\s) at ({360/\n * (\s )}:\radius) {$\s$};
	}
	\draw [] (0)  edge (1);
	\draw [] (0)  edge (2);
	\draw [] (0)  edge (3);
	\draw [] (0)  edge (4);
	\draw [] (0)  edge (5);
	\draw [] (1)  edge (2);
	\draw [] (1)  edge (3);
	\draw [] (1)  edge (4);
	\draw [] (1)  edge (5);
	\draw [] (2)  edge (3);
	\draw [] (2)  edge (4);
	\draw [] (2)  edge (5);
	\draw [] (3)  edge (4);
	\draw [] (3)  edge (5);
	\draw [] (4)  edge (5);
	
	\node[draw,circle] (+) at (4.5, 0) {+};
	\draw [] (0) edge (+);
	\end{tikzpicture} 
	\caption{\label{fig:graph-hiding} An example of a graph for all vertices except `$+$' can be found in $\Theta(\sqrt n)$ time. The vertex `$+$' requires $\Theta(n)$ time.}
\end{figure}

It is also possible to find an opposite example, where some of the vertices
require significantly more time to be found. Let us consider a complete graph
$K_n^+$ with an extra leaf as in Fig.~\ref{fig:graph-hiding}. The
eigenvalues of the normalized Laplacian matrix $\mathcal L$ are
\begin{equation}
0, \frac{n-1}{n-2},\frac{1}{n-2}\left(2n-3 \pm \frac{\sqrt{23 -19n + 4n^2}}{\sqrt{n-1}}\right),
\end{equation}
with multiplicity 1, $n-3$, 1. The eigenvalues converge to 0, 1, and 2,
respectively. By this we have a constant spectral gap between 0 and 1 for operator $\Id -\mathcal L$, hence by
applying the shift and rescaling transformation we have that the overlap
$\braket{\lambda_1}{\omega}$ gives the required time in complexity. For
the graph matrix $\Id -\mathcal L$, the the eigenvector corresponding to the largest eigenvalue
takes the form
\begin{equation}
\ket{\lambda_1} = \sqrt\frac{2}{|E|}\sum_{v\in V} \sqrt{\deg(v)}\ket{v}.
\end{equation}
Note that for $K^+_n$ we have $|E| = \binom{n-1}{2}  + 1= \Theta(n^2)$.
Furthermore, for all vertices except  vertex  `$+$', the degree is $\Theta(n)$. Using 
Lemma~\ref{lem:the-search-lemma} for these vertices we have computational
complexity $T=\Theta(\sqrt{n})$. For the vertex `$+$' the complexity is
$\Theta(n)$. This gives us the opposite situation compared to the star graph. Another example where some vertices require more time compared to others can be found in \cite{philipp2016continuous}.

These simple examples show what we can expect when considering random graphs. In
\cite{chakraborty2016spatial}, the authors show that for almost all \ER graphs we can
find a vertex in optimal $\Theta(\sqrt n)$ time. However, one could expect that
even for a simple \ER model, some vertices may require significantly more or less
time compared to the typical scenario. In the following sections we focus on the
\ER model to show that this is not the case. However, we propose that
instead of using an adjacency matrix, which is far more robust, it seems to be more
convenient to use the Laplacian matrix.

\section{Adjacency matrix} 

\subsection{Issues found in the paper of Chakraborty et al.} 

In this section, we start by pointing the issues found in paper
\cite{chakraborty2016spatial} regarding the efficiency of quantum spatial search on
random \ER graphs. The authors showed three results.
First, they demonstrate that the quantum spatial search considered in this dissertation
is optimal on random \ER graphs. Then they show the application for creating Bell
pairs and state transfer on the same graphs. Our comments concern the first
part of the results. We would like to emphasize that the comments concern mostly
the quality aspects instead of the conceptual aspect, and do not diminish the results
given in \cite{chakraborty2016spatial}.

Let us start with the results. In the paper, the authors claim that
\begin{quote}\it 
CTQW is almost surely optimal as long as $p\geq \log^{3/2}(n)/n$. Consequently, we show that quantum spatial search is in fact optimal for almost all graphs, meaning that the fraction of graphs of $n$ vertices for which this optimality holds tends to one in the asymptotic limit.
\end{quote}
The authors show it through a simplified version of Lemma~\ref{lem:the-search-lemma}, which we recall below
\begin{lemma}\label{lem:the-search-lemma-orig}
		Let $H_1$ be a Hamiltonian with eigenvalues $\lambda_1,\geq\lambda_2\geq \dots \lambda_n$ (satisfying $\lambda_1=1$ and $|\lambda_i| \leq c <1$ for all $i>1$) and eigenvectors $\ket{v_1}=\ket{s}$, $\ket{v_2},\dots,\ket{v_n}$, and let $H_2 = \ketbra{w}{w}$ with $|\braket{w}{s}| = \varepsilon$. For an appropriate choice of $r=\order{ 1}$, applying the Hamiltonian  $(1+r)H_1+H_2$ to the starting state $\ket{s}$ for time $\Theta(1/\varepsilon)$ results in a state $\ket{f}$ with $|\braket{w}{f}|^2 \geq \frac{1-c}{1+c}+o(1)$.
\end{lemma}
Here, $\ket{s}$ denotes the superposition of states in canonical basis. From now we assume that $\lambda_i \geq \lambda_j$ for $j<i$. The main
difference between Lemmas~\ref{lem:the-search-lemma}
and~\ref{lem:the-search-lemma-orig} is the form of the principal eigenvector. In
the latter the eigenvector has to be an uniform superposition, while in the
former it can be an arbitrary vector. While both lemmas are correct, the authors overused them when applying to random graphs.

The authors presented a proof suggesting $|\braket{\lambda_1}{s}| = 1-o(1)$ provided that $p\geq \log^{3/2}n/n$. The authors claimed that based on this it is enough to show the optimality of the CTQW. However, based on the example from the introduction of this chapter, $K^+$ graph, we can see that a large overlap is not a guarantee of optimal search for \emph{all} nodes. In fact, one can show that at most  $n(1-o(n))$ nodes can be found optimally.
\begin{proposition}[\cite{kukulski2020comment}]\label{proposition:convergence-vectors}
	Let $\ket{\varphi_n}=\sum_{i=0}^{n-1}a_{i,n}\ket{i}\in \RR^{\ZZ_{\geq0}}$ and $\ket{s_n} = \frac{1}{\sqrt n} \sum_{i=0}^{n-1}\ket i$. Suppose $\braket{\varphi_n}{s_n}\to 1-o(1)$. Then there exists $I_n\subseteq \{1,\dots,n\}$ such that $|I_n| = n(1-o(1))$ and
	\begin{equation}
	\max_{i\in I_n} |\sqrt n a_{i,n}-1| = o(1).
	\end{equation}
	Furthermore $|I_n| = n(1-o(1))$ is tight in the worst case scenario.
\end{proposition}
\begin{proof}
Let $\ket{\varphi_n} = \alpha_n\ket{s_n} + \beta_n \ket{s_n^\perp}$, with $\ket{s_n^\perp}=\sum_{i=0}^{n-1} b_{i,n} \ket{i}$ being a normed vector. Let $I^c_\varepsilon(n) \coloneqq \{i\in \{1,\dots,n\} \colon | \sqrt n  a_{i,n}-\alpha_n| > \varepsilon\}$. Since $\sqrt n a_{i,n}-\alpha_n = \sqrt n \beta_n b_{i,n}$, we have 
\begin{equation}
1 =|\braket{s_n^\perp}|^2 = \sum_{i=0}^{n-1} |b_{i,n}|^2 \geq \sum_{i\in I^c_\varepsilon(n)} |b_{i,n}|^2 > \frac{\varepsilon^2}{n \beta_n^2} |I^c_\varepsilon(n)|, 
\end{equation}
hence $|I^c_\varepsilon(n)|<  n \beta_n^2/\varepsilon^2$. Let $I_n \coloneqq \{1,\dots,n\} \setminus I^c_{\sqrt{\beta_n}}(n)$. Then since $\beta_n =o(1)$, we have $I_n = n(1-o(1))$
\begin{equation}
\begin{split}
\max_{i\in I_n} | \sqrt{n}a_i - 1| &\leq \max_{i\in I_n} | \sqrt{n}a_i - \alpha_n| + |1-\alpha_n| \\
&\leq \sqrt{\beta_n } +|1-\alpha_n | = o(1).
\end{split}
\end{equation}

Let us now show that $|I_n|=n(1-o(1))$ is tight. Let $f(n)=o(n)$ and let $\ket{\varphi_n} = \frac{1}{\sqrt{n-f(n)}} \sum_{i=0}^{n-f(n)-1} \ket{i}$. Vector $\ket{\varphi_n}$ satisfies the assumptions of the theorem, yet the maximal $|I_n|$ is of order $n-f(n)$.
\end{proof}
A simple example for which the scenario described by the proposition above occurs is a graph over $n$ vertices, where $n-f(n)$ vertices form a complete graph, and the remaining $f(n)$ vertices are isolated. For such graphs, $n-f(n)$ vertices can be still found optimally in time $\sqrt{n-f(n)} = \sqrt{n}(1-o(1))$, while the isolated vertices need $\Theta(n)$ time to be found.

Furthermore, the authors incorrectly derived the condition on $p$. In the paper, they used Theorem~1.4 from \cite{vu2007spectral} and the result from \cite{furedi1981eigenvalues}, which states that
\begin{equation}
\max_{i>1}|\lambda_i| \leq 2\sqrt{p(1-p)n}(1+o(1)),
\end{equation}
and for sufficiently large $n$
\begin{equation}
\lambda_1/(np) \sim \mathcal{N}\left (1,n\sqrt{\frac{1-p}{p}} \right ),
\end{equation}
where $\mathcal{N}(\mu,\sigma)$ is a Gaussian distribution with mean $\mu$ and variance $\sigma^2$. However, the first one requires $p(1-p)=\Omega(\log^4n/n)$, while the latter requires $p\in(0,1]$ to be a fixed number. Hence, the proof presented in \cite{chakraborty2016spatial} was correct only for the fixed nonzero $p$. 

Finally, the authors falsely approximated transition rate $\gamma$ by $\frac{1+r}{\lambda_1} \approx 1/(np)$ to \ER graphs. They took the Hamiltonian of the form
\begin{equation}
- \ketbra{\lambda_1} - \ketbra{w}{w} + \frac{1}{\lambda_1} \tilde A,
\end{equation}
where $\tilde A + \lambda_1\ketbra{\lambda_1}= A$, hence the transition rate $\gamma$ equals $\frac{1}{\lambda_1}$. Provided with high probability $|\lambda_1/(np)-1| \leq \delta = \frac{1}{\sqrt{n}}$, using perturbation theory similarly as in \cite{chakraborty2016spatial} one can have
\begin{equation}
p(t) = \frac{1}{1+n\delta^2/4} \sin^2\left( \frac{\sqrt{\delta^2/4+1/n}}{2}t\right).
\end{equation}
The formula implies that by choosing proper $T=\Theta(\sqrt{n})$ we can achieve the constant success probability. This derivation was proven by approximating $\ketbra{\lambda_1}$ with $\ketbra{s}$. However, approximation $\ket{\lambda_1}=\alpha\ket{s}+\beta\ket{s^\perp}$ allows only the approximation of the form  $\ketbra{\lambda_1}= |\alpha|^2\ketbra{s} + o(1)H_{\lambda_1}$, where $\|H_{\lambda_1}\|=\order{1}$. Hence, the new transition rate would equal $\frac{1}{|\alpha|^2np}$, and it is not obvious how good is the approximation of $\alpha$. Furthermore, part $o(1)H_{\lambda_1}$ should have the spectral norm of order $\order{ 1/\sqrt n}$ so that we could use the perturbation theory.

%prohibits usage of perturbation theory, because the transition rate should be changed into $\gamma'=(1+r)\frac{1}{np}(1+\order{1/\sqrt n}(1+o(1)) = \frac{(1+r)}{np}(1+o(1))$. However $(1+o(1))$ is insufficient for to allow an efficient search, which was already shown even for a very simple cases \cite{novo}.

For these reasons, we will try to provide similar results concerning the optimality of quantum spatial search on random \ER graphs, by using the adjacency matrix. We will consider the efficiency in two contexts:
\begin{enumerate}
	\item When can the quantum search find almost all nodes optimally?
	\item When can the quantum search find \emph{all} nodes optimally (no-hiding property)?
%	\item when (almost) all nodes can be found in time $T=\pi\sqrt n/2$ with probability $1-o(1)$
\end{enumerate}

The first objective requires $|\braket{s}{\lambda_1}|=1-o(1)$, and
$\max_i|\lambda_i| = o(\lambda_1)$. Then by using
Lemma~\ref{lem:the-search-lemma}  and
Proposition~\ref{proposition:convergence-vectors} we can find $n(1-o(1))$ nodes
in $\Theta(\sqrt n)$ with $\Theta(1)$ success probability. The second objective
requires also $\| \ket s - \ket{\lambda_1}\|_\infty =o(1/\sqrt n)$. Here, the application of
Lemma~\ref{lem:the-search-lemma} is straightforward.

%The third and last objective requires all previous appropriate mentioned requirements and $\|\lambda_1/(np) -1| = o(1/\sqrt n)$. We will show a derivation confirming that the condition is sufficient. Here we choose $\gamma=(1+r)/(np)$ to find node in $\pi\sqrt n/2$ with $1-o(1)$ success probability.

\subsection{Quantum search is almost always optimal} \label{sec:adj-er}
First, let us show the convergence of eigenvalues. We will use the theorems from \cite{chung2011spectra}, originally defined for Chung-Lu model $\randg_n^{\rm CL}(\omega)$, written here for $\omega$ being all-$np$ vector.
\begin{theorem}[\cite{chung2011spectra}]\label{theorem:chung_eigs}
	Let $A$ be an adjacency matrix of a random \ER graph with parameter $p$. If $p>\frac{8}{9}\ln(\sqrt 2n)/n$, then, with the probability at least $1-1/n$ we have
	\begin{gather}
	|\lambda_1(A) - np | \leq \sqrt{8np\ln(\sqrt 2n)},\\
	\max_{i\geq 2} |\lambda_i(A) | \leq \sqrt{8np\ln(\sqrt 2n)}.
	\end{gather}
\end{theorem}

Before proving $|\braket{\lambda_1}{s}|=1-o(1)$, we require another technical lemma.
\begin{lemma}[\cite{chung2011spectra}]\label{lemma:chung_norm}
	Let $A$ be an adjacency matrix of a random \ER graph with parameter $p$. Let for $n$ sufficiently large $p >
\frac{8}{9}\ln(n)/n$. Then, with the probability at least $1-o(1/n)$, for $n$
sufficiently large  we have
	\begin{equation}
	\|A - \mathbb E A\| \leq \sqrt{8 np \ln(n)}.
	\end{equation}
\end{lemma}
The lemma comes from the proof of Theorem~1 \cite{chung2011spectra} by choosing $\varepsilon=2/n$ therein.

\begin{lemma}
	Let $A$ be an adjacency matrix of a random \ER graph with parameter $p$. Provided $p=\omega(\log n/n)$, we have asymptotically almost surely $|\braket{\lambda_1(A)}{s}| =1-o(1)$.
\end{lemma}
\begin{proof}
	The proof goes similar as in \cite{chakraborty2016spatial}. Let $\ket{\lambda_1} = \alpha \ket s + \beta \ket{s^\perp}$. it is sufficient to show that $\alpha\geq1-o(1)$. Note that we have
	\begin{equation}
	\begin{split}
	(A-\EE A) \ket{\lambda_1} &= \sum_i \lambda_i\ketbra{\lambda_i} - np\ketbra{s}{s} = \lambda_1 \ket{\lambda_1} - np\alpha \ket{s} \\
	&= (\lambda_1-np\alpha^2)\ket{\lambda_1} - np\alpha\beta \ket{s^\perp},
	\end{split}
	\end{equation}
	and by this $\|(A-\EE A) \ket{\lambda_1}\|^2 \geq (\lambda_1-np\alpha^2)^2$. Since $\|B \ket{\varphi}\| \leq \|B\|$ for any choice of $B$ and $\ket\varphi$, by Lemma~\ref{lemma:chung_norm} we have
		\begin{equation}
		(\lambda_1-np\alpha^2)^2 \leq 8 np \ln(n),
		\end{equation}
 and by this
	\begin{equation}
	\alpha^2 \geq \frac{2\sqrt{ np \ln(n)}}{np} - \frac{\lambda_1}{np} = 2\sqrt{\frac{\ln n}{np}} - \frac{np(1-o(1))}{np} = 1-o(1),
	\end{equation}
	where the transformations are valid a.a.s., and the first inequality comes from Lemma~\ref{theorem:chung_eigs}. By $\alpha^2=1-o(1)$ we have $\alpha = 1-o(1)$.
\end{proof}

Taking all into account, by Lemma~\ref{lem:the-search-lemma} and  Proposition~\ref{proposition:convergence-vectors}, provided $p = \omega(\log(n)/n)$, almost all nodes can be found optimally in time $\Theta(\sqrt n)$ with constant success probability. Note that since $|\braket{\lambda_1}{s}| = 1-o(1)$, we can start the evolution in $\ket{s}$ instead of $\ket{\lambda_1}$.

\subsection{No-hiding theorem}  \label{sec:adj-er-nh} In this section we will show that if $p = \omega( \log^3(n)/(n\log^2\log n))$, the quantum search is optimal for \emph{all} nodes for
almost all graphs. We will show this by proving $\|\ket{\lambda_1}-\ket{s}\| = o(1/\sqrt n)$ a.a.s. Then, by the direct application of the Lemma~\ref{lem:the-search-lemma} we obtain the result.

Here we follow the proof shown by Mitra \cite{mitra2009entrywise} which assumed $p\geq \log ^6 n/n$.  The proof of the following proposition can be found in App.~\ref{app:er_proof_principal_eigenvector}.
\begin{proposition}[\cite{glos2018vertices}] \label{proposition:eigevector_ER}
	Let $\ket{\lambda_{1}}$ be a principal eigenvector of an adjacency matrix of a random \ER graph with parameter $p$. For the probability $p = \omega \left( \log^3(n)/(n\log^2\log n) \right)$ and some constant $c>0$ we have
	\begin{equation}
	\Vert \ket{\lambda_1} -\ket{s} \Vert_{\infty} \leq c \frac{1}{\sqrt{n}}\frac{\ln^{3/2}(n)}{\sqrt{np} \ln(np)}
	\end{equation}
	with probability $1-o(1)$.
\end{proposition}

\subsection{Conclusions for adjacency matrix}

In Sections~\ref{sec:adj-er} and \ref{sec:adj-er-nh} we have shown two significant thresholds for $p$ which
determine the known behaviour of quantum search based on CTQW for adjacency
matrices. For $p=\omega(\log n/n)$, almost all nodes can be found in
optimal time with constant success probability. Under stronger condition
$p=\omega(\log^3(n)/(n\log^2\log(n)))$, all vertices can be found in optimal
$\Theta(\sqrt n)$ time.

One could ask what happens below these two thresholds. \ER graph is a very special model, with the connectivity threshold at $ \ln n / n $. This means that for any $\varepsilon>0$, for $p<(1-\varepsilon)\ln(n)/n$ there is almost surely at least one isolated vertex. For such a vertex, the amplitude does not change, hence the success probability is $1/n$ independently on the evolution time. Furthermore, if $np<1$  then all connected components are of order at most $\order{\log(n)}$, hence one cannot expect the success probability better than $\order{\log(n)/n}$.

Still, when $p>(1+\varepsilon)\ln n/n$, the graphs are almost surely connected. Hence,
there is a  gap between the derived threshold and the connectivity threshold.
Unfortunately, for the adjacency matrix we have to take care of all three parameters
required in Lemma~\ref{lem:the-search-lemma}: largest eigenvalue, spectral gap,
and principal eigenvector. In the next section, we will provide a better result
using the Laplacian matrix instead of the adjacency matrix.

\section{Laplacian matrix}

Let us consider the Laplacian matrix $L=D-A$. Provided the graph $G$ is connected, its Laplacian is a nonnegative matrix with a single zero eigenvalue $\lambda_n=0$, with the corresponding eigenvector being an uniform superposition $\ket{s}$. Hence, not only $\braket{\lambda_n(L)}{s}=1$, but also $\|\ket{\lambda_n(L)}- \ket s\|_\infty=0$. Therefore, the conditions for optimality for almost all nodes in fact already imply the optimality for all nodes, provided the graph is almost surely connected.

Note that for the Laplacian matrix we will consider $H_G = \Id - \gamma L$ matrix, in order to provide a spectral gap next to the largest eigenvalue as required by Lemma~\ref{lem:the-search-lemma}. Since $- L$ is a nonpositive matrix, a shifting and rescaling procedure will always be required.

Let us now show that all nodes can be found in $\Theta(\sqrt n)$ time using the Laplacian matrix as long as $p=\Omega(\log n/n)$. We will demonstrate it in two parts, first assuming $p=\omega(\log n/n)$, then for $p=p_0 \ln n/n$ for $p_0\in\RR_{\geq 0}$.

\subsection{Case $ p=\omega(\log n/n)$}
 Under the condition $ p=\omega(\log n/n)$ we have $\lambda_2(L/(np))=1+\order{\sqrt{\frac{\log n}{np}}} \sim1$ \cite{kolokolnikov2014algebraic}. To show that $H_G=-\gamma L$ for a proper choice of $\gamma$ satisfies Lemma~\ref{lem:the-search-lemma-orig}, it is enough to show that $\lambda_n/(np) \to 1$ as well.
 By Theorem~1.5 from \cite{bryc2006spectral}, if $\tilde L$ is a symmetric
matrix whose off-diagonal elements have two-points distribution with mean 0 and
variance $p(1-p)$ and $\tilde{L}_{ii} =  \sum_{j \neq i} \tilde{L}_{ij}$, then
\begin{equation}
\lim_{n\to\infty} \frac{\|
	\tilde L\|}{\sqrt{2np(1-p)\log n}}=1. \label{eq:max-eigenvalue-laplacian}
\end{equation}
Note that $p$ may depend on $n$. Hence, we can extend
the Corollary 1.6 from the same paper.

Let $L= \tilde L + \EE L$, where $\EE L$ is an expectation of a
random Erd\H{o}s-R\'enyi Laplacian matrix. $\EE L$ has a single 0 eigenvalue and
all of the others equal $np$. By this we have $ \|\EE L\| = np$.
Then we have
\begin{equation}
\left|\frac{\lambda_1(L)}{np}-\frac{\|\EE L\|}{np}\right| \leq \frac{\|L-\EE L\|}{np} = \frac{\|\tilde L\|}{np}\to 0
\end{equation}
where the limit comes from Eq.~\eqref{eq:max-eigenvalue-laplacian}, assuming
$p=\omega(\log n/n)$.

We have shown that $\lambda_1(L) \sim np$, $\lambda_{n-1}(L) \sim np$, and
$\lambda_n=0$. Note that for $H_G=\Id-L/(np)$, we have $\lambda_1(H_G)=1$,
$\lambda_2(H_G)  =o(1)$, and $\lambda_n =o(1)$. Since the principal eigenvector
of $H_G$ is a uniform superposition, we have that for the combinatorial Laplacian
\emph{all vertices} can be found in optimal $\Theta(\sqrt n)$ time with $1+o(1)$
success probability.

\subsection{Case $ p=p_0\ln n/n$} 

Suppose $G$ is a random graph chosen according to
$\randgn[ER](p_0\frac{\ln(n)}{n})$  distribution, for $p_0>1$ being a constant.
Let $\delta_{\min}$ ($\delta_{\max}$) be a minimal (maximal) degree of a sampled
graph. Based on \cite{kolokolnikov2014algebraic} we can show that the algebraic
connectivity $\lambda_{n-1}(L)\sim \delta_{\min}$. Below we show similar results
for the largest eigenvalue. The proof can be found in
App.~\ref{app:er_proof_largest_eigenvalue}.

\begin{theorem}
Let $G$ be a random graph chosen according to $\randgn[ER](p)$. Let $p_0>0$ be
such that $np\geq p_0\ln(n)$. Let $\delta_{\max}\sim cnp$ for some $c>0$ almost
surely. Then almost surely $\lambda_1(L(G)) \sim cnp$ .
\end{theorem} 

Based on \cite{bollobas2001random} it can be shown that 
\begin{equation}
\delta_{\min} \sim (1-p_0) \left(W_{-1}\left  (\frac{1-p_0}{\ee p_0}\right)\right)^{-1}\ln(n)
\end{equation}
and
\begin{equation}
\delta_{\max} \sim (1-p_0) \left(W_{0}\left  (\frac{1-p_0}{\ee p_0}\right)\right)^{-1}\ln(n).
\end{equation}
 Here $W_{-1}$ and $W_0$ are Lambert W functions. 

Note that $\lambda_{1}(L) = \Theta(\lambda_{n-1}(L))$, which is sufficient to
show the optimality of the search. Indeed, let us consider graph matrix $H_G=\Id
- L/\lambda_1(L)$. Its largest eigenvalue equals 1, and the smallest one equals 0.
Since $\lambda_1(L)$ and $\lambda_{n-1}(L)$ grows with the same complexity, the
spectral gap for $H_G$ equals $1-\lambda_{n-1}(L)/\lambda_1(L)$, which is
constant. Applying Lemma~\ref{lem:the-search-lemma-orig} we obtain a lower bound for
success probability
\begin{equation}
\sim \frac{1-  \frac{\lambda_{n-1}(L)}{\lambda_1(L)}}{1+  \frac{\lambda_{n-1}(L)}{\lambda_1(L)}} = \frac{\lambda_1(L)-\lambda_{n-1}(L)}{\lambda_1(L)+\lambda_{n-1}(L)} \sim \frac{W_{-1}\left  (\frac{1-p_0}{\ee p_0}\right)-W_{0}\left  (\frac{1-p_0}{\ee p_0}\right)}{W_{-1}\left  (\frac{1-p_0}{\ee p_0}\right)+W_{0}\left  (\frac{1-p_0}{\ee p_0}\right)}.
\end{equation}
Note that the lower bound can be improved. Adding the identity matrix has no
impact on the quantum evolution, and rescaling of the form $H/b$ can be
compensated by the proper transformation of transition rate $\gamma$.

Without loss of generality, let $H'$ be a Hamiltonian with the largest
eigenvalue equal to 1. Let $H(a) = (H'+a \Id)/b$. We will search for such
$a,b\in \R$ that $H(a)$ will have a spectral gap as required in
Lemma~\ref{lem:the-search-lemma}, and so that $\frac{1-c}{1+c}$ is maximized.
Note that $b>0$ so that the order of eigenvalue is preserved. Furthermore,
$b=1+a$, as $H(a)$ should still have the largest eigenvalue equal to 1.

Let $\lambda_2,\lambda_n$ be the second largest and smallest eigenvalues of
$H'$. Note that the $c$ of $H(a)$ is defined as $c(a) \coloneqq \frac{\max
	\{|\lambda_2+a|,|\lambda_n+a|\}}{1+a}$. Since the success probability
$\frac{1-c}{1+c} =  \frac{2}{1+c}-1$ decreases in $c$, maximizing the
lower bound on the success probability is equivalent to minimizing $c$. Note that for
$a_{\rm thr}=-\frac{\lambda_2+\lambda_n}{2}$ we have $c(a)=
\frac{\lambda_2-\lambda_n}{2-\lambda_2-\lambda_n}$. For $a>a_{\rm thr}$, the
$\lambda_2+a$ plays a dominant role in the maximum in $c(a)$, hence
\begin{equation}
c(a) = \frac{\lambda_{2}+a}{1+a} = 1 + \frac{\lambda_2-1}{1+a},
\end{equation}
is a function decreasing in $a$ because $\lambda_2-1<0$. For $a< a_{\rm thr}$, the $\lambda_n+a$ plays a dominant role and we have
\begin{equation}
c(a) = -\frac{\lambda_{n}+a}{1+a} = -1 + \frac{\lambda_n-1}{1+a},
\end{equation}
which increases in $a$. This confirms that $c(a)$ has the global maximum at
$a_{\rm thr}$, which gives the optimal shifting and rescaling parameter. Note that
for optimal $a$ we have the success probability
\begin{equation}
\sim \frac{1-c(a_{\rm thr})}{1-c(a_{\rm thr})} = \frac{1-\frac{\lambda_2-\lambda_n}{2-\lambda_2-\lambda_n}}{1+\frac{\lambda_2-\lambda_n}{2-\lambda_2-\lambda_n}}= \frac{1-\lambda_2}{1-\lambda_n}.
\end{equation} 

Let us consider the Laplacian of \ER graphs, where $\lambda_{n-1}(L) \sim
(1-p_0) \left(W_{-1} \left(\frac{1-p_0}{\ee p_0}\right)\right)^{-1}\ln n$ and
$\lambda_{1}(L) \sim (1-p_0) \left(W_{0} \left(\frac{1-p_0}{\ee
	p_0}\right)\right)^{-1}\ln n $. Then for  $H_G=\Id -L/\lambda_1 $ we have
\begin{gather}
\lambda_1(H_G) = 1, \\
\lambda_2(H_G) \sim 1 - \frac{W_{0} \left(\frac{1-p_0}{\ee p_0}\right)}{W_{-1} \left(\frac{1-p_0}{\ee p_0}\right)},\\
\lambda_n(H_G) = 0.
\end{gather}
Using the optimal shift-and-rescaling technique we have that the success probability
can be upperbounded almost surely by
\begin{equation}
\sim \frac{1-c(a)}{1+c(a)}  = \frac{1-\lambda_2(H_G)}{1-\lambda_n(H_G)} =
\frac{W_{0} \left(\frac{1-p_0}{\ee p_0}\right)}{W_{-1} \left(\frac{1-p_0}{\ee
		p_0}\right)}.
\end{equation}
The function changes smoothly from 0 for $p_0=1$ to 1 for $p_0\to
\infty$, see Fig~\ref{fig:p0-theoretical}. This coincides with the intuition
behind the random \ER graphs. For probability $p<\frac{(1-\varepsilon)\ln n
}{n}$, the  graphs are almost surely disconnected, hence the Laplacian has multiple zero
eigenvalues, giving no spectral gap. On the other hand, when $p_0$
approaches $\infty$,  the value of $p$ becomes `closer' to $\omega(\log n/n)$ case which
was proved to attain full success probability. Thus $p_0\frac{\ln n}{n}$ is a
smooth transition case for the Laplacian matrix.

 \begin{figure}
	\centering
	\includegraphics{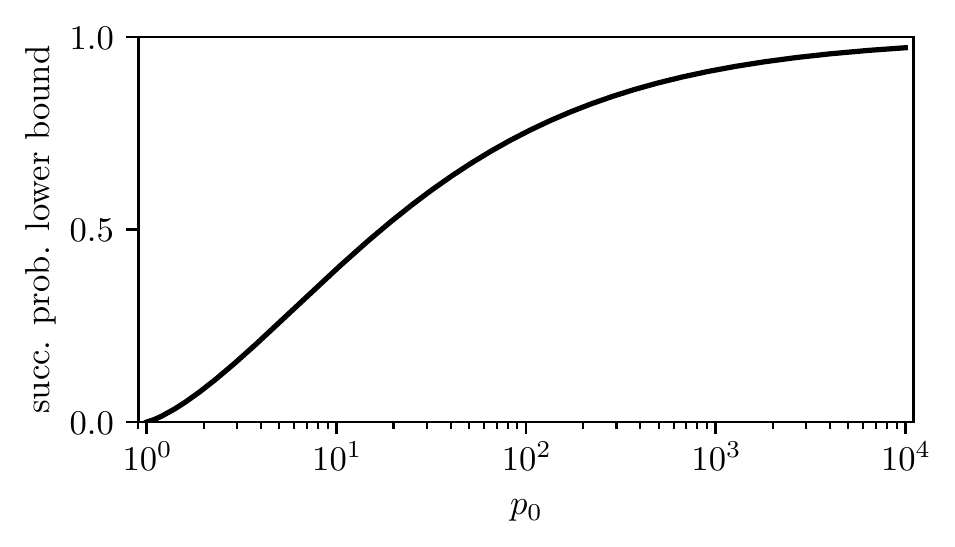}
	\caption{\label{fig:p0-theoretical} The limit lower bound of success probability of finding nodes for \ER graphs with parameter $p=p_0\frac{\ln n}{n}$. Note that the function changes smoothly from 0 for $p_0=1$ to 1 for $p_0\to \infty$. }
	\vspace{1cm}
	\centering
	\includegraphics{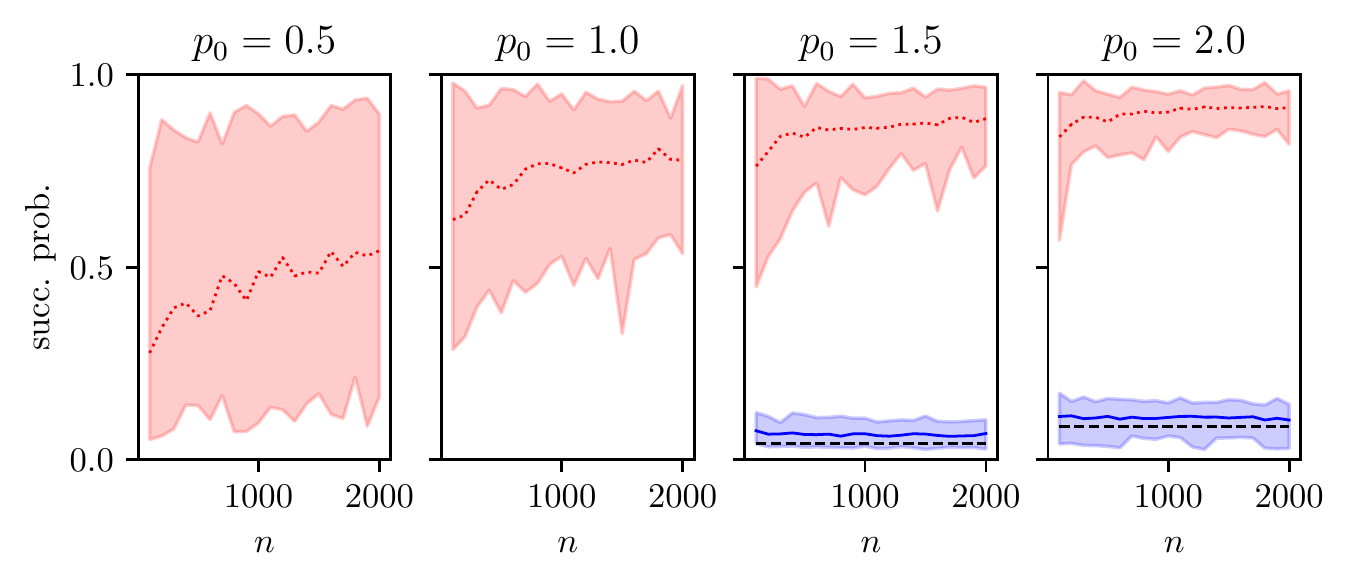}
	\caption{\label{fig:p0-statistics} The success probability and its lower bound for \ER graphs with $p=p_0\ln n/n$ with $H_G = \Id - L/\mu_1$. Dotted red line denote the mean success probability with $\gamma = \sum_{i\geq 2} \frac{|\braket{w}{\lambda_1}|^2}{1-\lambda_i} \Big / \sum_{i\geq 2} |\braket{w}{\lambda_1}|^2$, which is a transition rate proposed in \cite{chakraborty2016spatial}. The evolution time equals $\pi\sqrt n/2$. The blue solid line denotes $\frac{1-c}{1+c}$ where $c$ is the second maximal eigenvalue in absolute value of the optimally shifted and rescaled Laplacian. The dashed black line denotes the limit lower bound on the success probability computed according to the formula $W_{0} \left(\frac{1-p_0}{\ee p_0}\right)/W_{-1} \left(\frac{1-p_0}{\ee p_0}\right)$. Note that it is well-defined only for $p_0>1$. For each $p_0$ and $n=100,200,\dots,2000$, we sampled 50 graphs and its largest giant component was chosen. The areas span the minimum and maximum values obtained.}
\end{figure}

The derived value $\frac{W_{0} \left(\frac{1-p_0}{\ee p_0}\right)}{W_{-1}
	\left(\frac{1-p_0}{\ee p_0}\right)}$ is only a lower bound. The analysis of
$p=p_0\frac{\ln n}{n}$ for $p_0=0.5,1,1.5,2$ shows that not only for $p_0>1$ the
success probability is high, but even the nodes from a giant component retain
high success probability, see Fig.~\ref{fig:p0-statistics}. However, for $p_0
\leq 1$ almost surely there are isolated vertices which cannot be found in
$o(n)$ time.

\section{Conclusions} In this section we analyzed the efficiency of CTQW spatial
search for random \ER graphs. Our analysis considered both adjacency matrix and
the Laplacian matrix, however the obtained results favour the latter graph matrix. For
adjacency matrix, $p =\omega(\log^3(n)/n)$ is required at the moment to guarantee
that all vertices can be found in optimal $\Theta(\sqrt n)$ time with the constant
success probability. Relaxing the condition on $p$ to $p=\omega(\log n/n)$, we can
still find most of them.

For the Laplacian matrix, we can achieve constant success probability for \emph{all}
vertices already for $p>(1+\varepsilon)\ln n /n$ for any $\varepsilon>0$. This
is tight as for $p<(1-\varepsilon)\ln n/n$ sampled graphs almost surely have
isolated vertices which cannot be found by any reasonable quantum-walk based
search. While the results obtained for the adjacency matrix give only sufficient conditions and
could be theoretically improved, our consideration shows considerable advantage
of applying the Laplacian matrix over the adjacency matrix.

% !TeX spellcheck = en_GB
\chapter{Quantum spatial search on heterogeneous graphs} \label{sec:complex}
%intro to heterogeneous graphs
\chaptermark{Quantum spatial search on\ldots}

At this moment the state-of-the-art results concerning Childs and Goldstone
quantum spatial search can be found in \cite{chakraborty2020optimality}. Let
$H_G$ be a graph matrix with eigenvalues $\lambda_1=1 \geq \lambda_2 \geq \ldots
\geq \lambda_n=0$. Let $\ket{\lambda_i}$ be an eigenvector corresponding to
eigenvalue $\lambda_i$. Let $\Delta \coloneqq \lambda_1-\lambda_2$ be a spectral
gap. Let $w$ be a marked node and $\varepsilon\coloneqq
|\braket{w}{\lambda_1}|^2$. Finally, let
\begin{equation}
S_k \coloneqq \sum_{i=2}^n \frac{|\braket{w}{\lambda_i}|^2}{(1 - \lambda_i)^k}.
\end{equation}
Based on  Theorem~2 from \cite{chakraborty2020optimality}, if
\begin{equation} 
\sqrt \varepsilon < c \min \left ( \frac{S_1S_2}{S_3}, \Delta \sqrt {S_2}\right ) \label{eq:condition-general-spatial-search}
\end{equation} 
for sufficiently small $c$, then applying Hamiltonian $S_1H_G + \ketbra{w}{w}$
for time $T=\Theta(\frac{1}{\sqrt{\varepsilon}}\frac{\sqrt{S_2}}{S_1})$
transforms the initial state $\ket{\lambda_1}$ into $\ket{f}$ satisfying
$|\braket{w}{f}|^2 = \Theta(S_1^2/S_2)$.

Note that not all graphs satisfy the condition given in
Eq.~\eqref{eq:condition-general-spatial-search}
\cite{chakraborty2020optimality}. Furthermore, the condition, estimation on $T$ and
$\gamma=S_1$ requires knowledge about the marked node based on the definition of
$S_k$. Based on the introduction from the previous chapter and
\cite{philipp2016continuous,wong2016laplacian}, the optimal measurement time and
transition rate may depend on the marked node. While the same situation occurs in Lemmas~\ref{lem:the-search-lemma} and \ref{lem:the-search-lemma-orig}, statistics $S_k$ are not required in general for estimating the optimal measurement time or to detect the validity of the condition from Eq.~\eqref{eq:condition-general-spatial-search}.

In the following sections, we consider the choice of graph matrix $H_G$ and the
measurement time $T$. In particular, we consider a normalized Laplacian, which
was not widely considered in the theory of quantum search, except for a brief recall
in \cite{chakraborty2016spatial,chakraborty2020optimality}. Finally, we propose
an educated-guess method which does not require knowledge on the value of
optimal $T$.

One should note that there exists a continuous-time quantum walk model which was
proven to attain the quadratic speed-up over the corresponding markov
random walk \cite{chakraborty2020finding}. However, the model requires
a quadratically larger quantum system. 

\section{Choice of $H_G$ for heterogeneous graphs}

Neither adjacency matrix nor the Laplacian seems to be a good choice for governing the
efficiency of quantum spatial search. While there are known results concerning
the spectral gap for the adjacency matrix, the principal eigenvector takes usually
a complicated form. The principal eigenvector of the Laplacian matrix of a connected
graph is always an equal superposition, although there is a serious issue with
the spectral gap. Note that the largest eigenvalue satisfies $\delta_{\max} \leq
\lambda_1(L)\leq 2\delta_{\max}$
\cite{fiedler1973algebraic,anderson1985eigenvalues}, thus $\lambda_1(L) =
\Theta(\delta_{\max})$. However, it is known that the second smallest eigenvalue
satisfies $\lambda_{n-1} \leq \delta_{\min}$ \cite{fiedler1973algebraic}. Thus,
for the typical choice $\Id - L/\lambda_1$, we find out that the spectral gap is
at most $\order{\frac{\lambda_{n-1}(L)}{\lambda_1(L)}}$, which is at most
$\order{\frac{\delta_{\min}}{\delta_{\max}}}$. This bound can decrease
very rapidly. For example for \BA graphs, where the minimum degree equals $m_0$
and the maximum degree grows like $\Theta(\sqrt n)$ \cite{flaxman2005high}, the
spectral gap decreases like $\order{\frac{1}{\sqrt{n}}}$. This implies that
the necessary conditions from
Lemmas~\ref{lem:the-search-lemma} and \ref{lem:the-search-lemma-orig} are not
satisfied.

There are significant problems with using the adjacency matrix and the Laplacian as a
graph matrix. Before considering the normalized Laplacian, let us consider a uniform
random walk defined through matrix $P=D^{-1}A$. If the graph is connected the
mean first hitting time to vertex $v$ equals $\langle T_v \rangle =
\frac{|E|}{\deg(v)}S_1$, see App.~\ref{app:classical_search} for derivation.
Provided the spectral gap $\lambda_1(P)-\lambda_2(P)$ is constant, we have
$\langle T_v \rangle =\Theta(\frac{|E|}{\deg v})$. This is optimal, as $\langle
T_v\rangle \geq \frac{|E|}{d_j} - \frac{1}{2}$ for any choice of a connected graph
and $v$.

Closely correlated to the stochastic matrix $P$ is the normalized Laplacian
$\mathcal L = \Id -D^{-1/2} A D^{-1/2}$. The normalized Laplacian is a
nonnegative matrix with the spectrum lying in $[0,2]$ interval. Furthermore, provided
the graph is connected, the matrix has a single 0-eigenvalue with the corresponding
eigenvector
\begin{equation}
\ket{\lambda_1} = \frac{1}{\sqrt{2 |E|}} \sum_{v\in V} \sqrt{\deg(v)} \ket {v}.
\end{equation}
Note that for a connected graph $D$ is invertible and 
\begin{equation}
D^{-1/2}\mathcal{L} D^{1/2} = \Id- D^{-1} A = \Id - P,
\end{equation}
hence the normalized Laplacian is similar to the corresponding stochastic matrix
up to transformation $x\mapsto 1-x$. Hence the spectral gap is the
	same for $\Id -\mathcal{ L}$ and $P$. Furthermore, by
Lemma~\ref{lem:the-search-lemma} one can find vertex $v$ for the normalized
Laplacian in time $\Theta(1/\sqrt{\varepsilon}) =\Theta(\sqrt{\frac{|E|}{\deg
		v}})$ which is the square root of classical hitting time.

One could expect that similarly as it was for the Laplacian graph, the constant
spectral graph occurs very rarely. In the next subsection, we show that the
constant spectral gap occurs for heterogeneous random graph models,
including the paradigmatic \BA model.

\section{Special random graph models} 

\subsection{\CL graphs} 

In this section we will provide analytical evidence why the normalized Laplacian may
be a better graph matrix compared to the adjacency matrix. Since the
parametrization of \CL graphs lies in the $n$-dimensional space for $n$-vertex
graph, in this section we will consider a special parameter class $\omega_i =
n^{a+ \frac{i}{n}b}$ for $i=1,\dots n $ and $0<a<a+b\leq 1$. Note we assume
$b>0$, hence our parametrization \emph{does not} generalize \ER model. 

Let us first consider the adjacency matrix based graph matrix $\frac{1}{\lambda_1(A)}A$. In order to apply Lemma~\ref{lem:the-search-lemma}, one has to ensure there is a constant gap between the largest eigenvalues of $H_G$ (equal to 1 in here), and the maximum in absolute value over the remaining eigenvalues. Let $\tilde d = \frac{\|\omega \|_2^2}{\|\omega\|_1}$. Given the maximum expected degree $\omega_{\max}\coloneqq\max_i \omega_i > \frac{8}{9} \ln (\sqrt 2 n)$, one can show that asymptotically almost surely \cite{chung2011spectra}
\begin{gather}
|\lambda_{1} - \tilde d | \leq \sqrt{8\omega_{\max } \ln (\sqrt 2 n)},\\
\max_{i\geq 2} |\lambda_i| \leq \sqrt {8 \omega_{\max} \ln(\sqrt 2 n)}.
\end{gather}
Given $\sqrt {\omega_{\max} \ln(n)}/\tilde d = o(1)$ one has a constant spectral gap for $\frac{1}{\lambda_1(A)}A$. It is difficult to analyse the fraction in general, since the behaviour of $\tilde d$ strongly depends on the form of $\omega$. In the case of the proposed $\omega$ parameter class, for any valid $a,b$ we have $\omega_{\max} = n^{a+b}$ and $\tilde d =\frac{n^{1+2a+2b}}{2b\log n} \frac{b\log n}{n^{1+a+b}} = \frac{1}{2}n^{a+b}= \omega(\sqrt{n^{a+b}\ln n})$, see App.~\ref{sec:cl-proofs-pnorm}. Hence, based on Lemma~\ref{lem:the-search-lemma}, the time required to maximize the success probability for vertex $w$ is of order $\Theta(1/|\braket{w}{\lambda_1}|)$.

Let $\ket{\omega} = \frac{1}{\|\omega\|_2} \sum_i \ket{\omega_i}$. Similarly as it was for \ER graphs \cite{chakraborty2016spatial}, it can be shown that provided $\sqrt {\omega_{\max} \ln(n)}/\tilde d =o(1)$ we have $\braket{\omega}{\lambda_1} = 1-o(1)$. Based on this fact one can show that in the case of the chosen parametrization, as long as $a<3b$, for almost all nodes $i$ we have $\braket{i}{\lambda_1} \sim n^{a+\frac{i}{n-1}b} \big / \|\omega\|_2$, see App.~\ref{app:cl-proofs-eigenvector} . Hence for almost all vertices we have the time complexity 
\begin{equation}
  \|\omega\|_2 \big / n^{a+\frac{i}{n-1}b} \sim \sqrt \frac{n^{1+2a+2b}}{2b\log n} \frac{1}{n^{a+\frac{i}{n}b}}= \Theta\left (\sqrt{\frac{n}{\log n}}n^{b(1-\frac{i}{n})} \right ).
 \end{equation} 
The complexity strongly depends on $i$, hence the marked vertex. For $i=n$, \ie
the node with maximal expected degree, we have complexity $\Theta(\sqrt
\frac{n}{\log n} )$. Contrary for $i=1$ we have time complexity
$\Theta(\sqrt\frac{n}{\log n} n^{b})$ (note $n^{\frac{1}{n}}\sim 1$).

In the case of the normalized Laplacian, as long as the minimum expected degree $\omega_{\min} \coloneqq \min _i \omega_i=  \omega (\log n)$, positive eigenvalues satisfy
\begin{equation}
\max_{i\geq 2}|\lambda_i -1|  \leq 3 \sqrt{\frac{6 \ln (2n)}{\omega _{\min}}} = o(1).
\end{equation}
For the considered parameter family we have $\omega_{\min}>n^a = \omega(\log n)$. Furthermore, for the normalized Laplacian the minimum eigenvalue equals 0, hence $\Id -\mathcal L$ satisfies the requirements from Lemma~\ref{lem:the-search-lemma}. Hence the complexity of finding the marked node is $\Theta(\sqrt{|E|/\deg(w)})$, where $|E|$ and $\deg(w)$ are the number of edges and the degree of $w$ of \emph{sampled} graph. However, with  probability $1-o(1)$ for the considered $\omega$ we have almost surely $|E| \sim \EE E = \frac{n^{1+a+b}}{b\log n}$ and for the degree we have almost surely $\deg(i)\sim \omega_i = n^{a+\frac{i}{n}b}$, see App.~\ref{app:cl-proofs-edges} and \ref{app:cl-proofs-degree}. Hence the complexity for the normalized Laplacian is 
\begin{equation}
\sqrt{\frac{|E|}{\deg(i)}} \sim \sqrt{\frac{n^{1+a+b}}{b\log n} \frac{1}{n^{a+\frac{i}{n}b}}} = \Theta \left(\sqrt {\frac{n}{\log n}} \sqrt{n^{b(1-\frac{i}{n})}} \right )
\end{equation}
For $i=n$ we have complexity $\Theta(\sqrt\frac{n}{\log n})$ while for $i=1$ the
complexity equals $\Theta(\sqrt\frac{n}{\log n} \sqrt{n^{b}})$ hence it is better by
$\sqrt{n^b}$ compared to the adjacency matrix. According to our
derivation presented in the previous section, the classical search is quadratically
slower compared to quantum search using the normalized Laplacian.

In Fig.~\ref{fig:cl-visualization} we present a visualization of complexities
for various values of $a$ and $b$. Note that for any node the complexity for
quantum search with the adjacency matrix is between the quantum search with the
normalized Laplacian and the classical search. However, the worst case scenario
for the normalized Laplacian is at worst equal to the optimal measurement time
of the classical search. This is not the case for the adjacency matrix. For
$b>0.5$ the hardest to find node yields the complexity worse than many nodes for
the classical search. In the extreme case for $b\approx 1$ the complexity is worse for almost half of the nodes.

This example shows that the choice of the graph matrix has a crucial impact on the
efficiency of CTQW-based quantum search. Furthermore, we found a large nonregular
family of graphs, for which the quantum search is proved to be quadratically
faster than the classical random walk search.

\begin{figure}\centering 
	\includegraphics{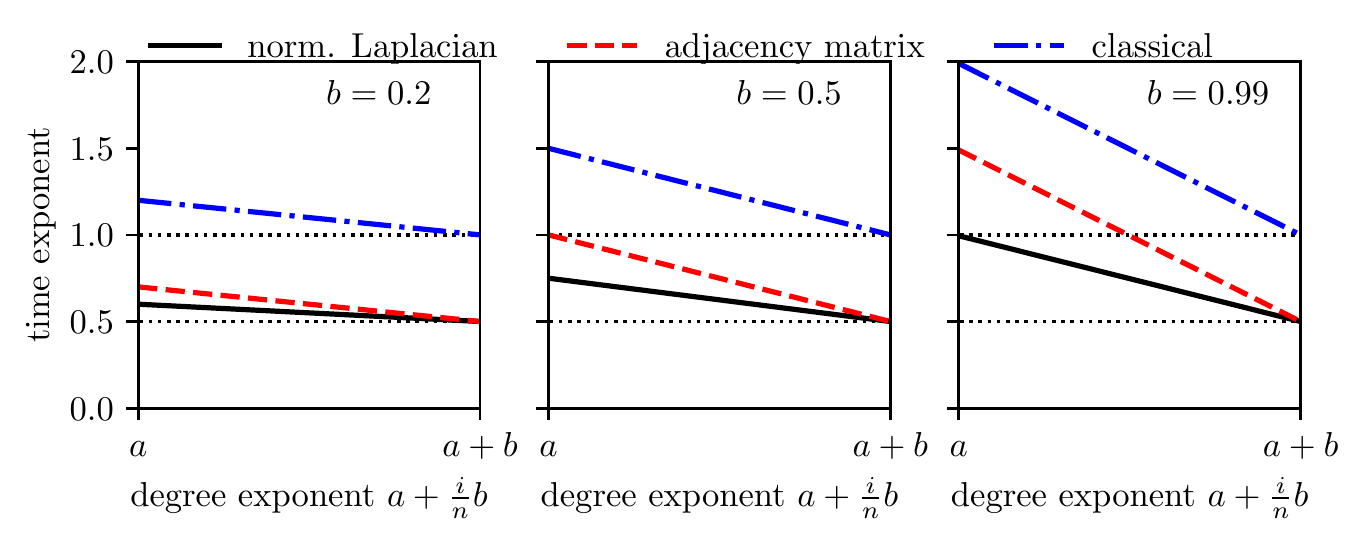}
\caption{\label{fig:cl-visualization} The visualization for time complexity
	depending on the expected degree of the vertex. The visualization is made for
	$b=0.2,0.5,0.99$, for we can choose $a=0.01$. Note that independently on the
	model chosen the time complexity decrease with the increase of the degree.
	Horizontal dotted lines are for readability only and are placed for general
	optimal quantum and classical search.}

\end{figure}

\subsection{\BA graphs} In this subsection we consider a \BA graph model which
is the paradigmatic random graph model for sampling complex networks. Let us
consider the efficiency of quantum search using the normalized Laplacian. Let us
start with proving the existence of a spectral gap. According to
\cite{durrett2006random}, we have $\frac{h^2}{2} \leq \lambda_{n-1}(\mathcal L)
\leq 2h$, with $h\geq0$ defined as in Section 6.2 in \cite{durrett2006random}.
Since $\frac{h^2}{2}\leq 2h$, we have $h\leq 4$. Hence if
$h=\Omega(1)$ we immediately have $\lambda_{n-1}=\Theta(1)$. In the case of \BA
model we have $h\geq\frac{\iota_a}{2m_0 + \iota_a}$ with $\iota_a$, see Section
6.4 in \cite{durrett2006random}. By Lemma 6.4.4 from \cite{durrett2006random} we
have $\iota_a=\Omega(1)$ for $m_0>1$. All the inequalities above show that
indeed $\lambda_{n-1}(\mathcal L)=\Theta(1)$ which gives a constant spectral
gap. This is sufficient to apply Lemma~\ref{lem:the-search-lemma}.

Now let us analyse the number of edges and degrees of nodes. By the very
construction, for any fixed $m_0$ the number of edges for $n$-vertex \BA graph
is at most $nm_0$. Since \BA graphs are connected, the number of edges is at
least $n-1$. From this we have $|E|=\Theta(n)$. The last added vertex has the degree
between 1 and $m_0$, hence for fixed $m_0$ the smallest degree is constant.
Contrary, the largest degree grows like $\Theta(\sqrt n)$ \cite{flaxman2005high}.
Hence the search complexity varies between $\Theta(\sqrt n)$ for constant degree
nodes and $\Theta (\sqrt[4] n)$ for the highest degree nodes.

In the case of \BA we obtained a complexity below $\Theta(\sqrt
n)$ lower bound for high degree vertices. However, since there are only
$\Theta(n)$ edges, then almost all nodes would have a fixed degree. Hence
better-than-optimal complexity happens only in rare cases. 

For the classical search we obtain analogical results: the time complexity
goes from $\Theta(\sqrt n)$ for high-degree nodes to $\Theta(n)$ for constant
degree nodes.

In the case of the adjacency matrix there is a lack of analytical derivation of
a spectral gap. For the Laplacian, following the reasoning presented previously we can
show that the spectral gap for $\Id -L/\lambda_1(L)$ is at most $\order{1/\sqrt
	n}$. Still it is possible to approach the matrices of this graph numerically thanks to
the recent results presented in \cite{chakraborty2020optimality}, recalled at the
beginning of this chapter. 

Let us start with investigating few examples for each graph matrix, see
Fig.~\ref{fig:ba-numerical-max} and \ref{fig:ba-numerical-min}. As we can see,
if the first node is marked then for the normalized Laplacian and the adjacency matrix, the
success probability stays roughly at $\Theta(1)$, and the time grows steadily.
For the Laplacian matrix the success probability decreases as $n$ increases.
Similarly, expected time $T/p(T)$ increases far more rapidly compared to the normalized
Laplacian and the adjacency matrix. For the last node, the success probability was
stable for all graph matrices. For both the normalized Laplacian and the Laplacian the
expected time grows similarly fast. However, for the adjacency matrix we observe
far more robust behaviour compared to any case considered so far.

\begin{figure}[ph!]
	\centering \includegraphics[scale=.8]{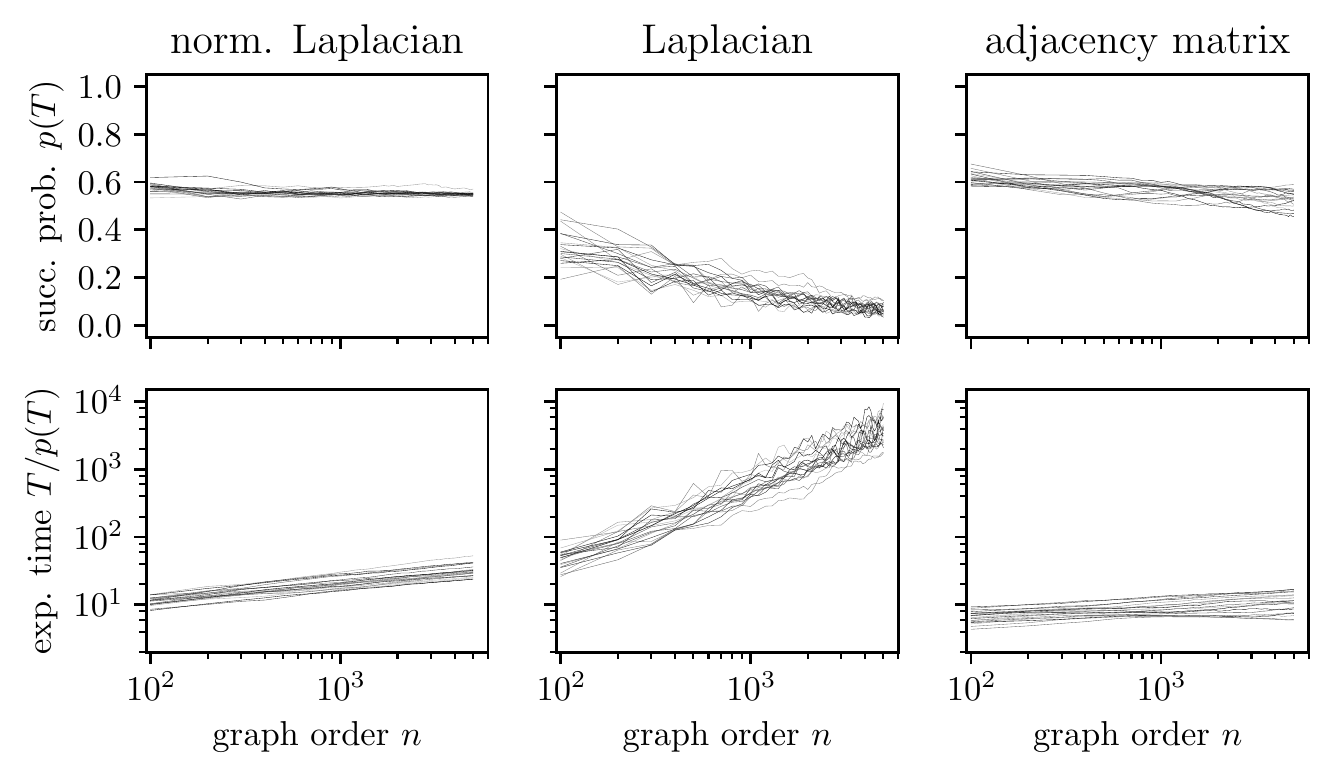}
	\caption{\label{fig:ba-numerical-max} Analysis of search efficiency of
		$\randg[BA](3)$ where the first node is marked. For each graph matrix we
		sampled 20 trajectories of graphs with orders 100, 200, \ldots, 5000. We chose
		a measurement time $T=\frac{1}{|\braket{w}{\lambda_1}}\frac{\sqrt {S_2}}{S_1}$
		and transition rate equal to $S_1$. The first row presents the trajectory of success
		probability calculated at $T$  for each trajectory of graphs. The second row
		presents the trajectory of $T/p(T)$.}
	
	\vspace{1cm}
	%\end{figure}
	%
	%\begin{figure}[t]\centering
	\includegraphics[scale=.8]{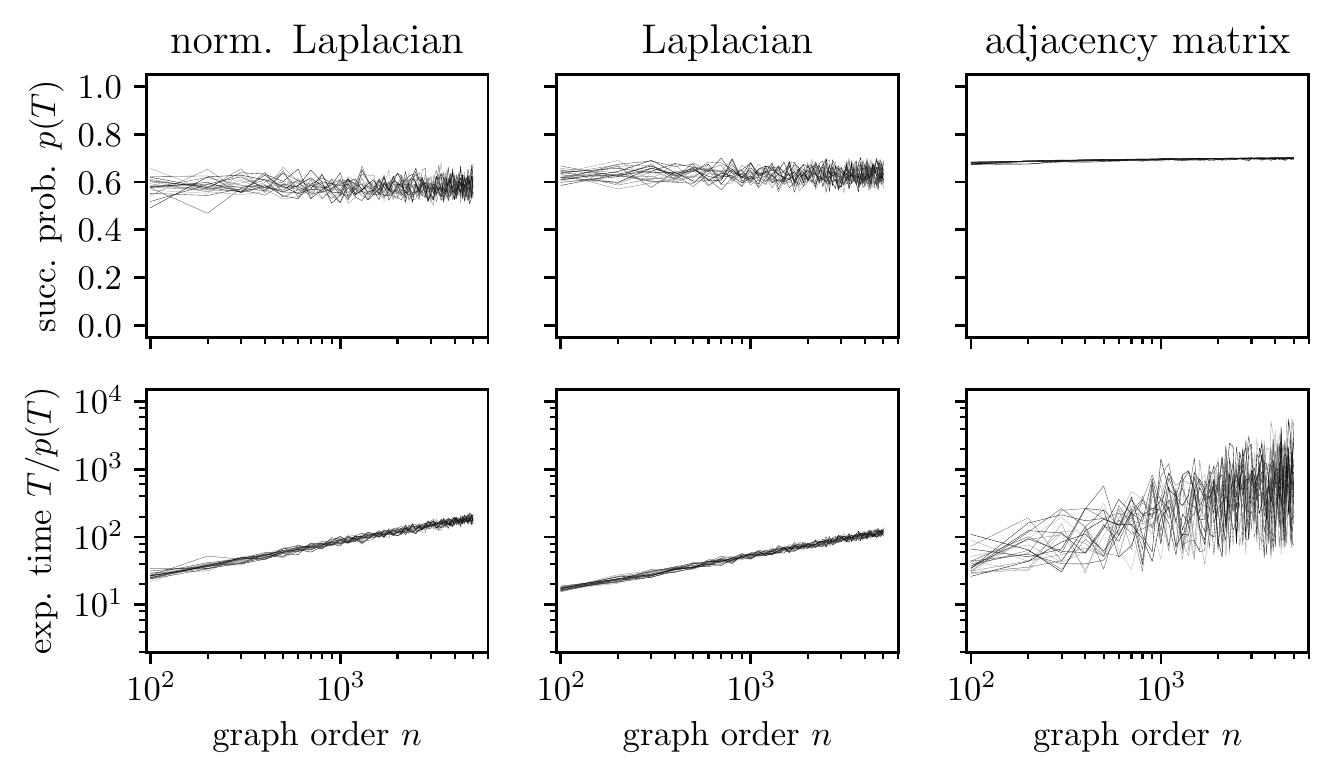}
	\caption{\label{fig:ba-numerical-min} Analysis similar to the one presented in Fig.~\ref{fig:ba-numerical-max}, except that the last node is marked. }
\end{figure}

The above observations are confirmed by the statistics of exponent $\alpha$
defined as $T/p(T) = \Theta(n^\alpha)$, see Fig.~\ref{fig:ba-numerical-max-exp}
and \ref{fig:ba-numerical-min-exp}. We can see that for both scenarios of the marked
nodes, the statistics for the normalized Laplacian reflect our predictions. In the case
of the first node being marked, the Laplacian had a complexity mostly $\Omega(n)$, which
is worse even compared to the classical procedure. We would like to emphasize
that this may be due to incorrectly chosen transition rate $\gamma$ and
measurement time $T$, as the educated guess presented in
\cite{chakraborty2020optimality} may not be proper for the considered graph matrix. Finally, for the adjacency matrix we observe that the required time is
far better even compared to the normalized Laplacian, which is in opposition to what
was observed for \CL graphs.

Finally for the last node being marked, the Laplacian matrix has the same efficiency as
the normalized Laplacian, namely roughly $\Theta(\sqrt n)$. However, for the adjacency
matrix the efficiency was between optimal for quantum search $\Theta(\sqrt{n})$ and
classically optimal $\Theta(n)$. Hence while we observed a speed-up compared to
the classical search, clearly the normalized Laplacian seems to be better in this
scenario. It is worth to note that as we observed in
Fig.~\ref{fig:ba-numerical-min} the trajectory of optimal measurement time is
very robust, which is also reflected in Fig.~\ref{fig:ba-numerical-min-exp} in
the regression quality measure. 

Finally, let us see the counter-intuitive behavior of the Laplacian matrix. For both
the normalized Laplacian and the adjacency matrix, finding the node with a higher degree
was simpler compared to finding the node with small degree. Based on our educated
guess for the normalized Laplacian, we may expect a similar property for the classical
search. However, for the Laplacian matrix it is contrary -- while small degree nodes
can be found in $\Theta(\sqrt n)$ time, it seems to be difficult to find
higher-order nodes. There may be two explanations of this phenomena.
Firstly, our choice of transition rate and optimal measurement time is not good
for the first node. Secondly, since the initial state of the Laplacian matrix is a
uniform superposition of basic states, such state may promote typical-degree
cases, which in case of \BA are finite-degree nodes.

\begin{figure}[ph!]
	\centering \includegraphics{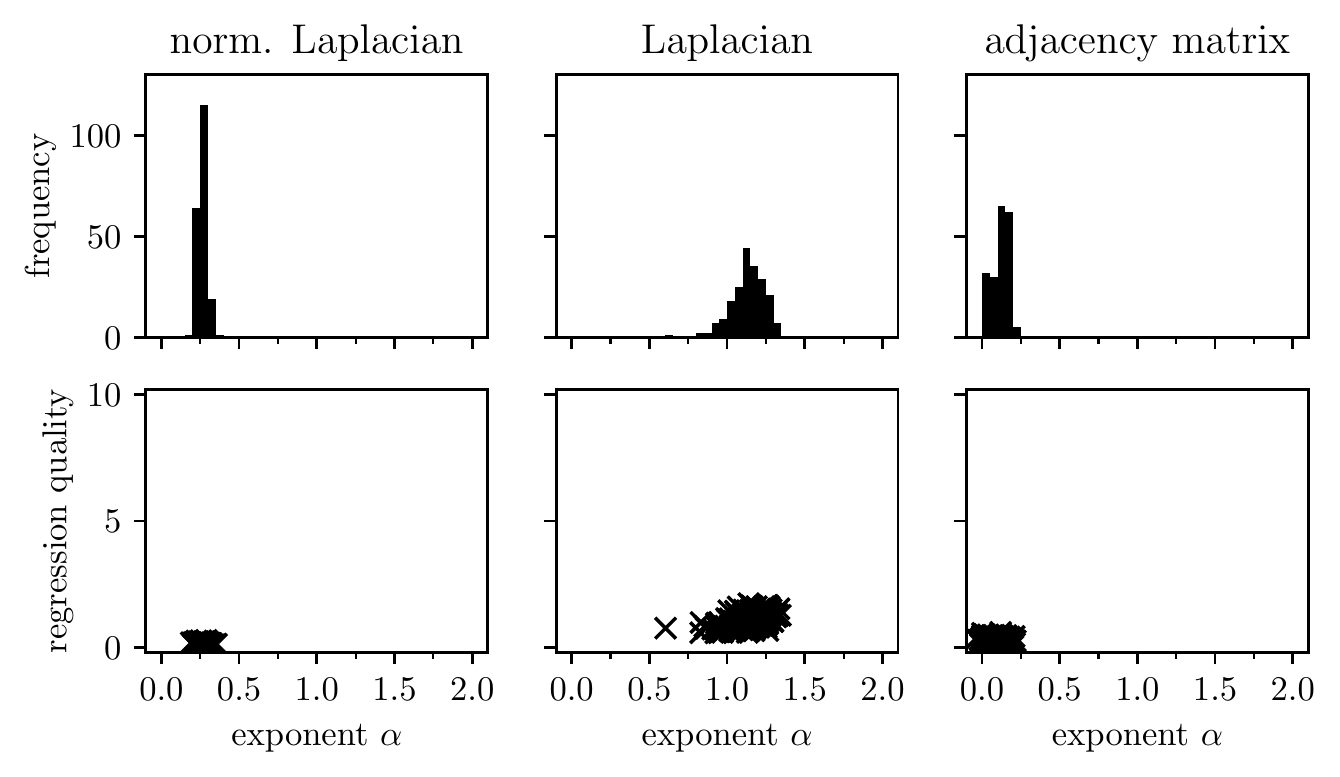}
\caption{\label{fig:ba-numerical-max-exp} Analysis of exponent $\alpha$ defined
	as $T/p(T)=\Theta(n^\alpha)$ for $\randg[BA](3)$ model. The quantum
	evolution is defined as in Fig.~\ref{fig:ba-numerical-max},
	and for each graph matrix we sampled 200 trajectories. The first row describes
	the statistics of exponents $\alpha$ calculated from linear regression fit
	$\log(T/p(T)) = \alpha \log n + \beta$. The second row presents a regression
	quality calculated based on the formula $\|\log(T) - \alpha \log n - \beta\|_2$. }
%\end{figure}
%
%\begin{figure}[t]\centering
	\includegraphics{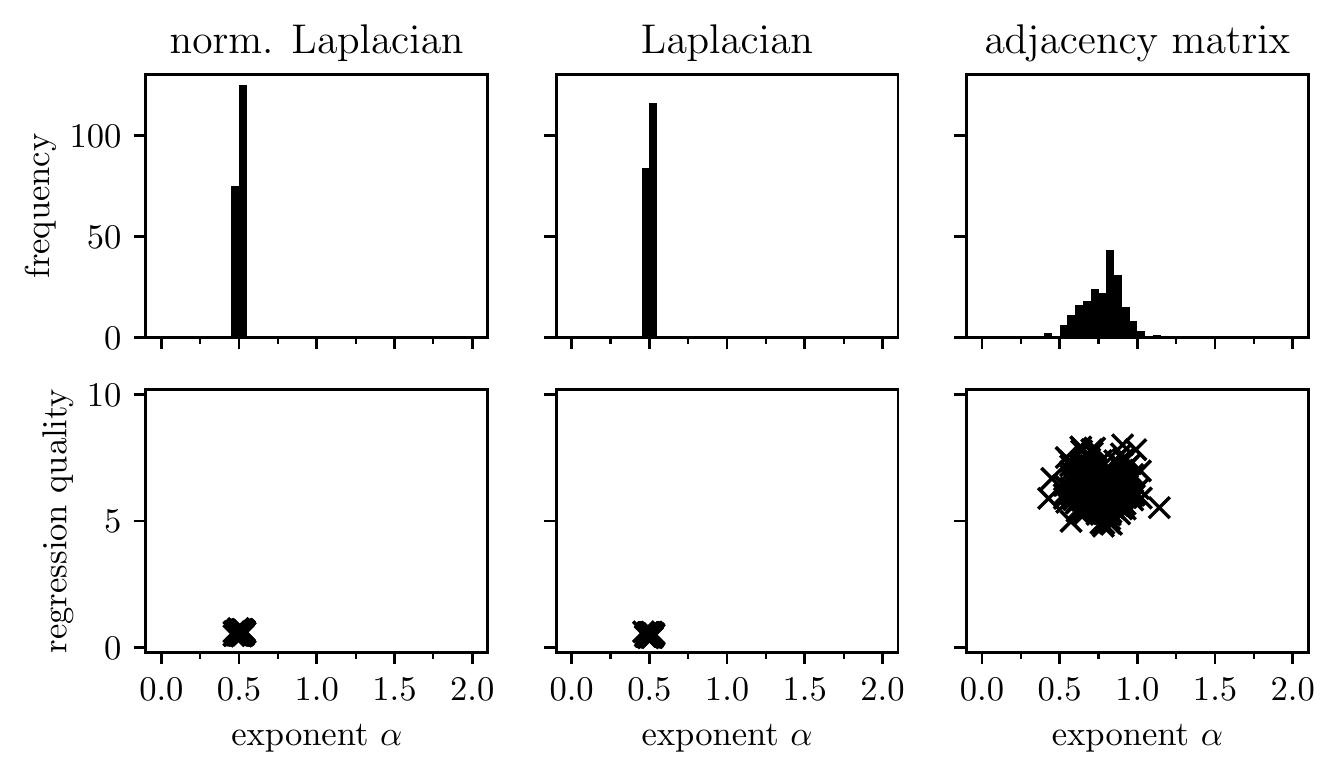}
	\caption{\label{fig:ba-numerical-min-exp} Analysis similar to the one presented in Fig.~\ref{fig:ba-numerical-max-exp}, except that the last node is marked. }
\end{figure}

\section{Optimal measurement time}

In order to make the CTQW-based search applicable one has to a priori determine
the optimal transition rate and the measurement time. However, as we have shown in
the previous section, the measurement time does not only depend on the sampled graph, but also
on the marked node. Note that the issue mentioned in the previous paragraphs is
typical for heterogeneous graphs \cite{osada2020continuous}, and was also
observed previously for very simple graphs \cite{philipp2016continuous}. For
vertex-transitive graphs, the choice of the transition rate and the measurement time does
not depend on the marked node, hence the analysis of the graphs is usually
sufficient to design the algorithm. However, in other cases it is less evident.

In this section, we will show that the knowledge about the optimal measurement time may not
be required, at least in the scenarios considered above. Suppose that the
graph-depending search procedure $\mathcal A_{G,w}(t)$ satisfies that for $t\geq
T_{\rm crit}(\mathcal A_{G,w})$ procedure $\mathcal A_{G,w}(t)$ finds the marked
node with a probability $1$. Also assume that $\mathcal A_{G,w}(t)$ runs for
time $t$. We will elaborate on the validity of this assumption for quantum search
algorithms at the end of this section.

If such procedure is encountered one can simply run the procedure for $\max_w
T_{\rm crit} (\mathcal A_{G,w})$ to find the arbitrary marked node with probability
1. However, if $T_{\rm crit}$ attains very different values (even in
complexity when increasing the number of nodes), such approach should be considered
as a waste of resources. For example based on Fig.~\ref{fig:cl-visualization},
unifying the measurement time to the most demanding node would destroy the quadratic
speed-up for almost all nodes in the case of the normalized Laplacian. For the
adjacency matrix for sufficiently large $b$ the time efficiency could be even
worse compared to the classical search for the same nodes.

Suppose we know that any node of the graph $G$ can be found by $\mathcal
A_{G,w}(t)$ for $t\in [Cn^{\beta_0},Cn^{\beta_0+\beta_1}]$, and that in
particular node $w$ can be found after time $Cn^\alpha$. Let $K\in \ZZ_{>0}$.
Then we can run $\mathcal A(t)$ for $t=Cn^{\beta_0}$,
$Cn^{\beta_0+\frac{1}{K}\beta_1}$, $\dots$, $Cn^{\beta_0+\beta_1}$. Based on our
assumptions on procedure $\mathcal A$, the marked node will be found by $\mathcal A
(Cn^{\beta_0 + \frac{k}{K}\beta_1})$ where $\beta_0 + \frac{k_w-1}{K}\beta_1 <
\alpha \leq \beta_0 + \frac{k_w}{K}\beta_1$. The whole procedure takes
\begin{equation}
\Theta\left (\sum_{k=0}^{k_w} Cn^{\beta_0 + \frac{k}{K}\beta_1} \right ) = \Theta(n^{\beta_0 + \frac{k_w}{K}\beta_1}).
\end{equation}
Instead of the optimal complexity $\Theta(n^{\alpha})$ we obtained a
complexity $\Theta(n^{k_w})$, hence the time complexity is increased by the factor $n^{\beta_1/K}$ at the worst.

We can make $n^{\beta_1/K}$ arbitrarily slow by increasing $K$.
For large, yet fixed $K$, the complexity depends only on the time required for
the longest run of $\mathcal A$, \ie $k=k_w$. However, for $n$-dependent $K$, the
overall time required for calculating $k<k_w$ may have impact on the final time
complexity.

Let us consider $ K=K' \log n$ for $K' >0$ being a real constant. The time
complexity of the whole procedure equals
\begin{equation}
\sum_{k=0}^{k_w} Cn^{\beta_0 + \frac{i}{K}\beta_1} = Cn^{\beta_0}\sum_{k=0} \left(n^{\beta_1/K}\right)^i = Cn^{\beta_0} \frac{n^{\frac{k_w}{K}\beta_1}-1}{n^{\beta_1/K}-1}.
\end{equation}
Note $n^{1/K} = \exp(1/K')$. Hence
\begin{equation}
\begin{split}
Cn^{\beta_0} \frac{n^{\frac{k_w}{K}\beta_1}-1}{n^{\beta_1/K}-1} &= Cn^{\beta_0} \frac{n^{\frac{k_w}{K}\beta_1}-1}{\exp(\beta_1/K')-1} \\
&\sim \frac{C}{\exp(\beta_1/K')-1} n^{\beta_0+\frac{k_w}{K}\beta_1} \eqqcolon f(k_w).
\end{split}
\end{equation}
By this we have
\begin{equation}
\frac{f(k_w+1)}{f(k_w)} = n^{\beta_1/K} = \exp(\beta_1/K')
\end{equation}
and $f(k_w+1) = \Theta(f(k_w))$. Furthermore, since $f(k_w-1) \leq
Cn^\alpha \leq f(k_w)$, we have that $n^\alpha = \Theta(f(k_w))$, hence our
procedure works optimally.

As we can see, the knowledge about the optimal measurement time is not required given certain
assumptions on the searching procedure $\mathcal A$. Let us now consider the
validity of the assumptions taken.  First, we assumed that the node can be
found in the interval $[Cn^{\beta_0}, Cn^{\beta_0+\beta_1}]$. Without loss of generality we
can assume $\beta_0=0$, as its value has no impact on the proof. For similar
reason, the value of $\beta_1$ is not required, although it is important to assure
that it is constant so that searching will be polynomial. At least for the normalized
Laplacian with a constant spectral gap it is guaranteed, as the probability of
measuring the the marked node $w$ at initial time $t$  equals $\frac{\deg
	w}{2|E|}\geq \frac{1}{2n^2}$. Hence, if preparation of the initial state
can be done polynomially fast, then search will take polynomial time as well.

The assumption that the success probability achieves one exactly may not be
significant. The success probability can be arbitrarily close
to one by repeating the internal procedure $\mathcal A$. In such scenario, the
probability of measuring incorrect nodes decreases exponentially. Since for
the graphs with a constant spectral gap there is a common lower bound
$\frac{1-\lambda_2}{1+\lambda_2}+o(1)$ on success probability, one can expect that at least for such
graphs the assumption is not meaningful.

Finally, we made an assumption that for $t\geq T$ the success probability is
constantly one. This assumption is not valid for two reasons: first, the
evolution is quasi-periodic which means that in large time regime the success
probability can equal zero multiple times despite the fact that the time is
greater than the optimal measurement time. This can be solved by measuring at
different measurement time $C_{\mathcal U}n^{\beta_0 + \frac{k}{K}\beta_1}$
where $C_{\mathcal U}$ follows the uniform distribution on interval  $[C -
\exp(\beta_1/K')/2, C + \exp(\beta_1/K')/2]$. Finally, for large time regimes,
the eigenvalues of small magnitude may have crucial impact on the evolution,
acting as a noise on the evolution defined by the oracle and the principal
eigenvector. This issue can be solved by choosing a sufficiently small $K'$ and
increasing the number of repeating $\mathcal A$ for each $k$.

We would like to stress out that the above consideration requires further
investigations. The first step towards verifying these conjectures would be
numerical confirmation of the proposed method. However, this is beyond the scope
of the thesis.

\section{Conclusions}

In this section we analyzed the efficiency of CTQW-based spatial search on
heterogeneous and complex graphs. We provided both analytical and numerical
evidence that the optimal measurement time depends strongly on the marked
vertex and the matrix graph. The first one is not surprising, as similar
situation is observed for the random walk search.

In the case of graph matrix, the normalized Laplacian provided the most stable,
always quadratic speed-up over the random walk search for the considered random graphs. Both the adjacency matrix and the Laplacian usually  offered the speed-up
compared to the classical search, although usually the normalized Laplacian turned out
to require even less computational resources. Moreover, for the Laplacian
matrix over \BA graph it seems that searching for a high-degree nodes takes more
time compared to small-degree nodes. We claim that this counter-intuitive result
deserves additional attention, in order to fully address the impact of the
choice of matrix graph on the efficiency of CTQW-based search.

Finally, we tackled the problem of choosing the optimal measurement time in case
it is not known even in complexity. This is particularly relevant for
heterogenuous graphs. We proposed a simple approach
which we believe can solve the problem in cases considered in this chapter.

%%%%%%%%%%%%%%%%%%%%%%%%%%%%%%%%%%%%%%%%%%%%%%%%%%%%%%%%%%%%%%%%%%%%%%%%%%%
%%%%%%%%%%%%%%%%%%%%%%%%%%%%%% ending %%%%%%%%%%%%%%%%%%%%%%%%%%%%%%%%%%%%%
%%%%%%%%%%%%%%%%%%%%%%%%%%%%%%%%%%%%%%%%%%%%%%%%%%%%%%%%%%%%%%%%%%%%%%%%%%%

% !TeX spellcheck = en_GB
\chapter{Final remarks} \label{sec:conclusions}

In this dissertation we have investigated a hypothesis claiming that there exist simple,
continuous-time quantum walk models which maintain interesting and crucial
properties of quantum walks for nontrivial graphs.

In Chapter~\ref{sec:nonmoralizing-qsw} we proposed and analyzed the
time-independent nonmoralizing quantum stochastic walk. The model was an
interpolation between two continuous-time models. It was a mixture of the
original Childs-Goldstone Continuous-Time Quantum Walk \cite{childs2004spatial}
and the non-unitary model introduced by Whitefield et al.
\cite{whitfield2010quantum}. The interpolated model was shown to be at least
superdiffusive (and likely ballistic) in the intermediate cases of
interpolation. Moreover, based on our investigation presented in
Chapter~\ref{sec:convergence-qsw}, this model still preserves the directed graph
structure well. This property is true also for the well-known local interaction
quantum stochastic walks.

There is still room for improvement for the presented results. First, the proposed model
was propagating fast, but at the cost of small amplitude transfer going in
the opposite direction than the graph structure. While based on the result
presented in Sec.~\ref{sec:digraph-structure-observance}, its significance seems to be
negligible, it is still an open question whether this transition can be removed
completely. In our opinion, it is not possible with the mapping from measurement
output to vertex set being fixed as proposed in \cite{montanaro2007quantum}. In
other words, the same measurement output would have to be interpreted as
different vertices based on time or other context.

Furthermore, it is an open question how to implement the nonmoralizing quantum
stochastic walk. Clearly, it should be possible to first transform the GKSL
evolution into standard Schr\"odinger equation, and eventually into the gate model.
Still, an effective procedure should be described and analyzed in order to
confirm that the simulation on a quantum computer is possible.
Alternatively, one could consider a physical process which directly implements the 
general quantum stochastic walk.

Finally, one could consider the algorithmic application of the introduced model. In
particular, heuristic optimization algorithms like simulated annealing or Tabu
Search are examples of algorithms which strongly rely on the concept of directed graphs. Indeed, these algorithms, defined as random walks over the objective
spaces, share the property that passing to solution with smaller objective value
(in the case of minimization procedure) is more likely than passing to solution
with higher objective value. In this context, primary task would be to
effectively encode the optimization problem into the introduced quantum walk.

In Chapters \ref{sec:hiding} and \ref{sec:complex} we considered the efficiency
of the first continuous-time quantum spatial search on non-trivial graphs. In
Chapter~\ref{sec:hiding} we analyzed \ER graphs, which was the first step to
more advanced graphs. We presented that Laplacian seems to be a valid choice for
almost-regular graphs, yielding full quadratic speed-up even close to the
connectivity threshold. It is worth to note that for both adjacency matrix and
Laplacian matrix it was in fact almost surely possible to find \emph{all} nodes
for still small value of the parameter $p$. In Chapter~\ref{sec:complex} we
showed that the normalized Laplacian is a far better choice for heterogeneous
graphs in most of the cases. Provided that the spectral gap of the graph is
constant, the normalized Laplacian provided a full quadratic speed-up over the random
walk search. Furthermore, we suggested the procedure which enable attaining the
optimal time complexity for finding the vertex even if the optimal measurement
time is not known.

It is still not evident what is the possible speed-up for other graphs for
CTQW search. In order to better understand the limitations and capabilities of
this simple model, it may be interesting to consider other random graph models. This would
simplify introducing more general theorems which hopefully would mostly depend on simpler
graph-theoretic properties. Note that currently the most general results
presented in \cite{chakraborty2020optimality} depend on the \emph{spectral} properties
of the graph matrix, which is far harder to describe for general graph
collections.

In this dissertation we have confirmed that fast quantum propagation is
possible with preserving the structure of directed graphs. Furthermore, the quantum
search defined on heterogeneous graphs like \BA or \CL graphs is still
quadratically faster with a careful choice of the graph matrix. Based on this, 
we claim that indeed simple quantum walk models maintain important properties
of quantum walks for nontrivial graph structures.

\appendix

\bibliographystyle{ieeetr}

\bibliography{rozprawa-adam}
\addcontentsline{toc}{chapter}{Bibliography}

% !TeX spellcheck = en_GB
\chapter{Proofs for Quantum Stochastic Walks} \label{sec:qsw-proofs}
\chaptermark{Proofs for QSW}

\section{Probability distributions of GQSW on finite and infinite paths}\label{sec:prop-qsw-prob-dist}
\subsection{Probability distribution for finite path}
\begin{proof}[Proof of Theorem~\ref{theorem:prob-on-path}]
	Let $L$ and $H$ be an operators defined according to the theorem.
	Using Eq.~\eqref{eq:integrated-qsw} we have
	\begin{equation}\label{eq:integrate_lind}
	\begin{split}
	M_\omega^t = \exp \left[t\omega \left( L\otimes L -\frac{1}{2} 
	L^2\otimes\Id -\frac{1}{2}\Id\otimes L^2\right) - \ii t 
	(1-\omega)\left( H \otimes \Id - \Id \otimes H \right)\right].
	\end{split}
	\end{equation}
	Now we note that in the case of the walk on a path, we have $L=H$
	and $[L\otimes L, L^m \otimes \Id] = 0$. Hence, the eigenvectors of
	$M^t_{\omega}$ are the same as the eigenvectors of $L \otimes L$. It is
	straightforward to check that
	\begin{equation}
	M_\omega^t=\sum_{i,j}\exp(-\omega \frac t2(\lambda_i-\lambda_j)^2) \exp( -\ii 
	t (1-\omega)(\lambda_i - \lambda_j)) 
	\ketbra{\lambda_i,\lambda_j}{\lambda_i,\lambda_j},
	\end{equation}
	where $\lambda_i$  and $\ket{\lambda_i}$ denote the eigenvalues and
	eigenvectors of $L$ and $H$. As $L$ is a tridiagonal Toeplitz matrix, its
	eigenvalues are given by~\cite{pasquini2006tridiagonal}
	\begin{equation}
	\lambda_j=2\cos\left (\frac{j\pi}{n+1}\right),
	\end{equation}
	where $1\leq j \leq n$. Furthermore the elements of the eigenvectors are
	\begin{equation}
	\braket{j}{\lambda_i}= \sqrt{\frac{2}{n+1}}\sin\left( 
	\frac{ij\pi}{n+1} \right) = \braket{i}{\lambda_j}.
	\end{equation}
	From this we get that the elements of $M_\omega^t$ in the computational basis 
	are
	\begin{equation}\small
	\begin{split}
	\bra{\gamma,\delta}M^t_{\omega}\ket{\kappa, \beta}&=\sum_{i,j=1}^n 
	\braket{i}{\lambda_{\kappa}}\braket{j}{\lambda_\beta}  
	\braket{i}{\lambda_\gamma} \braket{j}{\lambda_\delta} \times
	\\
	&\phantom{=\ }\times\exp(-\omega \frac{t}2(\lambda_i-\lambda_j)^2) \exp(\ii t 
	(1-\omega)(\lambda_i - \lambda_j)) \\
	&=\left(\frac{2}{n+1}\right)^2 \sum_{i,j=1}^n 
	\sin\left( \frac{\kappa i\pi}{n+1} \right)
	\sin\left( \frac{\beta j\pi}{n+1} \right)
	\sin\left( \frac{\gamma i\pi}{n+1} \right)
	\times \\
	&\phantom{=\ }\times \sin\left( \frac{\delta j\pi}{n+1} \right)\exp(-\omega \frac t2(\lambda_i-\lambda_j)^2) \exp(-\ii 
	t 
	(1-\omega)(\lambda_i - \lambda_j)).
	\end{split}
	\end{equation}
	Putting $\kappa=\beta=k$ and $\gamma=\delta=l$ we recover the desired result.
\end{proof}

\subsection{Probability distribution for infinite path} \label{sec:probability-infinite-path}
\begin{proof}[Proof of Theorem~\ref{theorem:prob-on-line}]
	In the case of a walk on a path $[-n, \ldots, n]$, the diagonal part of $\varrho(t)$ with initial state $\varrho(0)= \dyad 0$ satisfies
	\begin{equation}\small
	\begin{split}
	\bra{k}\varrho(t)\ket{k}&=\left(\frac{2}{2n+2}\right)^2 \sum\limits_{i,j=1}^{2n+1} 
	\sin\left(\frac{(k+n+1)i\pi}{2n+2}\right)
	\sin\left(\frac{(k+n+1)j\pi}{2n+2}\right)\times\\
	&\phantom{\ =}\times\sin\left(\frac{(n+1)i\pi}{2n+2}\right)
	\sin\left(\frac{(n+1)j\pi}{2n+2}\right)\times\\
	&\phantom{\ =}\times 
	\exp\left[-\frac 
	t2\omega(\lambda_i - \lambda_j)^2\right]
	\exp\left[ -\ii t (1-\omega) (\lambda_i-\lambda_j) 
	\right]\\
	&=\frac{1}{(n+1)^2} \sum\limits_{i,j=1}^{2n+1}
	\sin\left(\frac{ki\pi}{2n+2} + \frac{i\pi}{2}\right)
	\sin\left(\frac{kj\pi}{2n+2} + \frac{j\pi}{2}\right)sin\left(\frac{i\pi}{2}\right)\times\\
	&\phantom{\ =}\times
	\sin\left(\frac{j\pi}{2}\right)
	\exp\left[-\frac 
	t2\omega(\lambda_i - \lambda_j)^2\right]
	\exp\left[ -\ii t (1-\omega) (\lambda_i-\lambda_j) 
	\right].\\
	\end{split}
	\end{equation}
	Note, that for even $i$ or even $j$, the elements under the sum are equal to 
	zero. We get
	\begin{equation}
	\begin{split}
	\bra{k}\varrho(t)\ket{k}&=\frac{1}{(n+1)^2} 
	\sum\limits_{i,j=1,3,\dots,2n+1}
	\cos\left (\frac{ki\pi}{2(n+1)}\right )
	\cos\left (\frac{kj\pi}{2(n+1)}\right )
	\times \\
	&\phantom{\ =}\times \exp\left[- 2\omega t\left(\sin\frac{\pi 
		i}{2(n+1)}-\sin\frac{\pi j}{2(n+1)}\right)^2\right]\times \\
	&\phantom{\ =}\times \exp\left[-2\ii (1-\omega) t\left(\sin\frac{\pi 
		i}{2(n+1)}-\sin\frac{\pi 
		j}{2(n+1)}\right)\right].
	\end{split}
	\end{equation}
	The formula above is $1/4$ of the Riemann sum of the function
	\begin{equation}
	\begin{split}
	f(x)=&\cos\left (\frac{k\pi x}{2}\right )\cos\left (\frac{k\pi 
		y}{2}\right )\exp\left[-2\omega t\left (\sin\frac{\pi 
		x}{2}-\sin\frac{\pi 
		y}{2}\right )^2\right]\times \\
	&\times \exp\left[-2\ii(1-\omega)t\left (\sin\frac{\pi 
		x}{2}-\sin\frac{\pi 
		y}{2}\right )\right]
	\end{split}
	\end{equation}
	over the square $[0,2]\times [0,2]$ when we divide the region into equal
	squares. Hence, taking the limit $n \to \infty$ we get
	\begin{equation}
	\begin{split}
	\bra{k}\varrho(t)\ket{k} = & \frac{1}{4}\int_{0}^2\int_{0}^2 	\cos\left 
	(\frac{k\pi 
		x}{2}\right 
	)\cos\left (\frac{k\pi 
		y}{2}\right )\times\\
	&\times \exp\left[-2\omega t\left(\sin\frac{\pi 
		x}{2}-\sin\frac{\pi 
		y}{2}\right )^2\right] \times \\
	&		\times\exp\left[-2\ii (1-\omega)t\left (\sin\frac{\pi 
		x}{2}-\sin\frac{\pi 
		y}{2}\right)\right]\dd x \dd y.
	\end{split}
	\end{equation}
	After substituting $u=\frac{x\pi}{2}$ and $v=\frac{y\pi}{2}$ we have
	\begin{equation}
	\begin{split}
	\bra{k}\varrho(t)\ket{k}=&\frac{1}{\pi^2}\int_{0}^\pi\int_{0}^\pi \cos(ku)
	\cos(kv) \exp\left[-2\omega  t(\cos u-\cos v)^2\right]\times\\ & \times
	\exp\left[-2\ii (1-\omega)t(\cos u-\cos v)\right] \dd u \dd v.
	\end{split}
	\end{equation}
	By symmetry with respect to $x=0$ and $y=0$ we obtain the result.
\end{proof}
\sectionmark{Scaling exponent of interpolated standard\ldots}
\section{Scaling exponent of interpolated standard GQSW on infinite path graph}\label{sec:prop-qsw-scaling-exp}
\sectionmark{Scaling exponent of interpolated standard\ldots}

\subsection{Case $\omega=1$ }

%%%%%%%%%%%%%%%%%%%%%%%%%% 
%%%% technical lemmas %%%%
%%%%%%%%%%%%%%%%%%%%%%%%%%
\begin{lemma}[\cite{gradshteyn2014table}] \label{lem:sumOfPowerBinomial}
	For arbitrary $\alpha\in \R$, $n,m\in \N$ such that $m\leq n$  we have
	\begin{equation}
	\sum_{k=0}^{n} (-1)^k (k-\alpha)^m \binom{n}{k} = \begin{cases}
	0, & m<n,\\
	(-1)^n n!, & m = n.
	\end{cases}
	\end{equation}
\end{lemma}

\begin{lemma}[\cite{gradshteyn2014table}]\label{lem:sumDoubleBinomial}
	For arbitrary $n,p\in\N$ such that $p\leq n$ we have
	\begin{equation}
	\sum_{k=0}^{n-p} \binom{n}{k}\binom{n}{p+k} = \binom{2n}{n-p}.
	\end{equation}
\end{lemma}

\begin{lemma}\label{lem:cosIntegral}
	For arbitrary $k,l\in \N$ we have
	\begin{multline}
	\int_{-\pi}^{\pi}\cos(kx) \left[\cos(x)\right]^l \dd x \\=\begin{cases}
	\frac{2\pi}{2^{l-k}} \binom{l-k}{\frac{l-k}{2}} \prod_{i=0}^{k-1} 
	\frac{l-i}{l+k-2i}, & l \geq k \textrm{ and }(l=k\mod 2),\\
	0, & \textrm{otherwise.}\\
	\end{cases}
	\end{multline}
\end{lemma}
\begin{proof}
	Using the formula \cite{gradshteyn2014table}
	\begin{equation}\small
	\int \left[\cos(x)\right]^l \cos (kx) \dd x=
	\frac{1}{l+k}\left[\left(\cos(x)\right)^l 
	\sin(kx)+l\int\left[\cos(x)\right]^{l-1}
	\cos((k-1)x)\dd x\right]
	\end{equation}
	we obtain  
	\begin{equation}\label{eq:integralRecurence}
	\int_{-\pi}^\pi \left[\cos(x)\right]^l \cos (kx)\dd x = 
	\frac{l}{l+k}\int_{-\pi}^\pi \left[\cos(x) \right]^{l-1}
	\cos((k-1)x)\dd x.
	\end{equation}
	Moreover for arbitrary $l\in \N$ we have \cite{gradshteyn2014table}
	\begin{equation}
	\int \left[\cos (x) \right]^{2l} \dd x = \frac{1}{4^l}\binom{2l}{l}x +
	\frac{2}{4^l} \sum_{k=0} ^{l-1} \binom{2l}{k}\frac{\sin ((2l-2k)x)}{2l-2k},
	\end{equation}
	which provides us the formula
	\begin{equation}\label{eq:cosPowerIntegral}
	\int_{-\pi}^\pi \left[ \cos (x)\right]^{2l} \dd x = 
	\frac{2\pi}{4^l}\binom{2l}{l}.
	\end{equation}
	
	Suppose $l<k$. Then using Eq.~\eqref{eq:integralRecurence} we have 
	\begin{equation}
	\int_{-\pi}^\pi \left[\cos(x)\right]^l \cos (kx)\dd x = \prod_{i=0}^{l-1} 
	\frac{l-i}{l+k-2i} 
	\int_{-\pi}^{\pi} \cos((k-l)x) \dd x = 0.
	\end{equation}
	If $l\geq k$, then we obtain
	\begin{equation}
	\int_{-\pi}^\pi \left[ \cos(x)\right]^l \cos(kx) = \prod_{i=0}^{k-1} 
	\frac{l-i}{l+k-2i} 
	\int_{-\pi}^{\pi} \cos^{l-k} (x) \dd x.
	\end{equation}
	If $l - k$ is odd, then the integral equals 0. Otherwise using 
	Eq.~\eqref{eq:cosPowerIntegral} we have
	\begin{equation}
	\int_{-\pi}^\pi\cos(kx)\cos^l(x) = 
	\frac{2\pi}{2^{l-k}}\binom{l-k}{\frac{l-k}{2}}\prod_{i=0}^{k-1} 
	\frac{l-i}{l+k-2i}.
	\end{equation}
\end{proof}

%%%%%%%%%%%%%%%%%%%%%%%%%%%% 
%%%% Proposition Taylor %%%%
%%%%%%%%%%%%%%%%%%%%%%%%%%%%

\begin{proposition}\label{theorem:probabilityTaylorSeries}
	For an interpolated standard GQSW on an infinite path with an initial state
	$\varrho(0) = \dyad{0}$ and $\omega =1$, the diagonal part of $\varrho(t)$ is given by
	\begin{equation}
	\bra k \varrho(t)\ket k =  \sum_{n=|k|}^\infty 
	\frac{(-1)^{n+k}}{2^n}\binom{2n}{n}\binom{2n}{n+k}\frac{t^n}{n!}.
	\end{equation}
\end{proposition}
\begin{proof}
	Since the elements $\varrho_{kk}(t)\coloneqq \bra k \varrho(t)\ket k $ are symmetric with respect 
	to $k=0$, we assume $k\geq0$. By Theorem~\ref{theorem:prob-on-line} we have
	\begin{equation}
	\varrho_{kk}(t)=\frac{1}{4\pi^2}\int_{-\pi}^\pi\int_{-\pi}^\pi 
	\cos(kx)\cos(ky)\exp\left[-2t (\cos(x)-\cos(y))^2\right]\dd x \dd y.
	\end{equation}
	Suppose we have the Taylor series representation $\varrho_{kk}(t) = 
	\sum_{n=0}^{\infty} 
	\frac{A_{n,k}}{n!}t^n$. Then $A_{n,k}$ is of the form
	\begin{equation}
	\begin{split}
	A_{n,k} &= \frac{(-1)^n2^n}{4\pi^2}\int_{-\pi}^\pi\int_{-\pi}^\pi 
	\cos(kx)\cos(ky)\left 
	(\cos(x)-\cos(y)\right)^{2n}\dd x \dd y \\
	&= \frac{(-1)^n2^n}{4\pi^2}\sum_{l=0}^{2n}  
	\binom{2n}{l}(-1)^l\int_{-\pi}^\pi\cos(kx) \left[\cos(x) \right]^l \dd 
	x \;\times \\
	&\phantom{\ =}\times \int_{-\pi}^\pi 
	\cos(ky) \left[\cos (y)\right]^{2n-l} \dd y.
	\end{split}
	\end{equation}
	Let us define for simplicity
	\begin{equation}
	A_{n,k,l} \coloneqq 
	\binom{2n}{l}(-1)^l\int_{-\pi}^\pi\cos(kx)\left[\cos(x)\right]^l \dd 
	x\int_{-\pi}^\pi 
	\cos(ky) \left[\cos(y)\right]^{2n-l} \dd y.
	\end{equation}
	By Lemma~\ref{lem:cosIntegral} we have that $A_{n,k,l}$ is non-zero when $k-l$ 
	is even and takes the form
	\begin{equation}
	A_{n,k,l}  = 
	\frac{(-1)^l 4\pi^2}{2^{2n-2k}} \binom{2n}{l} \binom{l-k}{\frac{l-k}{2}} 
	\binom{2n-l-k}{n-\frac{l+k}{2}} 
	\prod_{i=0}^{k-1} \frac{l-i}{l+k-2i} \frac{2n-l-i}{2n-l+k-2i}. 
	\end{equation}
	Furthermore from condition $2n \geq 2n-l \geq k$ for $A_{n,k,l}\neq 0$, for $n<k$ we have $A_{n,k}=0$. 
	
	Again it is straightforward to find
	\begin{equation}
	A_{n,k,k} = \frac{(-1)^k 4\pi^2}{4^{n}} \binom{2n}{n}\binom{n}{k}
	\end{equation}
	and 
	\begin{equation}
	\frac{A_{n,k,l+2}}{A_{n,k,l}} = \frac{(n - \frac{l-k}{2})(n - 
		\frac{l-k}{2})}{(\frac{l+k}{2}+1)(\frac{l-k}{2}+1)}.
	\end{equation}
	Note that we increment $l$ by two instead of one because of the assumption that
	$l-k$ is even. One can verify, that the $A_{n,k,l}$ is of the form
	\begin{equation}
	A_{n,k,l} = 
	\frac{(-1)^k4\pi^2}{4^n}\binom{2n}{n}\binom{n}{\frac{l+k}{2}} 
	\binom{n}{\frac{l-k}{2}}.
	\end{equation}
	Finally we have
	\begin{equation}
	\begin{split}
	A_{n,k} &=  \frac{(-1)^n2^n}{4\cdot\pi^2} 
	\sum_{l\in\{k,k+2,\dots,2n-k\}}A_{n,k,l} \\
	&= 
	\frac{(-1)^{n+k}}{2^n}\binom{2n}{n}\sum_{l\in\{k,k+2,\dots,2n-k\}} 
	\binom{n}{\frac{l+k}{2}}\binom{n}{\frac{l-k}{2}}\\
	&= \frac{(-1)^{n+k}}{2^n}\binom{2n}{n}\sum_{l=0}^{n-k}
	\binom{n}{l+k}\binom{n}{l}\\
	&=\frac{(-1)^{n+k}}{2^n}\binom{2n}{n} \binom{2n}{n+k},
	\end{split}
	\end{equation}
	where in the third line we change the indices range and in the last line we use 
	Lemma~\ref{lem:sumDoubleBinomial}.
\end{proof}

\begin{proposition}\label{theorem:stochastic-moments-w1}
	For an interpolated standard GQSW on an infinite path with an
	initial state $\varrho(0) = \dyad{0}$, $\omega =1$, the $m$-th central moment
	$\mu_m(t)$ is polynomial in $t$ for $m$ even, and zero otherwise. Moreover for
	even $m$ we have
	\begin{equation}
	\lim_{t\to\infty}\frac{\mu_m(t)}{t^\frac{m}{2}}= \frac{m!}{\left( \frac m2 
		\right)!2^\frac{m}{2}}\binom{m}{\frac{m}{2}}.
	\end{equation}
\end{proposition}
\begin{proof}
	Note that odd moments equals 0 by symmetry of the probability distribution. 
	Suppose $m$ is even and $m>0$. Then by 
	Proposition~\ref{theorem:probabilityTaylorSeries} 
	we have
	\begin{equation}
	\begin{split}
	\mu_m(t) &= \sum_{k=-\infty}^{\infty} k^m  \sum_{n=|k|}^\infty 
	\frac{(-1)^{n+k}}{2^n}\binom{2n}{n}\binom{2n}{n+k}\frac{t^n}{n!}\\
	&= 
	\sum_{n=0}^{\infty}\frac{(-1)^n}{2^n}\binom{2n}{n}\frac{t^n}{n!}\sum_{k=-n}^{n}k^m
	(-1)^k\binom{2n}{n+k}\\
	&=\sum_{n=0}^{\infty}\frac{1}{2^n}\binom{2n}{n}\frac{t^n}{n!} 
	\sum_{k=0}^{2n}(k-n)^m   (-1)^{k}\binom{2n}{k}.
	\end{split}
	\end{equation}
	By Lemma~\ref{lem:sumOfPowerBinomial} formula above can be simplified
	\begin{equation}
	\mu_m(t) = 
	\sum_{n=1}^{\frac{m}{2}}\frac{1}{2^n}\binom{2n}{n}\frac{t^n}{n!}  
	\sum_{k=0}^{2n}(k-n)^m   (-1)^{k}\binom{2n}{k},
	\end{equation}
	hence the $m$-th central moment is a polynomial of degree $\frac{m}{2}$ with 
	respect to $t$. Moreover, the coefficient next to $t^\frac{m}{2}$ is
	\begin{equation}
	a_{\frac{m}{2}} = \frac{m!}{\left( \frac{m}{2} \right)! 
		2^\frac{m}{2}}\binom{m}{\frac{m}{2}}. 
	\end{equation}
\end{proof}

\subsection{Case $\omega<1$ }
\begin{proposition}\label{th:probabilityTaylorSeries-w01}
	For an interpolated standard GQSW on an infinite path with an initial state
	$\varrho(0) = \dyad{0}$ and $\omega \in(0,1)$, the diagonal part of
	$\varrho(t)$ is given by
	\begin{equation}
	\bra k \varrho(t)\ket k =  \sum_{n=|k|}^\infty 
	B_{n,k}\frac{t^n}{n!},
	\end{equation}
	where
	\begin{equation}
	B_{n,k}  = \frac{(-1)^{n+k}}{2^n}\sum_{l=0}^{\min 
		(\lfloor\frac{n}{2}\rfloor,n-|k|)} 
	\binom{n}{2l}
	4^{l}\omega^{n-2l}(1-\omega)^{2l}  
	(-1)^{l}
	\binom{2n-2l}{n-l}  
	\binom{2n-2l}{n-l+k}.
	\end{equation}
\end{proposition}
\begin{proof}
	If we denote
	$\varrho_{kk}(t)\coloneqq \bra k \varrho(t)\ket k=\sum_{n=0}^{\infty}\frac{B_{n,k}}{n!}t^n$, then one can find that
	$B_{n,k}$ is of the form
	\begin{equation}
	\begin{split}
	B_{n,k} &=  \sum_{l=0}^{n } 
	\binom{n}{l}\frac{1}{4\pi^2}\int_{-\pi}^\pi\int_{-\pi}^\pi 
	\cos(kx)\cos(ky)(-\omega)^{n-l}2^{n-l}\times  \\
	&\phantom{\ =}\times (\cos(x)-\cos(y))^{2n-l} 2^l  \ii^l (1-\omega)^l \dd 
	x\dd y .
	\end{split}
	\end{equation}
	Since $\varrho_{kk}(t) \in \R$, we can exclude the imaginary terms and we can
	simplify the formula
	\begin{equation}
	\begin{split}
	B_{n,k} &=\sum_{l=0}^{\lfloor \frac{n}{2}\rfloor} \binom{n}{2l}
	2^{2n-l}\omega^{n-2l}(1-\omega)^{2l} 
	\times \\
	&\phantom{\ =}\times \frac{(-1)^{n-l}}{2^{n-l} 
		4\pi^2}\int_{-\pi}^\pi\int_{-\pi}^\pi 
	\cos(kx)\cos(ky)(\cos(x)-\cos(y))^{2n-2l} \dd  x\dd y \\
	&=\sum_{l=0}^{\lfloor \frac{n}{2}\rfloor} \binom{n}{2l} 
	2^{2n-l}\omega^{n-2l}(1-\omega)^{2l}
	A_{n-l,k}.
	\end{split}
	\end{equation}
	
	From the proof of Proposition~\ref{theorem:probabilityTaylorSeries} we know, that $A_{n,k}$ takes the 
	form
	\begin{equation}
	A_{n,k} = \begin{cases}
	0, & |k|>n,\\
	\frac{(-1)^{n+k}}{8^n} \binom{2n}{n}  \binom{2n}{n+k}, & |k|\leq n.
	\end{cases}
	\end{equation}
	In our case we have the condition $|k|\leq n-\frac{l}{2}\leq n$. Hence we
	conclude, that $B_{n,k}$ is of the form
	\begin{equation}
	B_{n,k}  = \frac{(-1)^{n+k}}{2^n}\sum_{l=0}^{\min 
		(\lfloor\frac{n}{2}\rfloor,n-|k|)} 
	\binom{n}{2l}
	4^{l}\omega^{n-2l}(1-\omega)^{2l}  
	(-1)^{l}
	\binom{2n-2l}{n-l}  
	\binom{2n-2l}{n-l+k}.
	\end{equation}
	
\end{proof}

\begin{proposition}\label{theorem:stochastic-moments-w01}
	For a inteprolated standard GQSW on an infinite path with an
	initial state $\varrho(0) = \dyad{0}$, $\omega \in(0,1)$, the $m$-th central moment
	$\mu_m(t)$ is polynomial in $t$ for $m$ even, and zero otherwise. Moreover for
	even $m$ we have
	\begin{equation}
	\lim_{t\to\infty}\frac{\mu_m(t)}{t^m}= \binom{m}{\frac{m}{2}} (1-\omega)^{m}.
	\end{equation}
\end{proposition}
\begin{proof}	
	Thanks to the Proposition~\ref{th:probabilityTaylorSeries-w01}, for even $m$ we have
	\begin{equation}\small
	\begin{split}
	\mu_m(t)&= \sum_{k=-\infty}^{\infty} k^m  \sum_{n=0}^\infty 
	B_{n,k}\frac{t^n}{n!}\\
	&=\sum_{n=0}^{\infty}\frac{t^n}{n!}\sum_{k=-n}^{n}k^mB_{n,k}\\
	&=\sum_{n=0}^{\infty}\frac{(-1)^nt^n}{2^nn!}\sum_{k=-n}^{n}k^m 
	(-1)^k\times\\
	&\phantom{=\ }\times\sum_{l=0}^{\min 
		(\lfloor\frac{n}{2}\rfloor,n-|k|)} 
	\binom{n}{2l}
	4^{l}\omega^{n-2l}(1-\omega)^{2l}  
	(-1)^{l}
	\binom{2n-2l}{n-l}  
	\binom{2n-2l}{n-l+k}\\
	&=\sum_{n=0}^\infty \sum_{k=-n}^{n}\sum_{l=0}^{\min 
		(\lfloor\frac{n}{2}\rfloor,n-|k|)} C_{n,l,t} \omega^{n-2l}(1-\omega)^{2l}(-1)^k 
	k^m  \binom{2n-2l}{n-l+k},\\
	&=\sum_{n=0}^\infty \sum_{l=0}^{\lfloor\frac{n}{2}\rfloor} 
	C_{n,l,t} 
	\omega^{n-2l}(1-\omega)^{2l} \sum_{k=-(n-l)}^{n-l} (-1)^k 
	k^m  \binom{2n-2l}{n-l+k},\\
	&=\sum_{n=0}^\infty \sum_{l=0}^{\lfloor\frac{n}{2}\rfloor} 
	C_{n,l,t} 
	\omega^{n-2l}(1-\omega)^{2l} (-1)^{n-l}\sum_{k=0}^{2n-2l} (-1)^k 
	(k-(n-l))^m  \binom{2n-2l}{k},\label{eq:mth-moment-mixed}\\
	\end{split}
	\end{equation}
	where 
	\begin{equation}
	C_{n,l,t}= \frac{(-1)^nt^n}{2^nn!}  4^{l}(-1)^l 
	\binom{n}{2l}\binom{2n-2l}{n-l} .
	\end{equation}
	Let us denote
	\begin{equation}
	\alpha_{2n-2l,m} = \sum_{k=0}^{2n-2l} (-1)^k 
	(k-(n-l))^m  \binom{2n-2l}{k}.
	\end{equation} 
	From Lemma~\ref{lem:sumOfPowerBinomial} $\alpha_{2n-2l,m}$ is nonzero if 
	$m\geq2n-2l\iff l\geq n-\frac{m}{2}$. Hence Eq.~\eqref{eq:mth-moment-mixed} 
	can be simplified
	\begin{equation}
	\begin{split}
	\mu_m(t) = \sum_{n=0}^\infty \sum_{l=n-\frac{m}{2}}^{\lfloor\frac{n}{2}\rfloor} 
	C_{n,l,t} (-1)^{n-l}
	\omega^{n-2l}(1-\omega)^{2l}\alpha_{2n-2l,m}
	\end{split}
	\end{equation}
	The condition $n - \frac m2 \leq l\leq\lfloor\frac{n}{2}\rfloor$ implies that for $n>m$ we have necessarily zero elements in Taylor sequence. Hence, we have
	\begin{equation}
	\begin{split}
	\mu_m(t) &= \sum_{n=0}^m \sum_{l=n-\frac{m}{2}}^{\lfloor\frac{n}{2}\rfloor} 
	C_{n,l,t} (-1)^{n-l}
	\omega^{n-2l}(1-\omega)^{2l}\alpha_{2n-2l,m}\\
	&= \sum_{n=0}^m \sum_{l=n-\frac{m}{2}}^{\lfloor\frac{n}{2}\rfloor}  
	\frac{(-1)^nt^n}{2^nn!}  4^{l}(-1)^n 
	\binom{n}{2l}\binom{2n-2l}{n-l} 
	\omega^{n-2l}(1-\omega)^{2l}\alpha_{2n-2l,m}\\
	&= \sum_{n=0}^{m}\beta_{n,\omega}t^n.\label{eq:mixed-moment-general}
	\end{split}
	\end{equation}
	Let us calculate the leading term, $\beta_{m,\omega}$. Then we have 
	$l\in\{\frac{m}{2}\}
	$, $n=m$ with even $m$.
	\begin{equation}
	\begin{split}
	\beta_{m,\omega} &= \sum_{l=m-\frac{m}{2}}^{\lfloor\frac{m}{2}\rfloor} 
	\frac{(-1)^m}{2^mm!}  4^{l}(-1)^m 
	\binom{m}{2l}\binom{2m-2l}{m-l} \omega^{m-2l}(1-\omega)^{2l}\alpha_{2m-2l,m}\\
	&= \frac{1}{2^mm!}  2^{m}
	\binom{m}{m}\binom{m}{\frac{m}{2}}(1-\omega)^{m} \alpha_{m,m}\\
	&= \frac{1}{m!} \binom{m}{\frac{m}{2}} (1-\omega)^{m} 
	(-1)^mm! =  \binom{m}{\frac{m}{2}} (1-\omega)^{m},
	\end{split}
	\end{equation}
	where we used the fact, that $\alpha_{m,m} =(-1)^mm!$ by 
	Lemma~\ref{lem:sumOfPowerBinomial}.
\end{proof}
%\todo[inline]{Remove below?}
%Note $\beta_{0,\omega}=0$, since we start with probability distribution localized at position 0. Furthermore we have
%for $m=2$ and $n=1$
%\begin{equation}
%\begin{split}
%\beta_{1,\omega} &= \sum_{l=0}^{0}  
%\frac{-1}{2}  4^{l}(-1) 
%\binom{1}{2l}\binom{2-2l}{1-l} 
%\omega^{1-2l}(1-\omega)^{2l}\alpha_{2-2l,2} \\
%&=2\omega 
%\end{split}
%\end{equation}

% !TeX spellcheck = en_GB
\chapter{Proofs for quantum search}

\section{Proofs for \ER graphs}\label{app:er_proof}

\newtheorem*{proposition*}{Proposition}
\newtheorem*{theorem*}{Theorem}

\subsection{Convergence of the principal eigenvector of adjacency matrix} \label{app:er_proof_principal_eigenvector}
\begin{proposition*}[\cite{glos2018vertices}] 
	Let $\ket{\lambda_{1}}$ be a principal eigenvector of adjacency matrix of random \ER graph with parameter $p$. For the probability $p = \omega \left( \log^3(n)/(n\log^2\log n) \right)$and some constant $c>0$ we have
	\begin{equation}
	\Vert \ket{\lambda_1} -\ket{s} \Vert_{\infty} \leq c \frac{1}{\sqrt{n}}\frac{\ln^{3/2}(n)}{\sqrt{np} \ln(np)}
	\end{equation}
	with probability $1-o(1)$.
\end{proposition*}

\begin{proof}
	Note $\deg(v)$ follows a binomial distribution. Using
	Lindenberg's CLT and the fact that the convergence is uniform one can show that
	\begin{equation}
	\begin{split}
	%\begin{split}
	P\left( |\deg(v)-np| \leq 2\sqrt{\ln(n)np(1-p)} \right) &
	\approx P\left(|\mathcal X| \leq 2\sqrt{\ln(n)}\right) \\
	&\geq 1- \frac{1}{\sqrt{2\pi\ln(n)} n^2},
	%\end{split}
	\end{split}
	\end{equation}
	where $\mathcal X$ is a random variable with standard normal distribution. Let
	$A=\lambda_1\ketbra{\lambda_1}{\lambda_1}+\sum_{i\geq 2}
	\lambda_i\ketbra{\lambda_i}{\lambda_i}$ and $\ket{s} = \alpha
	\ket{\lambda_1}+\beta \ket{\lambda_1^\perp}$. Assume that
	$\ket{\lambda_1}$, $\ket{\lambda_1^\perp}$, $\ket{\lambda_i}$ are normed vectors and
	$\ket{\lambda_1^\perp}=\sum_{i \geq 2} \gamma_i \ket{\lambda_i}$. By the
	Perron-Frobenius Theorem we can choose a vector $\ket{\lambda_1}$ such that
	$\braket{v}{\lambda_1}\geq 0$ for all $v$ and hence obtain $\braket{s}{\lambda_1}=\alpha>0$.
	Thus %for all $i=1, \ldots ,n$
	\begin{equation}
	\begin{split}
	\left(A-E(A)\right)\ket{\lambda_1}
	&=\left(\lambda_1\ketbra{\lambda_1}{\lambda_1}+\sum_{i\geq
		2} 
	\lambda_i\ketbra{\lambda_i}{\lambda_i}-np\ketbra{s}{s}\right)\ket{\lambda_1}\\ 
	&= (\lambda_1 -np\alpha^2)\ket{\lambda_1}-np\alpha \beta 
	\ket{\lambda_1^{\perp}}.
	\end{split}
	\end{equation}
	With probability $1-o(1)$, using Theorem~\ref{lemma:chung_norm} we have
	\begin{equation}
	\begin{split}
	(\lambda_1 -np\alpha ^2)^2+(np)^2\alpha ^2\beta ^2&=\Vert \left(A-E(A)\right)\ket{\lambda_1} \Vert^2 \leq 8np \ln(n),
	\end{split}
	\end{equation}
	and by $\beta^2=1-\alpha^2$
	\begin{equation}
	\alpha^2np(np-2 \lambda_1)+\lambda_1^2 \leq 8np \ln(n).
	\end{equation}
	Eventually, we receive 
	\begin{equation}
	\begin{split}
	1 \geq \alpha \geq \alpha^2 &\geq \frac{\lambda_1^2-8np \ln(n)}{2 \lambda_1 np-(np)^2}  \geq 1-\frac{4}{2+\sqrt{\frac{np}{8\ln(\sqrt{2}n)}}} \geq 1-\frac{16}{\sqrt{\frac{np}{\ln(n)}}},
	\label{eq:szac.a}
	\end{split}
	\end{equation}
	where the fourth inequality comes from Theorem~\ref{theorem:chung_eigs}. We know that
	$|\deg(v)-np| \leq 2 \sqrt{n \ln(n) p(1-p)}$ with probability greater than
	$1-\frac{1}{n^2}$. Thus, with probability $1-\frac{1}{n}$, the above is true for
	all $v \in V$ simultaneously. Now, since $\deg(v)= \bra{v} A \ket{\mathbf{1}}$, 
	we have
	\begin{equation}
	\begin{split}
	\frac{np - 2 \sqrt{n \ln(n) p(1-p)}}{\lambda_1} &
	\leq \frac{1}{\lambda_1}\bra{v}\texttt{}A\ket{\mathbf{1}} \leq  \frac{np + 2 \sqrt{n \ln(n) p(1-p)}}{\lambda_1}.
	\end{split}
	\end{equation}
	The lower bound can be estimated as
	\begin{equation}
	\begin{split}
	\frac{np - 2 \sqrt{n \ln(n) p(1-p)}}{\lambda_1} 
	&\geq  \frac{1-2\sqrt{\ln(n)\frac{1-p}{np}}}{1+\sqrt{8\frac{\ln(\sqrt{2}n)}{np}}} \geq 
	\frac{1-2\sqrt{\frac{\ln(n)}{np}}}{1+4\sqrt{\frac{\ln(n)}{np}}} \eqqcolon d.
	\end{split}
	\end{equation}
	Where we use abound on $\lambda_1$ from Theorem~\ref{theorem:chung_eigs}. Similarly the upper bound
	\begin{equation}
	\frac{np +  2 \sqrt{n \ln(n) p(1-p)}}{\lambda_1} \leq \frac{1+2\sqrt{\frac{\ln(n)}{np}}}{1-4\sqrt{\frac{\ln(n)}{np}}} \eqqcolon u.
	\end{equation}
	Consequently 
	\begin{equation}
	\frac{d}{\sqrt n}\leq \frac{1}{\lambda_1}\bra{v} A\ket{s} \leq \frac{u}{\sqrt n} \label{eq:13}
	\end{equation}
	for all $v \in V$. Let $l=c\frac{\ln(n)}{\ln(\sqrt{\frac{np}{\ln(n)}}/4)}$, where $c=c(n,p) \in [1, 2)$ is chosen to satisfy $l=\left\lceil  \frac{\ln(n)}{\ln(\sqrt{\frac{np}{\ln(n)}}/4)} \right\rceil$. Hence
	\begin{equation}
	\frac{d^l}{\sqrt n}\leq \bra{v} \left(\frac{A}{\lambda}  \right)^l\ket{s}\leq \frac{u^l}{\sqrt n}
	\end{equation}
	for all $v \in V$.
	On the other hand
	\begin{equation}
	\begin{split}
	\left(\frac{1}{\lambda_1}A \right)^l \left(\alpha \ket{\lambda_1}+\beta \ket{\lambda_1^\perp}\right)&=\left(\ketbra{\lambda_1}{\lambda_1}+\sum_{i\geq 2} \left(\frac{\lambda_i}{\lambda_1}\right)^l \ketbra{\lambda_i}{\lambda_i}\right)\times \\
	&\phantom{\ =}\times 
	\left(\alpha \ket{\lambda_1}+\beta \ket{\lambda_1^\perp}\right)\\ 
	&=\alpha\ket{\lambda_1}+\beta \sum_{i\geq 2} \left(\frac{\lambda_i}{\lambda_1}\right)^l \gamma_i \ket{\lambda_i}.
	\end{split}
	\label{eq:14}
	\end{equation}
	Using Theorems~\ref{theorem:chung_eigs} and \ref{lemma:chung_norm} we are able to estimate  $\frac{\lambda_i}{\lambda_1}$ by 
	\begin{equation}
	\begin{split}
	\frac{\lambda_i}{\lambda_1} &\leq
	\frac{\sqrt{8np\ln(\sqrt{2}n)}}{np-\sqrt{8np\ln(\sqrt{2}n)}} =\frac{1}{\sqrt{\frac{np}{8\ln(\sqrt{2}n)}}-1} \leq \frac{4}{\sqrt{\frac{np}{\ln(n)}}}.
	\end{split}
	\end{equation}
	Thus
	\begin{equation}
	\begin{split}
	\left\Vert \beta \sum_{i\geq 2} \left(\frac{\lambda_i}{\lambda_1}\right)^l \gamma_i \ket{\lambda_i} \right\Vert_\infty 
	&\leq |\beta| \left\Vert  \sum_{i\geq 2} \left(\frac{\lambda_i}{\lambda_1}\right)^l \gamma_i \ket{\lambda_i} \right\Vert_2\\
	&\leq |\beta| \sqrt{ \sum_{i\geq 2} \gamma_i^2\left(\frac{4}{\sqrt{\frac{np}{\ln(n)}}}\right)^{2l}} \\&= \frac{|\beta|}{\left(\frac{\sqrt{\frac{np}{\ln(n)}}}{4}\right)^{l}}
	= \frac{|\beta |}{n^c}\\
	&\leq \frac{4}{\left(\frac{np}{\ln(n)}\right)^{1/4} n},
	\end{split}
	\label{eq:16}
	\end{equation}
	where the last inequality comes from Eq.~(\ref{eq:szac.a}) and \mbox{$\|\cdot\|_2$} denotes the Euclidean norm.
	By Eq.~\eqref{eq:13} and \eqref{eq:14} we get
	\begin{equation}
	\frac{d^l}{\sqrt n}\leq \alpha\braket{v}{\lambda_1}+ \bra{v} \left( \beta \sum_{i\geq 2} \left(\frac{\lambda_i}{\lambda_1}\right)^l \gamma_i \ket{\lambda_i} \right)  \leq\frac{u^l}{\sqrt n},
	\end{equation}
	for all $v \in V$ and using Eq.~\eqref{eq:szac.a} and \eqref{eq:16} we eventually obtain
	\begin{equation}
	\frac{\frac{d^l}{\sqrt{n}}-\frac{4}{\left(\frac{np}{\ln(n)}\right)^{1/4} n}}{1} \leq \braket{v}{\lambda_1}  \leq
	\frac{\frac{u^l}{\sqrt{n}}+\frac{4}{\left(\frac{np}{\ln(n)}\right)^{1/4} n}}{1-\frac{16}{\sqrt{\frac{np}{\ln(n)}}}}
	\end{equation}
	for all $v \in V$. In order to finish the proof it is necessary to show that
	\begin{equation}
	(1-d^l)+\frac{4}{\left(\frac{np}{\ln(n)}\right)^{1/4} \sqrt{n}} = \order{\frac{\log^{3/2}(n)}{\sqrt{np} \ln(np)}}
	\label{eq:21}
	\end{equation}
	and
	\begin{equation}
	(u^l-1)+\frac{4}{\left(\frac{np}{\ln(n)}\right)^{1/4} \sqrt{n}}= \order{\frac{\log^{3/2}(n)}{\sqrt{np} \ln(np)}}.
	\end{equation}
	We need to estimate how quickly  $d^l$ converges to $1$. Using the fact that $d
	\rightarrow 1$, it is enough to observe that
	\begin{equation}
	(1-d)l = \order{\frac{\ln^{3/2}(n)}{\sqrt{np} \log\left(  \sqrt{\frac{np}{\log(n)}}/4\right)} },
	\end{equation}
	and thus 
	\begin{equation}
	1-d^l \approx 1 - e^{(d-1)l} = \order{\frac{\ln^{3/2}(n)}{\sqrt{np} \ln(np)} }.
	\end{equation}
	The second term of LHS of Eq.~(\ref{eq:21}) converges to $0$ more rapidly than
	the bound, so it completes the proof for the lower bound. The same fact for the
	upper bound can be shown analogously.
\end{proof}

\subsection{Convergence of the largest eigenvalue of Laplacian} \label{app:er_proof_largest_eigenvalue}
\begin{theorem*}
	Let $G$ be a random graph chosen according to $\randgn[ER](p)$. Let $p_0>0$ be
	such that $np\geq p_0\log(n)$. Let $\delta_{\max}\sim cnp$ for some $c>0$ almost
	surely. Then almost surely $\lambda_1(L(G)) \sim cnp$ .
\end{theorem*} 
\begin{proof}
	Note, that since the eigenvector corresponding to 0 eigenvalue is the equal superposition, we have
	\begin{equation}
	\begin{split}
	\lambda_1 &= \max_{\{\ket{\phi}\bot\ket s: \braket{\phi}{\phi}=1\}} \bra \phi L \ket \phi \\
	&= \max_{\{\ket{\phi}\bot\ket s: \braket{\phi}{\phi}=1\}} (\bra \phi D \ket \phi-\bra \phi A \ket \phi).
	\end{split}
	\end{equation}
	Note that 
	\begin{equation}
	\begin{split}
	\lambda_{1}&\leq \max_{\{\ket{\phi}\bot\ket s: \braket{\phi}{\phi}=1\}} \bra \phi D \ket \phi +\max_{\{\ket{\phi}\bot\ket s: \braket{\phi}{\phi}=1\}} |\bra \phi A \ket \phi)|\\
	&\leq \delta_{\max} + C\sqrt{np}
	\end{split}
	\end{equation}
	by Theorem 2.5 from \cite{feige2005spectral}. Furthemore, we have 
	\begin{equation}
	\begin{split}
	\lambda_{1 } = \max_{\ket{\phi}} \bra \phi L \ket \phi \geq  \max_{i=1,\dots, n} \bra i L \ket i = \delta_{\max}.
	\end{split}
	\end{equation}
	Since $\sqrt{np} = o(np) $ we have $\lambda_1 \sim cnp$.
\end{proof}

\section{Proofs for \CL graphs}
\label{sec:cl-proofs}
In this section we assume $\omega_i= n^{a+\frac{i}{n} b}$.

\subsection{Complexity of $p$-norm of $\omega$} \label{sec:cl-proofs-pnorm}

\begin{theorem} 
	Let $\alpha,\beta>0$ be fixed numbers. Let $f_{a,b}(n) = \sum_{i=1}^n n^{a+bi/n}$. Then $f(n) = \frac{n^{1+\alpha+\beta}}{\beta\log(n)}(1+o(1))$.
\end{theorem}
\begin{proof}
	Let us consider the inner sum first
	\begin{equation}
	\sum_{i=1}^n n^{\alpha+\beta\frac{i}{n}} = n^\alpha \sum_{i=1}^n \left (n^{\beta/n}\right )^i = n^\alpha n^\frac{\beta}{n} \frac{n^\beta-1}{n^\frac{\beta}{n}-1} 
	\end{equation}
	Note that $n^\alpha n^{\beta/n} (n^\beta -1) = \Theta(n^{\alpha+\beta})$, hence we only need to derive the complexity of the denominator:
	\begin{equation}
	\begin{split}
	\frac{1}{1-n^{\beta/n}} &= \frac{1}{1- \exp\left ( \frac{\beta}{n} \log(n)\right )} = \frac{n}{\beta\log(n)}\frac{\frac{\beta}{n} \log(n)}{1- \exp\left ( \frac{\beta}{n} \log(n)\right )} \\
	&= \frac{n}{\beta\log(n)} (1+o(1)),
	\end{split}
	\end{equation}
	which ends the proof.
\end{proof}	
Note that in particular $\|\omega\|_1 = f_{a,b}(n) \sim \frac{n^{1+a+b}}{b\log(n)}$ and $\|\omega\|_2 = \sqrt{f_{2a,2b}(n)} \sim \frac{n^{\frac{1}{2}+a+b}}{\sqrt{2b\log(n)}}$. Since $b$ is a constant, we can discard it with $\Theta$ notation.

\subsection{Number of edges for \CL graphs} \label{app:cl-proofs-edges}
Let $\mathcal{E}$ be a random variable denoting the number of edges of random \CL graph with the proposed $\omega$. Let $p_{ij} = \frac{\omega_{i}\omega_j}{\|\omega\|_1}$. Then
\begin{equation}
\begin{split}
\EE\mathcal E &= \sum_{i,j} p_{ij} = \sum_{i,j} \frac{\omega_{i}\omega_j}{\|\omega\|_1} = \|\omega\|_1, \\
{\rm Var} [\mathcal E]&= \sum_{i,j}p_{ij}(1-p_{ij}) = \EE\mathcal E - \sum_{i,j}p_{ij}^2 \leq \EE\mathcal E.
\end{split}
\end{equation}
By the Chebyshev inequality
\begin{equation}
\begin{split}
\PP(|\mathcal{E}-\EE\mathcal{E}| &\geq \varepsilon \EE\mathcal{E}) \leq \frac{{\rm Var} [\mathcal E]}{\varepsilon^2(\EE\mathcal{E})^2}  \leq \frac{\EE\mathcal E }{\varepsilon^2(\EE\mathcal{E})^2} =\frac{1}{\varepsilon^2 \|\omega\|_1}.
\end{split}
\end{equation}
Let us take $\varepsilon = 1/\log n$. Then we have
\begin{equation}
\begin{split}
\PP(|\mathcal{E}-\EE\mathcal{E}| &\geq  \EE\mathcal{E} / \log n)  \leq \frac{1}{\varepsilon^2 \|\omega\|_1} \sim  \frac{\log^3(n)}{n^{1+a+b}} \leq \frac{1}{n^{1+a+b+\varepsilon}}.
\end{split}
\end{equation}
Hence the number of edges $\mathcal E$ concentrates around $\EE\mathcal{E}$. Note that the upper bound on the probability is $1/n^{1+\varepsilon}$ for $\varepsilon$, which, thanks to Borel-Cantelli lemma, means that the almost all graphs in a sequence will have this property.

\subsection{Degree convergence} \label{app:cl-proofs-degree}
Let $\mathcal{D}_i$ be a random variable denoting the degree of the $i$-th vertex of edges of random \CL graph with the proposed $\omega$. We do not assume $i$ is fixed. Let $p_{ij} = \frac{\omega_{i}\omega_j}{\|\omega\|_1}$. Then
\begin{equation}
\begin{split}
\EE\mathcal D_i &= \omega_i, \\
{\rm Var} [\mathcal D_i]&= \sum_{j}p_{ij}(1-p_{ij}) = \omega_i - \sum_{j}\frac{\omega_i^2\omega_j^2}{\|\omega\|_1^2} \leq \omega_i.
\end{split}
\end{equation}
By the Chebyshev inequality
\begin{equation}
\begin{split}
\PP(|\mathcal{D}_i-\omega_i| &\geq \varepsilon \omega_i) \leq \frac{\omega_i}{\varepsilon^2\omega_i^2}   =\frac{1}{\varepsilon^2 \omega_i}.
\end{split}
\end{equation}
Let us take $\varepsilon = 1/\log n$. Then we have
\begin{equation}
\begin{split}
\PP(|\mathcal{D}_i-\omega_i| &\geq  \omega_i / \log n)  \leq  \frac{\log^2n}{n^{a+\frac{i}{n}b}}. 
\end{split}
\end{equation}
Note that for any choice of $i$ and $a,b>0$ the probability converges to 0, but the series of probabilities is not converging for any choice of $i$. Hence we can say at best there is infinite subsequence of graphs s.t. the degree is close to the expected degree.

Let us use the Hoeffding theorem this time we have
\begin{equation}
\begin{split}
\PP(|\mathcal{D}_i-\omega_i| \geq  \varepsilon\omega_i ) &\leq 2 \exp \left( - \frac{2\varepsilon^2 \omega_i^2}{n}\right) = 2 \exp \left( - 2\varepsilon^2 n^{2a+2\frac{i}{n}b-1}\right) \\
&= \frac{2}{n^\frac{2\varepsilon^2 n^{2a+2\frac{i}{n}b-1}}{\log n}}
\end{split}
\end{equation}
Note that the series of probabilities is convergent if $\frac{2\varepsilon^2 n^{2a+2\frac{i}{n}b-1}}{\log n} \geq C > 1$ for some fixed $C$, which is can be relaxed by $\varepsilon=1/\log n$  to
\begin{equation}
\begin{split}
n^{2a+2\frac{i}{n}b-1} & \geq  \frac{C}{2} \log^3 n \\
2a+2\frac{i}{n}b-1 &> 0  \\
a+\frac{i}{n}b &> \frac{1}{2}  \\
\end{split}
\end{equation}
So if $i$ is chosen in such a way that $\omega_i > n^{\frac{1}{2}+\varepsilon}$, then almost all degrees concentrate around their expectation. Thus, for almost all graphs, $\mathcal{D}_i = \Theta(\omega_i)$.

\subsection{Convergence of $\ket{\lambda_1}$ for adjacency graphs} \label{app:cl-proofs-eigenvector}
\paragraph{The overlap} Let $A=\sum_i\lambda_i\ket{\lambda_i}$. Let $\ket{\omega} =
\sum_i\omega_i \ket i$. Note that $\EE A = \frac{1}{\|\omega\|_1} \ketbra{\omega}$
By \cite{chung2011spectra} we have a.a.s.
\begin{equation}
\| A- \EE A\| \leq \sqrt{8 \delta_{\max} \log n}.
\end{equation}
Let $\frac{1}{\|\omega\|_2}\ket{\omega}= \alpha \ket{\lambda_1} + \beta \ket{\lambda_1^\perp}$. Note that
\begin{equation}
\begin{split}
(A- \EE A) \ket{\lambda_1} &= \lambda_1 \ket{\lambda_1} - \frac{\|\omega\|_2}{\|\omega\|_1} \alpha \ket \omega \\
&= \left ( \lambda_1 - \frac{\|\omega_2^2}{\|\omega\|_1} \alpha^2\right ) \ket{\lambda_1} - \frac{\|\omega\|_2^2}{\|\omega\|_1} \alpha\beta \ket{\lambda_1^\perp}.
\end{split}
\end{equation}
By this we have $\|(A- \EE A) \ket{\lambda_1}\|_2^2 \geq \left ( \lambda_1 - \frac{\|\omega\|_2^2}{\|\omega\|_1} \alpha^2\right)^2$. Since $\|(A-\EE A)\ket{\lambda_1} \| \leq \|(A-\EE A)\|\|\ket{\lambda_1} \|=\|A-\EE A\|$, we have
\begin{equation}
\begin{split}
\left ( \lambda_1 - \frac{\|\omega\|_2^2}{\|\omega\|_1} \alpha^2\right ) ^2 &\leq  8 \delta_{\max} \log n \\ 
\alpha^2 &\geq \frac{\|\omega\|_1}{\|\omega\|_2^2}\left (\lambda_1- \sqrt{8 \delta_{\max} \log n } \right ) \\
&\geq \frac{\|\omega\|_1}{\|\omega\|_2^2}\left (\frac{\|\omega\|_2^2}{\|\omega\|_1} - 2\sqrt{8 \delta_{\max} \log n } \right ) \\
& \geq 1 - 2\frac{\|\omega\|_1}{\|\omega\|_2^2}\sqrt{8 \delta_{\max} \log n}.
\end{split}
\end{equation}
So as long as $\frac{\|\omega\|_1}{\|\omega\|_2^2}\sqrt{\delta_{\max} \ln n} =o(1)$ we have that $\alpha=1-o(1)$. 

This is satisfied for the proposed $\omega$ for any $a,b>0$.

\paragraph{Convergence for almost all nodes} We follow the proof similar to the one presented in \cite{kukulski2020comment}. 
Let $\ket{\lambda_1}=\sum_{i}\gamma_{i}\ket{i}$ and $\ket{\bar \omega} =  \sum_{i}\bar \omega_i\ket i$ where $\bar \omega_i=\frac{\omega_i}{\|\omega\|_2}$. We know that $\alpha\coloneqq \braket{\lambda_1}{\bar \omega}= 1-o(1)$. We will search for a indexes set $I_n\subseteq \{1,\dots,n\}$ such that $|I_n| = n(1-o(1))$ and
\begin{equation}
\max_{i\in I_n} |\gamma_{i,n}-1| = o(1).
\end{equation}

Let $\ket{\lambda_1} = \alpha\ket{\bar \omega} + \beta \ket{\bar \omega ^\perp}$, with $\ket{\bar \omega ^\perp}=\sum_{i} \bar \omega_{i}^\perp \ket{i}$ being a normed vector. Let 
\begin{equation}
I_\varepsilon(n)^c \coloneqq \{i\in \{1,\dots,n\} \colon |  \gamma_{i}/\bar \omega_{i}-\alpha| > \varepsilon\}
\end{equation} 
be the collection of indices for which values in vectors are not sufficiently close.
Since $\gamma_{i}/\bar \omega_{i} -\alpha =  \beta \bar \omega ^\perp_{i}/\bar \omega_{i}$, we have 
\begin{equation}
\begin{split}
1 \geq|\braket{\bar \omega ^\perp}|^2 &= \sum_{i=0}^{n-1} |\bar \omega ^\perp_i|^2 \geq \sum_{i\in I^c_\varepsilon(n)} |\bar \omega ^\perp_i|^2 > \left (\frac{\varepsilon\min_{i}\bar \omega_{i}}{\beta} \right )^2 |I^c_\varepsilon(n)|,
\end{split}
\end{equation}
hence $|I^c_\varepsilon(n)|<  \left (\frac{\beta}{\varepsilon\min_{i}\bar \omega _{i}} \right )^2$. We will expect $|I^c_\varepsilon(n)| = o(n)$, which gives us following condition on $a,b$:
\begin{equation}
\begin{split}
|I^c_\varepsilon(n)|&<  \left (\frac{\beta}{\varepsilon\min_{i}\bar \omega _{i}} \right )^2 \leq \frac{\frac{\|\omega\|_1}{\|\omega\|_2^2} \sqrt{8\delta_{\max} \log n}}{\varepsilon^2 \frac{n^{2a}}{\|\omega\|_2^2}} = \frac{\|\omega\|_1 \sqrt{8n^{a+b} \log n}}{\varepsilon^2 n^{2a}} \\
&\sim 2\sqrt 2 \frac{n^{1+a+b}}{b\log n} \frac{\sqrt{n^{a+b} \log n}}{\varepsilon^2 n^{2a}} =
\frac{2 \sqrt 2}{b\varepsilon^2}  \frac{n^{1-\frac{1}{2}a+\frac{3}{2}b} }{\sqrt{\log n}}
\end{split}
\end{equation}
We require $|I^c_\varepsilon(n)|=o(n)$, which translates to
\begin{equation}
\begin{split}
\frac{2 \sqrt 2}{b\varepsilon^2}  \frac{n^{1-\frac{1}{2}a+\frac{3}{2}b} }{\sqrt{\log n}} &= o(n)\\
\frac{1}{\varepsilon^2}  & = o\left(\frac{\sqrt {\log n}}{n^{\frac{1}{2}a-\frac{3}{2}b}} \right)\\
%\varepsilon^2  & = \omega \left(\frac{n^{\frac{1}{2}a-\frac{3}{2}b}}{\sqrt {\log n}} \right)\\
\varepsilon  & = \omega \left(\frac{n^{\frac{1}{4}a-\frac{3}{4}b}}{\sqrt {\log n}} \right).
\end{split}
\end{equation}
We also will require $\varepsilon = o(1)$. In order to satisfy both condition we will need $\frac{1}{4}a - \frac{3}{4}b <0$, which is equivalent to $a< 3b$.

Let $I_n \coloneqq \{1,\dots,n\} \setminus I^c_{\varepsilon'}(n)$ with $\varepsilon' = n^{\frac{1}{4}(a-3b)}$. Then  $\varepsilon' =o(1)$ and $I_n = n(1-o(1))$, and furthermore
\begin{equation}
\begin{split}
\max_{i\in I_n} | \gamma_{i}/{\bar \omega}_i-1| &\leq \max_{i\in I_n} |\gamma_{i}/{\bar \omega}_i-\alpha| + |1-\alpha| \leq \varepsilon' +|1-\alpha | = o(1).
\end{split}
\end{equation}

%Let us now show by example that $|I_n|=n(1-o(1))$ is tight, given $\alpha_{i,n} = \frac{1}{\sqrt n}$. Let $f(n)=o(n)$ and let $\ket{\varphi_n} = \frac{1}{\sqrt{n-f(n)}} \sum_{i=0}^{n-f(n)-1} \ket{i}$ satisfies the assumptions of the theorem, yet the maximal $|I_n|$ is of order $n-f(n)$.

\section{Classical search}\label{app:classical_search}

Let $P= AD^{-1}$ be a stochastic matrix of uniform walk on undirected graph. Let $P_j^\infty$ be its unique stationary state. Let $P_{ij}(t)$ be the probability of being at $j$ at time $t$ starting at node $i$. Finally let $R_{ij} \coloneqq \sum_{t=0}^{\infty} (P_{ij}(t) - P_j^\infty)$
Then we have
\begin{equation}
\begin{split}
R_{jj} - R_{ij} &= \sum_{t=0}^\infty (P_{jj}(t) - P_j^\infty- P_{ij}(t) + P_j^\infty) 
=\sum_{t=0}^\infty (P_{jj}(t) - P_{ij}(t)) \\
&=\sum_{t=0}^\infty (\bra j P^t \ket j  - \bra j  P^t \ket i) =\bra j \sum_{t=0}^\infty (P^t (\ket j  - \ket i)) .\\
\end{split}
\end{equation}
Note we cannot move $\ket j - \ket i$ outside the series, since $\sum_{t=0}^\infty P^t$ is not converging. Let $\langle T_{ij} \rangle$ be a mean first passage time from $i$ to $j$.
Then $\langle T_{ij} \rangle =  \frac{2|E|}{d_j}[R_{jj} - R_{ij}]$ for $i\neq j$ \cite{noh2004random} and $\langle T_{jj} \rangle = 0$. The mean first passage time starting at stationary state equals
\begin{equation}
\begin{split}
\langle T_j \rangle &\coloneqq \sum_{i=1}^n \frac{d_i}{2|E|}\langle T_{ij} \rangle =   \frac{1}{d_j}\sum_{\substack{i=1\\ i \neq j}}^nd_i(R_{jj} - R_{ij})\\
&= \frac{1}{d_j} \bra j \sum_{t=0}^\infty P^t \Big((2|E|-d_j)\ket j  - d_i\sum_{\substack{i=1\\ i \neq j}}^n\ket i\Big)  \\
&= \frac{1}{d_j} \bra j \sum_{t=0}^\infty \Big(P^t( 2|E|\ket j  - \ket{P^\infty})\Big)= \frac{1}{d_j}  \sum_{t=0}^\infty \Big(2|E| \bra j P^t \ket j  - \braket{j}{P^\infty}\Big)\\
&= \frac{1}{d_j}  \sum_{t=0}^\infty \Big(2|E| \bra j P^t \ket j  - d_j\Big)= \frac{2|E|}{d_j}  \sum_{t=0}^\infty \Big( \bra j P^t \ket j  - \frac{d_j}{2|E|}\Big)\\
&= \frac{2|E|}{d_j}  \sum_{t=0}^\infty \Big( \bra j P^t \ket j  - \frac{d_j}{2|E|}\Big).
\end{split}
\end{equation}
Let $\ket{\mu_i}$ be an eigenvector of normalized Laplacian with eigenvalue $\mu_i>0$. Based on the formula for $P_{jj}(t)$ before Eq.~(2.1) from \cite{sinclair2012algorithms} we have
\begin{equation}
\sum_{i=1}^n \frac{d_i}{|E|}\langle T_{ij} \rangle = \frac{2|E|}{d_j}  \sum_{t=0}^\infty \sum_{i\geq 2} \lambda_i^t \braket{j}{\mu_i}^2 = \frac{2|E|}{d_j}   \sum_{i\geq 2} \frac{1}{1-\lambda_i} \braket{j}{\mu_i}^2  =   \frac{2|E|}{d_j}S_1.
\end{equation}

Let us upper bound it from the above and from below
\begin{equation}
\langle T_{j} \rangle =  \frac{2|E|}{d_j}   \sum_{i\geq 2} \frac{1}{1-\lambda_i} \braket{j}{\mu_i}^2 \leq  \frac{2|E|}{d_j} \frac{1}{1-\lambda_2} \sum_{i\geq 2} \braket{j}{\mu_i}^2 \leq   \frac{2|E|}{d_j} \frac{1}{\Delta} .
\end{equation}
Note that $\varepsilon = \frac{d_j}{2|E|}$, which confirms one bound. Similarly for the other side we have.
\begin{equation}
\begin{split}
\langle T_{j} \rangle &= \frac{2|E|}{d_j}  \sum_{t=0}^\infty \sum_{i\geq 2} \lambda_i^t \braket{j}{\mu_i}^2 = \frac{2|E|}{d_j}   \sum_{i\geq 2} \frac{1}{1-\lambda_i} (e_j^{(i)})^2 \\
&>  \frac{|E|}{d_j}   \sum_{i\geq 2}  (e_j^{(i)})^2 =  \frac{|E|}{d_j}    (1- (e_j^{(1)})^2 )\\
&=  \frac{|E|}{d_j}    (1- \frac{d_j}{2|E|} )= \frac{|E|}{d_j} -\frac{1}{2}.
\end{split}
\end{equation}

\end{document}